\documentclass[reqno]{amsart}%
\usepackage{amsfonts}
\usepackage{amsmath}
\usepackage{amssymb}
\usepackage{graphicx}
\usepackage{hyperref}%
\newtheorem{theorem}{Theorem}
\theoremstyle{plain}

\newtheorem{definition}{Definition}
\newtheorem{example}{Example}

\newtheorem{lemma}{Lemma}
\newtheorem{notation}{Notation}

\newtheorem{remark}{Remark}

\numberwithin{equation}{section}
\ifx\pdfoutput\relax\let\pdfoutput=\undefined\fi
\newcount\msipdfoutput
\ifx\pdfoutput\undefined\else
\ifcase\pdfoutput\else
\ifx\paperwidth\undefined\else
\ifdim\paperheight=0pt\relax\else\pdfpageheight\paperheight\fi
\ifdim\paperwidth=0pt\relax\else\pdfpagewidth\paperwidth\fi
\fi\fi\fi
\begin{document}
\title[$p$-Adic Statistical Field Theory and DBMs]{A Correspondence Between Deep Boltzmann Machines and $p$-Adic Statistical
Field Theories}
\author[Z\'{u}\~{n}iga-Galindo]{W. A. Z\'{u}\~{n}iga-Galindo}
\address{University of Texas Rio Grande Valley\\
School of Mathematical \& Statistical Sciences\\
One West University Blvd\\
Brownsville, TX 78520, United States }
\email{wilson.zunigagalindo@utrgv.edu}
\thanks{The author was partially supported by the Lokenath Debnath Endowed Professorship.}

\begin{abstract}
There is a strong interest in studying the correspondence between Euclidean
quantum fields and neural networks. This correspondence takes different forms
depending on the type of networks considered. In this work, we study this
correspondence in the case of deep Boltzmann machines (DBMs) having a
tree-like topology. We use $p$-adic numbers to encode this type of topology. A
$p$-adic continuous DBM is a statistical field theory (SFT) defined by an
energy functional on the space of square-integrable functions defined on a
$p$-adic $N$-dimensional ball. The energy functionals are non-local, meaning
they depend on the interaction of all the neurons forming the network. Each
energy functional defines a probability measure. A natural discretization
process attaches to each probability measure a finite-dimensional Boltzmann
distribution, which describes a hierarchical DBM. We provide a mathematically
rigorous perturbative method for computing the correlation functions. A
relevant novelty is that the general correlation functions cannot be
calculated using the Wick-Isserlis theorem. We give a recursive formula for
computing the correlation functions of an arbitrary number of points using
certain $3$-partitions of the sets of indices attached to the points.

\end{abstract}
\maketitle
\tableofcontents

\section{Introduction}

This article explores the interplay between the physics of energy based models
in machine learning see, e.g., \cite{Humbelli et al}-\cite{Benigio}, the
correspondence between neural networks (NNs) and statistical field theories
(SFTs), see, e.g., \cite{Zuniga-He-Zambrano}-\cite{Buice2}, and the $p$-adic
(ultrametric) spin glasses, see, e.g., \cite{R-T-V}-\cite{DysonFreeman}.

We present a new class $\left\{  \boldsymbol{v},\boldsymbol{h}\right\}  ^{4}%
$-SFTs on $L_{\mathbb{R}}^{2}\left(  \mathbb{Z}_{p}\right)  \times
L_{\mathbb{R}}^{2}\left(  \mathbb{Z}_{p}\right)  $, where $\mathbb{Z}_{p}$ is
the ring of $p$-adic integers, $p$ is a fixed prime number, and $L_{\mathbb{R}%
}^{2}\left(  \mathbb{Z}_{p}\right)  $ is the $\mathbb{R}$-vector space of
square-integrable functions for the Haar measure $dx$ on $\mathbb{Z}_{p}$. The
elements of $\mathbb{Z}_{p}$ are organized in an infinite rooted tree, with
valence $p$, with a neuron at each point. Here, for the sake of simplicity, we
discuss the results in dimension one. A SFT corresponds to a probability
measure%
\begin{equation}
\mathbb{P}(\boldsymbol{v},\boldsymbol{h};\boldsymbol{\theta})=\frac
{\exp\left(  -E(\boldsymbol{v},\boldsymbol{h};\boldsymbol{\theta})\right)
}{\mathcal{Z}^{\left(  2\right)  }\left(  \boldsymbol{\theta}\right)
}\mathbb{P}_{K_{1}}\left(  \boldsymbol{v}\right)  \otimes\mathbb{P}_{K_{2}%
}\left(  \boldsymbol{h}\right)  , \label{Eq-SFT}%
\end{equation}
where the fields $\boldsymbol{v},\boldsymbol{h}:\mathbb{Z}_{p}\rightarrow
\mathbb{R}$ satisfy $\left\Vert \boldsymbol{v}\right\Vert _{2}$, $\left\Vert
\boldsymbol{h}\right\Vert _{2}<M$, where $M$ is a fixed positive constant,
$\mathbb{P}_{K}\left(  \boldsymbol{v}\right)  \otimes\mathbb{P}_{K}\left(
\boldsymbol{h}\right)  $ is a Gaussian probability measure on a $\sigma
$-algebra of subsets of $L_{\mathbb{R}}^{2}\left(  \mathbb{Z}_{p}\right)
\times L_{\mathbb{R}}^{2}\left(  \mathbb{Z}_{p}\right)  $,
\begin{equation}
E(\boldsymbol{v},\boldsymbol{h};\boldsymbol{\theta})=-\left\langle
a,\boldsymbol{v}\right\rangle -\left\langle b,\boldsymbol{h}\right\rangle -%
{\displaystyle\iint\limits_{\mathbb{Z}_{p}\times\mathbb{Z}_{p}}}
\boldsymbol{h}\left(  x\right)  w\left(  x,y\right)  \boldsymbol{v}\left(
y\right)  dydx+\left\langle c,\boldsymbol{v}^{4}\right\rangle +\left\langle
d,\boldsymbol{h}^{4}\right\rangle \label{Eq-SFT-2}%
\end{equation}
is an energy functional, $\boldsymbol{\theta=}\left(  a,b,w,c,d\right)  $, and
$\mathcal{Z}^{\left(  2\right)  }\left(  \boldsymbol{\theta}\right)  $ is the
partition function. If $c=d=0$, $E(\boldsymbol{v},\boldsymbol{h}%
;\boldsymbol{\theta})$ is the energy function of a $p$-adic continuous spin
glass, where the product of the Boltzmann constant $K_{B}$ and the temperature
$T$ is fixed to be $1$. We identify a $\left\{  \boldsymbol{v},\boldsymbol{h}%
\right\}  ^{4}$-SFT with a $p$-adic continuous deep Boltzmann machine (DBM).
In (\ref{Eq-SFT-2}), kernel $w\left(  x,y\right)  $ controls the strength of
the interaction between the neurons at positions $x$ and $y$. The\ NN
corresponding to an\ energy functional of type (\ref{Eq-SFT-2}) is a
`small-world type network,' which means that any two neurons in the network
interact. It is widely accepted that the small-world property plays a central
role in the functioning of many biological networks; see, e.g., \cite{Sporns},
\cite{Muldoon et al}.

The visible field $\boldsymbol{v}$ and the hidden field $\boldsymbol{h}$ are
intended to model signals and data. We use real-valued functions defined on
$\mathbb{Z}_{p}$ to represent signals and data, which always take values from
a bounded set. For instance, the electrical and biological circuits always
produce bounded voltage signals. Thus, the SFTs used in machine learning have
a natural cutoff, `$M$'; this cutoff plays an important role in the
mathematical formulation of these theories. Notably\textit{,} this cutoff is
not a consequence of the fact that the $1$-dimensional unit ball is compact,
because\ the fields $\boldsymbol{v}$ and $\boldsymbol{h}$\ are not necessarily
continuous functions. We prefer $L_{\mathbb{R}}^{2}\left(  \mathbb{Z}%
_{p}\right)  $ over $L_{\mathbb{R}}^{2}\left(  \mathbb{Q}_{p}\right)  $
because the discretization of any $p$-adic continuous $\left\{  \boldsymbol{v}%
,\boldsymbol{h}\right\}  ^{4}$-SFT over $L_{\mathbb{R}}^{2}\left(
\mathbb{Z}_{p}\right)  $ gives a standard hierarchical Boltzmann machine.

There exists a natural discretization process, which corresponds to a
projection of the form%
\[
\Pi_{l}:L_{\mathbb{R}}^{2}\left(  \mathbb{Z}_{p}\right)  \times L_{\mathbb{R}%
}^{2}\left(  \mathbb{Z}_{p}\right)  \rightarrow\mathcal{D}^{l}(\mathbb{Z}%
_{p})\times\mathcal{D}^{l}(\mathbb{Z}_{p}),
\]
where $l$ is a positive integer, and $\mathcal{D}^{l}(\mathbb{Z}_{p}%
)\times\mathcal{D}^{l}(\mathbb{Z}_{p})$ is a finite-dimensional real vector
space. The push-forward measure of $\mathbb{P}(\boldsymbol{v},\boldsymbol{h}%
;\boldsymbol{\theta})$ to $\mathcal{D}^{l}(\mathbb{Z}_{p})\times
\mathcal{D}^{l}(\mathbb{Z}_{p})$ gives a Boltzmann probability distribution
$\mathbb{P}_{l}(\boldsymbol{v}_{l},\boldsymbol{h}_{l};\boldsymbol{\theta}%
_{l})$ defined using a discrete energy functional $E_{l}(\boldsymbol{v}%
_{l},\boldsymbol{h}_{l};\boldsymbol{\theta}_{l})$, which in turn is a natural
discretization of $E(\boldsymbol{v},\boldsymbol{h};\boldsymbol{\theta})$. This
energy functional is naturally attached to a generalized restricted Boltzmann
machine (RBM). For instance, the classical Gaussian-Gaussian RBMs constitute a
particular class of NNs obtained from the $\mathbb{P}_{l}(\boldsymbol{v}%
_{l},\boldsymbol{h}_{l};\boldsymbol{\theta}_{l})$, \cite{Decelle et
al}-\cite{Fischer-Igel}. Intuitively, the SFT corresponding to $\mathbb{P}%
(\boldsymbol{v},\boldsymbol{h};\boldsymbol{\theta})$ is the (thermodynamic)
limit of the discrete SFTs corresponding to $\mathbb{P}_{l}(\boldsymbol{v}%
_{l},\boldsymbol{h}_{l};\boldsymbol{\theta}_{l})$ when the number of neurons
tends to infinity, i.e., $\lim_{l\rightarrow\infty}\mathbb{P}_{l}%
(\boldsymbol{v}_{l},\boldsymbol{h}_{l};\boldsymbol{\theta}_{l})=\mathbb{P}%
(\boldsymbol{v},\boldsymbol{h};\boldsymbol{\theta})$ in some sense. The
rigorous study of this limit is quite tricky, and it depends strongly on how
the topologies of the networks attached to the $\mathbb{P}_{l}(\boldsymbol{v}%
_{l},\boldsymbol{h}_{l};\boldsymbol{\theta}_{l})$ scale as $l$ tends to
infinity. Here, we use a top-down approach; we construct a vast class of
probability measures $\mathbb{P}(\boldsymbol{v},\boldsymbol{h}%
;\boldsymbol{\theta})$, which admit a natural discretization process that
produces the Boltzmann distributions $\mathbb{P}_{l}(\boldsymbol{v}%
_{l},\boldsymbol{h}_{l};\boldsymbol{\theta}_{l})$ corresponding to a large
class of deep Boltzmann machines.

The construction presented here gives a correspondence between $p$-adic SFTs
and $p$-adic (ultrametric) Boltzmann machines (BM). Each of these BMs is a
continuous version of a discrete BM attached to $\mathbb{P}_{l}(\boldsymbol{v}%
_{l},\boldsymbol{h}_{l};\boldsymbol{\theta}_{l})$. Our main contribution is a
(mathematically rigorous) perturbative theory for computing the correlation
functions of the $\left\{  \boldsymbol{v},\boldsymbol{h}\right\}  ^{4}%
$-SFTs\ of type (\ref{Eq-SFT}) in arbitrary dimension.

The primary motivation for developing quantum field theories (QFTs) over a
completely disconnected space-time comes from Volovich's interpretation of
Bronstein's\ inequality. In the 1930s, Bronstein showed that general
relativity and quantum mechanics imply that the uncertainty $\Delta x$ of any
length measurement satisfies $\Delta x\geq L_{\text{Planck}}:=\sqrt
{\frac{\hbar G}{c^{3}}}$, where $L_{\text{Planck}}$ is the Planck length
($L_{\text{Planck}}\approx10^{-33}$ $cm$). This inequality implies that
space-time is not an infinitely divisible continuum (mathematically speaking,
space-time must be a completely disconnected topological space at the level of
the Planck scale). In the 1980s, Volovich proposed the conjecture that
space-time at the Planck scale has a $p$-adic nature \cite{Volovich1}. This
conjecture has propelled a wide variety of investigations in theoretical
physics, particularly in QFT, see, e.g., \cite{V-V-Z}-\cite{Zuniga-RMP-2022}.

The space $\mathbb{Q}_{p}$ has a very rich mathematical structure that allows
a rigorous mathematical formulation of the QFTs. The $p$-adic theories share
many standard features with the classical ones but have several crucial
differences. First, the $p$-adic numbers are not an ordered field; any QFT
with a $p$-adic time is acausal. Second, the $p$-adic space-time has a fractal
hierarchical structure. This feature is extremely useful in applications to
complex systems and neural networks. Third, in the $p$-adic framework, there
is a natural way for discretizing functions, where the discrete ones are
continuous. This property facilitates the construction of measures in QFT by
limiting processes.

Until recently, $p$-adic QFTs were considered mathematical toys without any
physical content, useful for understanding the problems of true QFTs. In
\cite{Zuniga2023}, we initiate the study of the correspondence between
$p$-adic statistical field theories (SFTs) and neural networks (NNs). This
work introduces a $p$-adic version of the convolutional deep Boltzmann
machines where only binary data is considered without implementation. By
adapting the mathematical techniques introduced by Le Roux and Benigio in
\cite{Le roux et al 1}, we show that these machines are universal
approximators for binary data tasks. In \cite{Zuniga-He-Zambrano}, the authors
discuss the implementation of a class of $p$-adic convolutional networks and
obtain desired results on a feature detection task based on hand-writing
images of decimal digits. The main novelty of the $p$-adic convolutional DBMs
is that they use significantly fewer parameters than the conventional ones.

Neural networks whose states are $p$-adic numbers were studied in
\ \cite{Albeverio-Khrennikov-Tirozzi}-\cite{Krennikov-tirozzi}. These models
are entirely different from the ones considered here. These ideas have been
used to develop non-Archimedean models of brain activity and mental processes
\cite{Khrenikov2A}. In \cite{ZZ1}-\cite{ZZ2}, $p$-adic versions of the
cellular neural networks were studied. These models involved abstract
evolution equations.

We now describe the specific results presented in this article and its
organization. In Section \ref{Section_2}, we quickly review the basic aspects
of $p$-adic analysis required. Section \ref{Section_3} is dedicated to the
construction of a large family of Gaussian probability measures $\mathbb{P}%
_{K}$ on $L_{\mathbb{R}}^{2}\left(  \mathbb{Z}_{p}^{N}\right)  $, where $K$ is
a general covariance kernel, see Theorem \ref{Theorem2}. In Section
\ref{Section_4}, we provide the precise definition of the energy functionals
$E(\boldsymbol{v},\boldsymbol{h};\boldsymbol{\theta})$, $\boldsymbol{\theta
=}\left(  a,b,w,c,d\right)  $, which requires the product measure
$\mathbb{P}_{K_{1}}\left(  \boldsymbol{v}\right)  \otimes\mathbb{P}_{K_{2}%
}\left(  \boldsymbol{h}\right)  $, where $K_{1}$, $K_{2}$ are covariance kernels.

In case $c=d=0$, the energy functional is denoted as $E^{\text{spin}%
}(\boldsymbol{v},\boldsymbol{h};\boldsymbol{\theta})$; for these functionals,
we give a general perturbative theory for the correlation functions when the
fields $\boldsymbol{v},\boldsymbol{h}$ belong to a ball centered at the origin
with radius $M$ in $L_{\mathbb{R}}^{2}\left(  \mathbb{Z}_{p}^{N}\right)  $.
When self-interactions are allowed, i.e., $c\neq0$ or $d\neq0$, \ our
perturbative theory requires that $a$, $b$ be constant functions but with an
arbitrary sign\ and that the fields $\boldsymbol{v},\boldsymbol{h}$ belong to
a ball centered at the origin with radius $M$ in $L_{\mathbb{R}}^{\infty
}\left(  \mathbb{Z}_{p}^{N}\right)  $. In machine learning applications, the
parameters $\boldsymbol{\theta=}\left(  a,b,w,c,d\right)  $ are adjusted using
a gradient-type algorithm; thus, assuming the parameters' sign is not
convenient. The fact that the sign of the coupling parameters $c$, $d$ is
arbitrary does not appear in the classical $\phi^{4}$-QFTs, see, e.g.,
\cite{Zuniga-RMP-2022}, \cite{Kleinert et al}.

By restricting the energy functionals $E(\boldsymbol{v},\boldsymbol{h}%
;\boldsymbol{\theta})$ to suitable finite-dimensional space, a discrete energy
functional $E_{l}\left(  \boldsymbol{v}_{l},\boldsymbol{h}_{l}%
;\boldsymbol{\theta}_{l}\right)  $ is obtained, see \cite{Zuniga-He-Zambrano},
\cite{Zuniga2023}, \cite{Zuniga-RMP-2022}. This discretization procedure is
reviewed in Section \ref{Section_5}. By using the discretization procedure,
intuitively, it is possible to discretize the measure $\mathbb{P}_{K_{1}%
}\otimes\mathbb{P}_{K_{2}}$ to obtain a finite-dimensional Boltzmann
probability distribution $\mathbb{P}_{l}(\boldsymbol{v}_{l},\boldsymbol{h}%
_{l};\boldsymbol{\theta}_{l})$, which corresponds to a $p$-adic discrete DBM.
In Section \ref{Section_6}, we rigorously show this fact, see Theorem
\ref{Theorem2}. Section \ref{Section_7} is dedicated to the rigorous
definition of correlation functions and the generating functionals. We show
that the correlation functions can be computed from the generating functionals
using functional derivatives like in the standard case. We compute explicitly
the generating functional $Z^{\text{spin}}(J_{1},J_{2})$ corresponding to
$E^{\text{spin}}(\boldsymbol{v},\boldsymbol{h};\boldsymbol{\theta})$ under a
general hypothesis about the Green function $G_{w}\left(  x,y\right)  $
attached to kernel $w\left(  x,y\right)  $, see Theorem \ref{Lemma6}.

Section \ref{Section_8} begins with an explicit formula for the correlation
functions attached to the SFT corresponding to $E^{\text{spin}}(\boldsymbol{v}%
,\boldsymbol{h};\boldsymbol{\theta})$, see Theorem \ref{Theorem2}. Then, we
give an explicit class of BMs and compute the correlation functions
$\boldsymbol{G}^{\left(  1+0\right)  }\left(  x\right)  $, $\boldsymbol{G}%
^{\left(  0+1\right)  }\left(  y\right)  $, $\boldsymbol{G}^{\left(
1+1\right)  }\left(  x,y\right)  $, with $x,y\in\mathbb{Z}_{p}^{N}$, which
are, respectively, the means of the visible and hidden fields, and the
correlation between the visible and hidden fields. We obtain that%
\[
\boldsymbol{G}^{\left(  1+0\right)  }\left(  x\right)  =A\left(  G_{w}\ast
b\right)  \left(  x\right)  ,\text{ \ }\boldsymbol{G}^{\left(  0+1\right)
}\left(  x\right)  =A\left(  \widetilde{G}_{w}\ast a\right)  \left(  x\right)
\text{,}%
\]
and%
\[
\text{ }\boldsymbol{G}^{\left(  1+1\right)  }\left(  x,y\right)  =\left(
\widetilde{G}_{w}\ast a\right)  \left(  x\right)  \text{ }\left(  G_{w}\ast
b\right)  \left(  y\right)  ,
\]
where $A$ is a positive constant, $G_{w}$ is the Green function,
$\widetilde{G}_{w}\left(  x\right)  =G_{w}\left(  -x\right)  $. The Green
function does not agree with the two-point correlation function like in the
standard and $p$-adic QFTs of one scalar field.

We also show that $\boldsymbol{G}^{\left(  1+1\right)  }\left(  x,y\right)  $
does not decay as a function of the distance $\left\Vert x-y\right\Vert _{p}$.
We argue that this fact results from the long-range interaction between the
neurons. An important novelty is that the general correlation functions cannot
be calculated using the Wick-Isserlis Theorem. The main result of Section
\ref{Section_8}, see Theorem \ref{Theorem3}, is a recursive formula for
computing the correlation functions $\boldsymbol{G}^{\left(  m+n\right)
}\left(  \left[  x_{i}\right]  _{1\leq i\leq m},\left[  y_{i}\right]  _{1\leq
i\leq n}\right)  $, using certain $3$-partitions of the sets of indices
attached to the points $x_{i}$ for $1\leq i\leq m$, and $y_{i}$ for $1\leq
i\leq n$.

In the last section, we present a perturbative technique for computing the
correlation functions $\boldsymbol{G}^{\left(  m+n\right)  }\left(
x,y;c,d\right)  $, for $c,d\in\mathbb{R}$, see Theorem \ref{Theorem4}.

The literature about statistical mechanics and machine learning is enormous;
see, e.g., \cite{Statisitical mechanics}-\cite{Statisitical mechanics2}, and
the references therein. To the best of our knowledge, our results are new and
complementary to the ones available in the literature. Most of the known
contributions study the statistical mechanics of the learning process. In
contrast, here we study SFTs coming from $p$-adic continuous spin glasses,
whose energy functionals have discretizations that define hierarchical
Boltzmann machines as used in practice, see, e.g., \cite{Zuniga-He-Zambrano},
\cite{Honglak et al}-\cite{Fischer-Igel}. Interestingly, about twenty years
ago, it was proposed to use ultrametric spin glasses in object recognition,
\cite{Capauto et al}.

\section{$p$-adic DBNs compare to Standard DBNs}

In this section, we compare the $p$-adic DBMs with the standard counterparts
and comment on the SFTs attached to these two types of networks. It is
relevant to mention here that this work is a continuation of \cite{Zuniga2023}%
-\cite{Zuniga-He-Zambrano}. Still, its goals and results differ entirely from
those presented in the publications mentioned above.

A discrete $p$-adic discrete deep belief network, denoted as
$DBN(p,l,\boldsymbol{\theta}_{l})$, is defined by an energy functional of type%

\[
E_{l}\left(  \boldsymbol{v}_{l},\boldsymbol{h}_{l};\boldsymbol{\theta}%
_{l}\right)  =-%
{\displaystyle\sum\limits_{\boldsymbol{i}\in G_{l}}}
a_{\boldsymbol{i}}v_{\boldsymbol{i}}-%
{\displaystyle\sum\limits_{\boldsymbol{i}\in G_{l}}}
b_{\boldsymbol{i}}h_{\boldsymbol{i}}-%
{\displaystyle\sum\limits_{\boldsymbol{i},\boldsymbol{j}\in G_{l}}}
h_{\boldsymbol{i}}w_{\boldsymbol{i},\boldsymbol{j}}v_{\boldsymbol{j}}+%
{\displaystyle\sum\limits_{\boldsymbol{i}\in G_{l}}}
c_{\boldsymbol{i}}v_{\boldsymbol{i}}^{4}+%
{\displaystyle\sum\limits_{\boldsymbol{i}\in G_{l}}}
d_{\boldsymbol{i}}h_{\boldsymbol{i}}^{4},
\]
where $G_{l}$ is a topological group; here, this means that it is a rooted
tree with $l$ layers, and also an additive group. As additive group $G_{l}$ is
isomorphic to the integers modulo $p^{l}$, this means that the elements of
$G_{l}$ has the form%
\[
\boldsymbol{i}=i_{0}+i_{1}p+\ldots+i_{k}p^{k}+\ldots+i_{l-1}p^{l-1}\text{,
with }i_{k}\in\left\{  0,1,\ldots,p-1\right\}  .
\]
The $p$-adic unit ball $\mathbb{Z}_{p}$ is the (inductive) limit of $G_{l}$
when $l$ tends to infinity. Thus, the elements of $\mathbb{Z}_{p}$ has the
form $\sum_{k=0}^{\infty}i_{k}p^{k}$.

From the perspective of the energy functionals, the discrete $p$-adic DBMs are
generalizations of the standard restricted Boltzmann machines and are a new
class of DBMs. However, the notion of deep architecture in the $p$-adic case
completely differs from the standard one. $p$-Adic numbers index the hidden
and visible units. Since the $p$-adics are naturally organized in a tree-like
structure, the hidden units $h_{\boldsymbol{i}}$, $\boldsymbol{i}\in G_{l}$,
(respectively, the visible units $v_{\boldsymbol{i}}$, $\boldsymbol{i}\in
G_{l}$) are are organized in a rooted tree $G_{l}$ with $l$ layers.

The network $DBN(p,l,\boldsymbol{\theta}_{l})$ is obtained from an continuous
$DBN(p,\boldsymbol{\theta})$, by a discretization process. In
$DBN(p,l,\boldsymbol{\theta}_{l})$, only the hidden and visible units (nodes
of $G_{l}$) at the top level (level $l$) are interacting through the weights
$w\left(  i,j\right)  $; the other hidden units located at layers $k$ with
$1<k<l$ are \ connected in a tree-like structure, but there are no weights
attached to these connections. In the standard DBNs, the energy functional
contains terms for connection between units of neighboring hidden layers.
Then, the $p$-adic DBNs are a new type of DBNs.

The topology of the unit ball$\ \mathbb{Z}_{p}$ strongly influences the
functions defined on it. If $\boldsymbol{h}:\mathbb{Z}_{p}\rightarrow
\mathbb{R}$ is the hidden field of a network $DBN(p,\boldsymbol{\theta})$,
then the restriction of $\boldsymbol{h}$ to $G_{l}$ gives the hidden field
$\boldsymbol{h}_{l}:G_{l}\rightarrow\mathbb{R}$ of network
$DBN(p,l,\boldsymbol{\theta}_{l})$, and since $G_{k}\subset G_{l}$, as sets,
for $1<k<l$, then the restriction of $\boldsymbol{h}_{l}$ to $G_{k}$ gives the
hidden field $\boldsymbol{h}_{k}:G_{k}\rightarrow\mathbb{R}$ of network
$DBN(p,k,\boldsymbol{\theta}_{k})$. This fact is a consequence that
$\mathbb{Z}_{p}\smallsetminus\left\{  0\right\}  $ is a self-similar set
constructed from finite rooted trees $G_{l}$ by the action of a group
consisting of translations and dilations, which are multiplications by powers
of $p$.

In principle, the $p$-adic $DBN(p,l,\boldsymbol{\theta}_{l})$ have less $w$
parameters than the conventional ones. This fact suggests that the $p$-adic
DBNs are faster than the classical counterparts. In \cite{Valueva et al}, the
authors proposed a convolutional neural network architecture in which the
neural network is divided into hardware and software parts to increase
performance and reduce the cost of implementation resources. They also propose
to use the residue number system in the hardware part to implement the
convolutional layer of the neural network. The use of the proposed
heterogeneous implementation reduces the average time of image recognition by
41.17\%. In terms of our results, the mentioned hardware implementation
computes terms of the $\sum_{\boldsymbol{i}\in G_{l}}w\left(  \boldsymbol{i}%
-\boldsymbol{j}\right)  \boldsymbol{h}(\boldsymbol{i})\boldsymbol{v}%
(\boldsymbol{j})$ for groups $G_{l}$ of the form $\mathbb{Z}/m^{l}\mathbb{Z}$,
where $m$ is a suitable chosen positive integer. In addition due to the fact
that $G_{l}$ is an additive group in convolutional $p$-adic DBNs, the number
of $w$ parameters has a linear size while in the standard case is quadratic,
see \cite{Zuniga-He-Zambrano}.

The convolutional networks are particularly relevant in applications, and
consequently, the $p$-adic versions are particularly interesting. Then, two
fundamental questions emerge: are \ the $p$-adic convolutional DBNs universal
approximators? Does the computational power of the $p$-adic
convolutional\ DBNs increase as the number of levels of the $G_{l}$ tree
increases? The answer to both questions is yes, if the discrete visible and
hidden fields are binary-value random variables. These results were
established in \cite{Zuniga2023}. We briefly discuss these results.

Suppose that $\boldsymbol{Q}$ is a binary discrete distribution. It is not
difficult to see that one can consider $\boldsymbol{Q}=\boldsymbol{Q}\left(
\boldsymbol{v}_{l}\right)  $ as a distribution on a tree $G_{l}$, for some
suitable $l$. We use a $DBN(p,l,\boldsymbol{\theta}_{l})$ network with
Boltzmann probability distribution
\[
\boldsymbol{P}_{l}\left(  \boldsymbol{v}_{l};\boldsymbol{\theta}_{l}\right)
\varpropto\exp\left(  -E_{l}\left(  \boldsymbol{v}_{l},\boldsymbol{h}%
_{l};\boldsymbol{\theta}_{l}\right)  \right)
\]
to approximate $\boldsymbol{Q}\left(  \boldsymbol{v}_{l}\right)  $ using the
KL relative entropy (Kullback--Leibler divergence). We assume that
$KL(\boldsymbol{Q}\left(  \boldsymbol{v}_{l}\right)  \mid\boldsymbol{P}%
_{l}\left(  \boldsymbol{v}_{l};\boldsymbol{\theta}_{l}\right)  )>0$. We
construct a new DBM, $DBN(p,l+1,\boldsymbol{\theta}_{l},\boldsymbol{w}%
_{l+1},b_{j_{0}}^{l+1})$, here $\boldsymbol{w}_{l+1}\in\mathbb{R}^{\#G_{l}}$,
$b_{j_{0}}^{l+1}\in\mathbb{R}$, with an extra layer and an extra hidden unit,
and with the same visible units, whose energy functional
\[
E_{l+1}(\boldsymbol{v}_{l+1},\boldsymbol{h}_{l+1};\boldsymbol{\theta}%
_{l},\boldsymbol{w}_{l+1},b_{j_{0}}^{l+1})=E_{l+1}(\boldsymbol{v}%
_{l},\boldsymbol{h}_{l+1};\boldsymbol{\theta}_{l},\boldsymbol{w}%
_{l+1},b_{j_{0}}^{l+1})
\]
is an extension of $E_{l}\left(  \boldsymbol{v}_{l},\boldsymbol{h}%
_{l};\boldsymbol{\theta}_{l}\right)  $, $\boldsymbol{h}_{l+1}=\left[
\begin{array}
[c]{c}%
\boldsymbol{h}_{l}\\
h_{j_{0}}^{l+1}%
\end{array}
\right]  $, and $h_{j_{0}}^{l+1}$ is the extra hidden unit. Then
\begin{equation}
KL(\boldsymbol{Q}\left(  \boldsymbol{v}_{l}\right)  \mid\boldsymbol{P}%
_{l+1}\left(  \boldsymbol{v}_{l};\boldsymbol{\theta}_{l},\boldsymbol{w}%
_{l+1},b_{j_{0}}^{l+1}\right)  )<KL(\boldsymbol{Q}\left(  \boldsymbol{v}%
_{l}\right)  \mid\boldsymbol{P}_{l}\left(  \boldsymbol{v}_{l}%
;\boldsymbol{\theta}_{l}\right)  ), \label{1}%
\end{equation}
for some $\boldsymbol{\theta}_{l},\boldsymbol{w}_{l+1},b_{j_{0}}^{l+1}$, see
\cite[Theorem 1]{Zuniga2023}. \textbf{\ }Inequality (\ref{1}) implies that the
$p$-adic discrete DBNs are universal approximators. Furthermore, the
computational power of a $DBN(p,l,\boldsymbol{\theta}_{l})$ increases with the
number of layers of $G_{l}$, and thus $DBN(p,l,\boldsymbol{\theta}_{l})$ \ is
an actual `deep network' in the standard sense. The mentioned results are very
technical and nonconstructive result. They are based on the work of Le Roux
and Bengio, see \cite{Zuniga2023} for further details.

In our view, for practical applications, only discrete $p$-adic DBNs are
needed. The demonstration that the $p$-adic discrete convolutional DBNs are
universal approximators is purely existential, which means that it does not
provide any criteria for implementing practical networks. In
\cite{Zuniga-He-Zambrano}, we showed that the $p$-adic convolutional DBMs can
do computations with actual data (binary images). We also discuss some types
of learning algorithms, the tuning the parameters using contrastive
divergence. We also implemented and tested $p$-adic DBNs with images of the
MNIST dataset.

There is a consensus about the need of a theory for understanding how large
biological and artificial networks (with deep architectures) work. In the
biological framework, it is known that in such networks different parts, which
are not necessary physically connected, collaborate to perform a task (an
emergent behavior). In our view, this assertion can be easily accepted in the
case of artificial neural networks, since most of them are bioinspired
constructions. The comparison the neural activity between different parts of a
large neural, drives naturally to the computation of correlation functions.
Then, in principle, the activity of a neural network can be described by
computing all the possible correlations functions associated with all the
regions of \ a given network. A such collection of correlations functions is
just a statistical field theory on a network. Recently, it has been proposed
the existence of a correspondence between neural networks (NNs) and quantum
field theories (QFTs), more precisely with Euclidean QFTs, see, e.g.,
\cite{Batchits et al}, \cite{Erbin}, \cite{Halverson et al}, \cite{Moritz
Dahmen}, \cite{Roberts-Yaida}, \cite{Yaida}. On the biological side, Buice and
Cowan developed the neocortex activity's statistical field theory (an
Euclidean field theory). Using this theory, the neural activity can be
understood through the correlation functions, which can be computed
perturbatively from a generating functional, see \cite{Buice}- \cite{Chow},
and the references therein. In all these works, a statistical field theory is
given by a ill-defined measure of the form%
\[
\frac{\exp(-E\left(  \varphi;\theta\right)  )}{Z}\mathcal{D}\varphi\text{,
where }Z=%
{\displaystyle\int\limits_{\text{all fields}}}
\mathcal{D}\varphi\exp(-E\left(  \varphi;\theta\right)  ),
\]
here $\mathcal{D}\varphi$ is a `spurious measure' over the set all fields, and
$E\left(  \varphi;\theta\right)  $ is an energy functional depending on the
parameters $\theta$. A rigorous construction of the measure $\mathcal{D}%
\varphi$, for relevant physical systems, starting from a discrete energy
functional is a complicated mathematical problem. In all the works above
mentioned on SFTs associated with NNs the authors do not even mention the
issue of the construction of the measure $\mathcal{D}\varphi$; they use
standard quantum field techniques to compute meaningful quantities. To our
knowledge, this work is the first effort to construct mathematically rigorous
SFTs associated with certain types of NNs. We should emphasize that working
with `true measures' versus `spurious measures' requires an entirely different
mathematical framework. Another approach is using techniques of quantum field
theory to compute correlations in a perturbative form, see, e.g.,
\cite{Yaida}; in this way, the existence of the measure $\mathcal{D}\varphi$
is avoided, but in this case, there is no a correspondence between NNs and SFTs.

In all the works mentioned above on SFTs associated with NNs, the authors
assume that the classical perturbative theory based on Feynman diagrams can be
immediately applied. However, the energy functionals attached to neural
networks tend to be nonlocal. For instance, terms like%
\[%
{\displaystyle\int\limits_{\mathbb{Z}_{p}^{N}}}
w(x,y)\boldsymbol{h}(x)\boldsymbol{v}(y)d^{N}xd^{N}y
\]
are purely nonlocal, and then, the use of functional integrals can be used
with caution since locality is required to use perturbative quantum field
theory; see \cite[Section 5.1]{Zinn-Justin}. This is consistent with our
finding that the general correlation functions cannot be calculated using the
Wick-Isserlis theorem.

\section{\label{Section_2} Basic facts on $p$-adic analysis}

In this section we fix the notation and collect some basic results on $p$-adic
analysis that we will use through the article. For a detailed exposition on
$p$-adic analysis the reader may consult \cite{V-V-Z}, \cite{A-K-S}%
-\cite{Taibleson}.

\subsection{The field of $p$-adic numbers}

Along this article $p$ will denote a prime number. The field of $p-$adic
numbers $\mathbb{Q}_{p}$ is defined as the completion of the field of rational
numbers $\mathbb{Q}$ with respect to the $p-$adic norm $|\cdot|_{p}$, which is
defined as
\[
|x|_{p}=%
\begin{cases}
0 & \text{if }x=0\\
p^{-\gamma} & \text{if }x=p^{\gamma}\dfrac{a}{b},
\end{cases}
\]
where $a$ and $b$ are integers coprime with $p$. The integer $\gamma
=ord_{p}(x):=ord(x)$, with $ord(0):=+\infty$, is called the\textit{\ }$p-$adic
order of $x$. We extend the $p-$adic norm to $\mathbb{Q}_{p}^{N}$ by taking%
\[
||x||_{p}:=\max_{1\leq i\leq N}|x_{i}|_{p},\qquad\text{for }x=(x_{1}%
,\dots,x_{N})\in\mathbb{Q}_{p}^{N}.
\]
We define $ord(x)=\min_{1\leq i\leq N}\{ord(x_{i})\}$, then $||x||_{p}%
=p^{-ord(x)}$.\ The metric space $\left(  \mathbb{Q}_{p}^{N},||\cdot
||_{p}\right)  $ is a complete ultrametric space. As a topological space
$\mathbb{Q}_{p}$\ is homeomorphic to a Cantor-like subset of the real line,
see, e.g., \cite{V-V-Z}, \cite{A-K-S}.

Any $p-$adic number $x\neq0$ has a unique expansion of the form
\[
x=p^{ord(x)}\sum_{j=0}^{\infty}x_{j}p^{j},
\]
where $x_{j}\in\{0,1,2,\dots,p-1\}$ and $x_{0}\neq0$. By using this expansion,
we define \textit{the fractional part }$\{x\}_{p}$\textit{ of }$x\in
\mathbb{Q}_{p}$ as the rational number
\[
\{x\}_{p}=%
\begin{cases}
0 & \text{if }x=0\text{ or }ord(x)\geq0\\
p^{ord(x)}\sum_{j=0}^{-ord(x)-1}x_{j}p^{j} & \text{if }ord(x)<0.
\end{cases}
\]
In addition, any $x\in\mathbb{Q}_{p}^{N}\smallsetminus\left\{  0\right\}  $
can be represented uniquely as $x=p^{ord(x)}v$, where $\left\Vert v\right\Vert
_{p}=1$.

\subsection{Topology of $\mathbb{Q}_{p}^{N}$}

For $r\in\mathbb{Z}$, denote by $B_{r}^{N}(a)=\{x\in\mathbb{Q}_{p}%
^{N};||x-a||_{p}\leq p^{r}\}$ the ball of radius $p^{r}$ with center at
$a=(a_{1},\dots,a_{N})\in\mathbb{Q}_{p}^{N}$, and take $B_{r}^{N}%
(0):=B_{r}^{N}$. Note that $B_{r}^{N}(a)=B_{r}(a_{1})\times\cdots\times
B_{r}(a_{N})$, where $B_{r}(a_{i}):=\{x\in\mathbb{Q}_{p};|x_{i}-a_{i}|_{p}\leq
p^{r}\}$ is the one-dimensional ball of radius $p^{r}$ with center at
$a_{i}\in\mathbb{Q}_{p}$. The ball $B_{0}^{N}$ equals the product of $N$
copies of $B_{0}=\mathbb{Z}_{p}$, the ring of\textit{ }$p-$adic integers. We
also denote by $S_{r}^{N}(a)=\{x\in\mathbb{Q}_{p}^{N};||x-a||_{p}=p^{r}\}$ the
sphere of radius\textit{ }$p^{r}$ with center at $a=(a_{1},\dots,a_{N}%
)\in\mathbb{Q}_{p}^{N}$, and take $S_{r}^{N}(0):=S_{r}^{N}$. We notice that
$S_{0}^{1}=\mathbb{Z}_{p}^{\times}$ (the group of units of $\mathbb{Z}_{p}$),
but $\left(  \mathbb{Z}_{p}^{\times}\right)  ^{N}\subsetneq S_{0}^{N}$. The
balls and spheres are both open and closed subsets in $\mathbb{Q}_{p}^{N}$. In
addition, two balls in $\mathbb{Q}_{p}^{N}$ are either disjoint or one is
contained in the other.

As a topological space $\left(  \mathbb{Q}_{p}^{N},||\cdot||_{p}\right)  $ is
totally disconnected, i.e., the only connected \ subsets of $\mathbb{Q}%
_{p}^{N}$ are the empty set and the points. A subset of $\mathbb{Q}_{p}^{N}$
is compact if and only if it is closed and bounded in $\mathbb{Q}_{p}^{N}$,
see, e.g., \cite[Section 1.3]{V-V-Z}, or \cite[Section 1.8]{A-K-S}. The balls
and spheres are compact subsets. Thus $\left(  \mathbb{Q}_{p}^{N}%
,||\cdot||_{p}\right)  $ is a locally compact topological space.

\subsection{The Haar measure}

Since $(\mathbb{Q}_{p}^{N},+)$ is a locally compact topological group, there
exists a Haar measure $d^{N}x$, which is invariant under translations, i.e.,
$d^{N}(x+a)=d^{N}x$, \cite{Halmos}. If we normalize this measure by the
condition $\int_{\mathbb{Z}_{p}^{N}}dx=1$, then $d^{N}x$ is unique.

\begin{notation}
We will use $\Omega\left(  p^{-r}||x-a||_{p}\right)  $ to denote the
characteristic function of the ball $B_{r}^{N}(a)=a+p^{-r}\mathbb{Z}_{p}^{N}$,
where
\[
\mathbb{Z}_{p}^{N}=\left\{  x\in\mathbb{Q}_{p}^{N};\left\Vert x\right\Vert
_{p}\leq1\right\}
\]
is the $N$-dimensional unit ball. For more general sets, we will use the
notation $1_{A}$ for the characteristic function of set $A$.
\end{notation}

\subsection{The Bruhat-Schwartz space}

A complex-valued function $\varphi$ defined on $\mathbb{Q}_{p}^{N}$ is
\textit{called locally constant} if for any $x\in\mathbb{Q}_{p}^{N}$ there
exist an integer $l(x)\in\mathbb{Z}$ such that%
\begin{equation}
\varphi(x+x^{\prime})=\varphi(x)\text{ for any }x^{\prime}\in B_{l(x)}^{N}.
\label{local_constancy}%
\end{equation}
A function $\varphi:\mathbb{Q}_{p}^{N}\rightarrow\mathbb{C}$ is called a
Bruhat-Schwartz function\textit{ }(or a test function) if it is locally
constant with compact support. Any test function can be represented as a
linear combination, with complex coefficients, of characteristic functions of
balls. The $\mathbb{C}$-vector space of Bruhat-Schwartz functions is denoted
by $\mathcal{D}(\mathbb{Q}_{p}^{N})$. We denote by $\mathcal{D}_{\mathbb{R}%
}(\mathbb{Q}_{p}^{N})$\ the $\mathbb{R}$-vector space of Bruhat-Schwartz
functions. For $\varphi\in\mathcal{D}(\mathbb{Q}_{p}^{N})$, the largest number
$l=l(\varphi)$ satisfying (\ref{local_constancy}) is called the exponent of
local constancy (or the parameter of constancy) of $\varphi$.

We denote by $\mathcal{D}_{m}^{l}(\mathbb{Q}_{p}^{N})$ the finite-dimensional
space of test functions from $\mathcal{D}(\mathbb{Q}_{p}^{N})$ having supports
in the ball $B_{m}^{N}$ and with parameters \ of constancy $\geq l$. We now
define a topology on $\mathcal{D}(\mathbb{Q}_{p}^{N})$ as follows. We say that
a sequence $\left\{  \varphi_{j}\right\}  _{j\in\mathbb{N}}$ of functions in
$\mathcal{D}(\mathbb{Q}_{p}^{N})$ converges to zero, if the two following
conditions hold true:

(1) there are two fixed integers $k_{0}$ and $m_{0}$ such that \ each
$\varphi_{j}\in$ $\mathcal{D}_{m_{0}}^{k_{0}}(\mathbb{Q}_{p}^{N})$;

(2) $\varphi_{j}\rightarrow0$ uniformly.

$\mathcal{D}(\mathbb{Q}_{p}^{N})$ endowed with the above topology becomes a
topological vector space.

\subsection{$L^{\rho}$ spaces}

Given $\rho\in\lbrack0,\infty)$, we denote by$L^{\rho}\left(
\mathbb{Q}
_{p}^{N}\right)  :=L^{\rho}\left(
\mathbb{Q}
_{p}^{N},d^{N}x\right)  ,$ the $\mathbb{C}-$vector space of all the complex
valued functions $g$ satisfying
\[
\left\Vert g\right\Vert _{\rho}=\left(  \text{ }%
{\displaystyle\int\limits_{\mathbb{Q}_{p}^{N}}}
\left\vert g\left(  x\right)  \right\vert ^{\rho}d^{N}x\right)  ^{\frac
{1}{\rho}}<\infty,
\]
where $d^{N}x$ is the normalized Haar measure on $\left(  \mathbb{Q}_{p}%
^{N},+\right)  $. The corresponding $\mathbb{R}$-vector spaces are denoted as
$L_{\mathbb{R}}^{\rho}\left(
\mathbb{Q}
_{p}^{N}\right)  =L_{\mathbb{R}}^{\rho}\left(
\mathbb{Q}
_{p}^{N},d^{N}x\right)  $, $1\leq\rho<\infty$.

If $U$ is an open subset of $\mathbb{Q}_{p}^{N}$, $\mathcal{D}(U)$ denotes the
$\mathbb{C}$-vector space of test functions with supports contained in $U$,
then $\mathcal{D}(U)$ is dense in
\[
L^{\rho}\left(  U\right)  =\left\{  \varphi:U\rightarrow\mathbb{C};\left\Vert
\varphi\right\Vert _{\rho}=\left\{
{\displaystyle\int\limits_{U}}
\left\vert \varphi\left(  x\right)  \right\vert ^{\rho}d^{N}x\right\}
^{\frac{1}{\rho}}<\infty\right\}  ,
\]
for $1\leq\rho<\infty$, see, e.g., \cite[Section 4.3]{A-K-S}. We denote by
$L_{\mathbb{R}}^{\rho}\left(  U\right)  $ the real counterpart of $L^{\rho
}\left(  U\right)  $.

\subsection{The Fourier transform}

Set $\chi_{p}(y)=\exp(2\pi i\{y\}_{p})$ for $y\in\mathbb{Q}_{p}$. The map
$\chi_{p}(\cdot)$ is an additive character on $\mathbb{Q}_{p}$, i.e., a
continuous map from $\left(  \mathbb{Q}_{p},+\right)  $ into $S$ (the unit
circle considered as multiplicative group) satisfying $\chi_{p}(x_{0}%
+x_{1})=\chi_{p}(x_{0})\chi_{p}(x_{1})$, $x_{0},x_{1}\in\mathbb{Q}_{p}$.\ The
additive characters of $\mathbb{Q}_{p}$ form an Abelian group which is
isomorphic to $\left(  \mathbb{Q}_{p},+\right)  $. The isomorphism is given by
$\kappa\rightarrow\chi_{p}(\kappa x)$, see, e.g., \cite[Section 2.3]{A-K-S}.

Given $\xi=(\xi_{1},\dots,\xi_{N})$ and $x=(x_{1},\dots,x_{N})\allowbreak
\in\mathbb{Q}_{p}^{N}$, we set $\xi\cdot x:=\sum_{j=1}^{N}\xi_{j}x_{j}$. The
Fourier transform of $\varphi\in\mathcal{D}(\mathbb{Q}_{p}^{N})$ is defined
as
\[
\mathcal{F}\varphi(\xi)=%
{\displaystyle\int\limits_{\mathbb{Q} _{p}^{N}}}
\chi_{p}(\xi\cdot x)\varphi(x)d^{N}x\quad\text{for }\xi\in\mathbb{Q}_{p}^{N},
\]
where $d^{N}x$ is the normalized Haar measure on $\mathbb{Q}_{p}^{N}$. The
Fourier transform is a linear isomorphism from $\mathcal{D}(\mathbb{Q}_{p}%
^{N})$ onto itself satisfying
\begin{equation}
(\mathcal{F}(\mathcal{F}\varphi))(\xi)=\varphi(-\xi), \label{Eq_FFT}%
\end{equation}
see, e.g., \cite[Section 4.8]{A-K-S}. We will also use the notation
$\mathcal{F}_{x\rightarrow\kappa}\varphi$ and $\widehat{\varphi}$\ for the
Fourier transform of $\varphi$.

The Fourier transform extends to $L^{2}$. If $f\in L^{2}\left(  \mathbb{Q}%
_{p}^{N}\right)  $, its Fourier transform is defined as
\[
(\mathcal{F}f)(\xi)=\lim_{k\rightarrow\infty}%
{\displaystyle\int\limits_{||x||_{p}\leq p^{k}}}
\chi_{p}(\xi\cdot x)f(x)d^{N}x,\quad\text{for }\xi\in%
\mathbb{Q}
_{p}^{N},
\]
where the limit is taken in $L^{2}\left(  \mathbb{Q}_{p}^{N}\right)  $. We
recall that the Fourier transform is unitary on $L^{2}\left(  \mathbb{Q}%
_{p}^{N}\right)  ,$ i.e. $||f||_{2}=||\mathcal{F}f||_{2}$ for $f\in L^{2}$ and
that (\ref{Eq_FFT}) is also valid in $L^{2}$, see, e.g., \cite[Chapter III,
Section 2]{Taibleson}.

\subsection{Distributions}

The $\mathbb{C}$-vector space $\mathcal{D}^{\prime}\left(  \mathbb{Q}_{p}%
^{N}\right)  $ of all continuous linear functionals on $\mathcal{D}%
(\mathbb{Q}_{p}^{N})$ is called the Bruhat-Schwartz space of distributions.
Every linear functional on $\mathcal{D}(\mathbb{Q}_{p}^{N})$ is continuous,
i.e. $\mathcal{D}^{\prime}\left(  \mathbb{Q}_{p}^{N}\right)  $\ agrees with
the algebraic dual of $\mathcal{D}(\mathbb{Q}_{p}^{N})$, see, e.g.,
\cite[Chapter 1, VI.3, Lemma]{V-V-Z}. We denote by $\mathcal{D}_{\mathbb{R}%
}^{\prime}\left(  \mathbb{Q}_{p}^{N}\right)  $ the dual space of
$\mathcal{D}_{\mathbb{R}}\left(  \mathbb{Q}_{p}^{N}\right)  $.

We endow $\mathcal{D}^{\prime}\left(  \mathbb{Q}_{p}^{N}\right)  $ with the
weak topology, i.e. a sequence $\left\{  T_{j}\right\}  _{j\in\mathbb{N}}$ in
$\mathcal{D}^{\prime}\left(  \mathbb{Q}_{p}^{N}\right)  $ converges to $T$ if
$\lim_{j\rightarrow\infty}T_{j}\left(  \varphi\right)  =T\left(
\varphi\right)  $ for any $\varphi\in\mathcal{D}(\mathbb{Q}_{p}^{N})$. The
map
\[%
\begin{array}
[c]{lll}%
\mathcal{D}^{\prime}\left(  \mathbb{Q}_{p}^{N}\right)  \times\mathcal{D}%
(\mathbb{Q}_{p}^{N}) & \rightarrow & \mathbb{C}\\
\left(  T,\varphi\right)  & \rightarrow & T\left(  \varphi\right)
\end{array}
\]
is a bilinear form which is continuous in $T$ and $\varphi$ separately. We
call this map the pairing between $\mathcal{D}^{\prime}\left(  \mathbb{Q}%
_{p}^{N}\right)  $ and $\mathcal{D}(\mathbb{Q}_{p}^{N})$. From now on we will
use $\left(  T,\varphi\right)  $ instead of $T\left(  \varphi\right)  $.

Every $f$\ in $L_{loc}^{1}$ defines a distribution $f\in\mathcal{D}^{\prime
}\left(  \mathbb{Q}_{p}^{N}\right)  $ by the formula
\[
\left(  f,\varphi\right)  =%
{\displaystyle\int\limits_{\mathbb{Q}_{p}^{N}}}
f\left(  x\right)  \varphi\left(  x\right)  d^{N}x.
\]

\subsection{\label{SEction Fourier Transform}The Fourier transform of a
distribution}

The Fourier transform $\mathcal{F}\left[  T\right]  $ of a distribution
$T\in\mathcal{D}^{\prime}\left(  \mathbb{Q}_{p}^{N}\right)  $ is defined by%
\[
\left(  \mathcal{F}\left[  T\right]  ,\varphi\right)  =\left(  T,\mathcal{F}%
\left[  \varphi\right]  \right)  \text{ for all }\varphi\in\mathcal{D}\left(
\mathbb{Q}_{p}^{N}\right)  \text{.}%
\]
The Fourier transform $T\rightarrow\mathcal{F}\left[  T\right]  $ is a linear
and continuous isomorphism from $\mathcal{D}^{\prime}\left(  \mathbb{Q}%
_{p}^{N}\right)  $\ onto $\mathcal{D}^{\prime}\left(  \mathbb{Q}_{p}%
^{N}\right)  $. Furthermore, $T=\mathcal{F}\left[  \mathcal{F}\left[
T\right]  \left(  -\xi\right)  \right]  $.

Let $T\in\mathcal{D}^{\prime}\left(  \mathbb{Q}_{p}^{n}\right)  $ be a
distribution. Then supp$T\subset B_{L}^{N}$ if and only if $\mathcal{F}\left[
T\right]  $ is a locally constant function, and the exponent of local
constancy of $\mathcal{F}\left[  T\right]  $ is $\geq-L$. In addition%

\[
\mathcal{F}\left[  T\right]  \left(  \xi\right)  =\left(  T\left(  y\right)
,\Omega\left(  p^{-L}\left\Vert y\right\Vert _{p}\right)  \chi_{p}\left(
\xi\cdot y\right)  \right)  ,
\]
see, e.g., \cite[Section 4.9]{A-K-S}.

\subsection{The direct product of distributions}

Given $F\in\mathcal{D}^{\prime}\left(  \mathbb{Q}_{p}^{N}\right)  $ and
$G\in\mathcal{D}^{\prime}\left(  \mathbb{Q}_{p}^{M}\right)  $, their
\textit{direct product }$F\times G$ is defined by the formula%
\[
\left(  F\left(  x\right)  \times G\left(  y\right)  ,\varphi\left(
x,y\right)  \right)  =\left(  F\left(  x\right)  ,\left(  G\left(  y\right)
,\varphi\left(  x,y\right)  \right)  \right)  \text{ for }\varphi\left(
x,y\right)  \in\mathcal{D}\left(  \mathbb{Q}_{p}^{N+M}\right)  .
\]
The direct product is commutative: $F\times G=G\times F$. In addition the
direct product is continuous with respect to the joint factors.

\subsection{The convolution of distributions}

Given $F,G\in\mathcal{D}^{\prime}\left(  \mathbb{Q}_{p}^{N}\right)  $, their
convolution $F\ast G$ is defined by%
\[
\left(  F\ast G,\varphi\right)  =\lim_{k\rightarrow\infty}\left(  F\left(
y\right)  \times G\left(  x\right)  ,\Omega\left(  p^{-k}\left\Vert
y\right\Vert _{p}\right)  \varphi\left(  x+y\right)  \right)
\]
if the limit exists for all $\varphi\in\mathcal{D}\left(  \mathbb{Q}_{p}%
^{N}\right)  $. We recall that if $F\ast G$ exists, then $G\ast F$ exists and
$F\ast G=G\ast F$, see, e.g., \cite[Section 7.1]{V-V-Z}. If $F,G\in
\mathcal{D}^{\prime}\left(  \mathbb{Q}_{p}^{N}\right)  $ and supp$G\subset
B_{L}^{n}$, then the convolution $F\ast G$ exists, and it is given by the
formula%
\[
\left(  F\ast G,\varphi\right)  =\left(  F\left(  y\right)  \times G\left(
x\right)  ,\Omega\left(  p^{-L}\left\Vert y\right\Vert _{p}\right)
\varphi\left(  x+y\right)  \right)  \text{ for }\varphi\in\mathcal{D}\left(
\mathbb{Q}_{p}^{N}\right)  .
\]
In the case in which $G=\psi\in\mathcal{D}\left(  \mathbb{Q}_{p}^{n}\right)
$, $F\ast\psi$ is a locally constant function given by
\[
\left(  F\ast\psi\right)  \left(  y\right)  =\left(  F\left(  x\right)
,\psi\left(  y-x\right)  \right)  ,
\]
see, e.g., \cite[Section 7.1]{V-V-Z}.

\subsection{The multiplication of distributions}

Set $\delta_{k}\left(  x\right)  :=p^{Nk}\Omega\left(  p^{k}\left\Vert
x\right\Vert _{p}\right)  $ for $k\in\mathbb{N}$. Given $F,G\in\mathcal{D}%
^{\prime}\left(  \mathbb{Q}_{p}^{N}\right)  $, their product $F\cdot G$ is
defined by%
\[
\left(  F\cdot G,\varphi\right)  =\lim_{k\rightarrow\infty}\left(  G,\left(
F\ast\delta_{k}\right)  \varphi\right)
\]
if the limit exists for all $\varphi\in\mathcal{D}\left(  \mathbb{Q}_{p}%
^{N}\right)  $. If the product $F\cdot G$ exists then the product $G\cdot F$
exists and they are equal.

We recall that \ the existence of the product $F\cdot G$ is equivalent \ to
the existence of $\mathcal{F}\left[  F\right]  \ast\mathcal{F}\left[
G\right]  $. In addition, $\mathcal{F}\left[  F\cdot G\right]  =\mathcal{F}%
\left[  F\right]  \ast\mathcal{F}\left[  G\right]  $ and $\mathcal{F}\left[
F\ast G\right]  =\mathcal{F}\left[  F\right]  \cdot\mathcal{F}\left[
G\right]  $, see, e.g., \cite[Section 7.5]{V-V-Z}.

\section{\label{Section_3} Gaussian measures in $L_{\mathbb{R}}^{2}\left(
\mathbb{Z}_{p}^{N}\right)  $}

\subsection{Some preliminary results on $L_{\mathbb{R}}^{2}\left(
\mathbb{Z}_{p}^{N}\right)  $}

We set
\[
L_{\mathbb{R}}^{2}\left(  \mathbb{Z}_{p}^{N}\right)  :=\left\{  f:\mathbb{Z}%
_{p}^{N}\rightarrow\mathbb{R};%
{\displaystyle\int\limits_{\mathbb{Z}_{p}^{N}}}
\left\vert f\right\vert ^{2}d^{N}x<\infty\right\}  .
\]
Thus $L_{\mathbb{R}}^{2}\left(  \mathbb{Z}_{p}^{N}\right)  $ is a real Hilbert
space with inner product $\left\langle f,g\right\rangle =\int_{\mathbb{Z}%
_{p}^{N}}fgd^{N}x$, and norm $\left\Vert f\right\Vert ^{2}=\left\langle
f,f\right\rangle $, for $f,g\in L_{\mathbb{R}}^{2}\left(  \mathbb{Z}_{p}%
^{N}\right)  $. We denote by $L^{2}\left(  \mathbb{Z}_{p}^{N}\right)  $ the
complexification of $L_{\mathbb{R}}^{2}\left(  \mathbb{Z}_{p}^{N}\right)  $,
which is a complex Hilbert space with inner product $\left\langle
h,k\right\rangle =\int_{\mathbb{Z}_{p}^{N}}h\overline{k}d^{N}x$, for $h,k\in
L^{2}\left(  \mathbb{Z}_{p}^{N}\right)  $, where the bar denotes the complex
conjugate. Along this article, we \ embed $L^{2}\left(  \mathbb{Z}_{p}%
^{N}\right)  $ into $L^{2}\left(  \mathbb{Q}_{p}^{N}\right)  $, by extending
the functions in $L^{2}\left(  \mathbb{Z}_{p}^{N}\right)  $ as zero outside of
the unit ball. In particular, if $f\in L_{\mathbb{R}}^{2}\left(
\mathbb{Z}_{p}^{N}\right)  $, we denote by $\widehat{f}$ its Fourier transform
when $f$ is considered as a function from $L^{2}\left(  \mathbb{Q}_{p}%
^{N}\right)  $. Notice that the support of $\widehat{f}$ is not necessarily
contained in $\mathbb{Z}_{p}^{N}$.

\begin{notation}
(i) We set $\left[  \xi\right]  _{p}:=\max\left\{  1,\left\Vert \xi\right\Vert
_{p}\right\}  $, for $\xi\in\mathbb{Q}_{p}^{N}$. If $h:\mathbb{Q}_{p}%
^{N}\rightarrow\mathbb{R}$, we write $h\left(  \xi\right)  =h\left(  \left[
\xi\right]  _{p}\right)  $ to mean that there is a function $s:\mathbb{R}%
_{+}\rightarrow\mathbb{R}$ such that $h\left(  \xi\right)  =s\left(  \left[
\xi\right]  _{p}\right)  $. Along this article we will identify $s$ with $h$.

\noindent(ii) We denote by $C(\mathbb{Q}_{p}^{N},\mathbb{R})$ the $\mathbb{R}%
$-vector space of functions $f:\mathbb{Q}_{p}^{N}\rightarrow\mathbb{R}$.
\end{notation}

\begin{remark}
The H\"{o}lder inequality and the fact that $\int_{\mathbb{Z}_{p}^{N}}%
d^{N}x=1$ imply that%
\[
\left\Vert f\right\Vert _{1}\leq\left\Vert f\right\Vert _{\rho}\text{, for
}1\leq\rho\leq\infty,
\]
which in turn implies that
\[
L^{\rho}(\mathbb{Z}_{p}^{N})\subset L^{1}(\mathbb{Z}_{p}^{N})\text{, }%
1\leq\rho\leq\infty.
\]
Given $f\in L^{1}(\mathbb{Z}_{p}^{N})$, we denote by $\left\Vert f\right\Vert
_{\infty}$ the essential supremum of $f$. Notice that $\mathcal{D}%
(\mathbb{Z}_{p}^{N})\subset L^{\rho}(\mathbb{Z}_{p}^{N})$, for $1\leq\rho
\leq\infty$, and that $\mathcal{D}(\mathbb{Z}_{p}^{N})$ is dense in $L^{\rho
}(\mathbb{Z}_{p}^{N})$, for $1\leq\rho<\infty$.
\end{remark}

From now on we fix a function $\widehat{K}:\mathbb{Q}_{p}^{N}\rightarrow
\mathbb{R}$ \ satisfying the following:

\begin{enumerate}
\item[(H1)] $\widehat{K}\in$ $L_{\mathbb{R}}^{1}\left(  \mathbb{Q}_{p}%
^{N}\right)  \cap C(\mathbb{Q}_{p}^{N},\mathbb{R})$,

\item[(H2)] $\widehat{K}\left(  \xi\right)  =\widehat{K}(\left[  \xi\right]
_{p})$, for $\xi\in\mathbb{Q}_{p}^{N}$,

\item[(H3)] $\widehat{K}\left(  \xi\right)  >0$, for $x\in\mathbb{Q}_{p}^{N}$.
\end{enumerate}

Then, since $\widehat{K}(\left[  \xi\right]  _{p})$ is integrable, $K$ is a
continuous real-valued function. Furthermore, as a distribution, $K$ is
supported in the unit ball, consequently $K:\mathbb{Z}_{p}^{N}\rightarrow
\mathbb{R}$ is a continuous real-valued function supported in the unit ball,
see Section \ref{SEction Fourier Transform}, alternatively see \cite[Section
4.9]{A-K-S}, and thus $K\in L_{\mathbb{R}}^{1}\left(  \mathbb{Z}_{p}%
^{N}\right)  $. Notice that $\left\Vert \widehat{K}\right\Vert _{1}%
=\widehat{K}(\left[  0\right]  _{p})=\widehat{K}\left(  1\right)  >0$.

We now define the operator%
\[%
\begin{array}
[c]{ccc}%
L_{\mathbb{R}}^{2}\left(  \mathbb{Z}_{p}^{N}\right)  & \rightarrow &
L_{\mathbb{R}}^{2}\left(  \mathbb{Z}_{p}^{N}\right) \\
f & \rightarrow & \square_{K}f,
\end{array}
\]
where $\left(  \square_{K}f\right)  \left(  x\right)  :=K(x)\ast f(x)$. Since
$\left(  \mathbb{Z}_{p}^{N},+\right)  $ is an Abelian group, the function
$K(x)\ast f(x)$\ is supported in the unit ball. Thus, this mapping is a
well-defined linear, bounded operator, since%
\[
\left\Vert \square_{K}f\right\Vert _{2}\leq\left\Vert K\right\Vert
_{1}\left\Vert f\right\Vert _{2}\text{.}%
\]

Notice that%
\[
\left(  \square_{K}f\right)  \left(  x\right)  :=\mathcal{F}_{\xi\rightarrow
x}^{-1}\left(  \widehat{K}(\xi)\mathcal{F}_{x\rightarrow\xi}f\right)  .
\]

\begin{lemma}
\label{Lemma1}$\square_{K}$ is a symmetric, positive operator.
\end{lemma}

\begin{proof}
Take $f,g\in L_{\mathbb{R}}^{2}\left(  \mathbb{Z}_{p}^{N}\right)  $, then by
using the fact that the Fourier transform preserves the inner product in
$L^{2}\left(  \mathbb{Q}_{p}^{N}\right)  $, we have
\begin{align*}
\left\langle \square_{K}f,g\right\rangle  &  =%
{\displaystyle\int\limits_{\mathbb{Z}_{p}^{N}}}
\left(  \square_{K}f\right)  gd^{N}x=%
{\displaystyle\int\limits_{\mathbb{Q}_{p}^{N}}}
\left(  \square_{K}f\right)  \text{ }\overline{g}d^{N}x\\
&  =%
{\displaystyle\int\limits_{\mathbb{Q}_{p}^{N}}}
\left(  \widehat{\square_{K}f}\right)  \text{ }\overline{\widehat{g}}d^{N}\xi=%
{\displaystyle\int\limits_{\mathbb{Q}_{p}^{N}}}
\widehat{K}(\xi)\widehat{f}\text{ }\overline{\widehat{g}}d^{N}\xi=%
{\displaystyle\int\limits_{\mathbb{Q}_{p}^{N}}}
\widehat{f}\text{ }\overline{\left(  \widehat{K}(\xi)\widehat{g}\right)
}d^{N}\xi\\
&  =%
{\displaystyle\int\limits_{\mathbb{Q}_{p}^{N}}}
f\text{ }\left(  \overline{\square_{K}g}\right)  d^{N}x=%
{\displaystyle\int\limits_{\mathbb{Q}_{p}^{N}}}
f\text{ }\left(  \square_{K}g\right)  d^{N}x=\left\langle f,\square
_{K}g\right\rangle .
\end{align*}
Thus, the operator $\square_{K}$ is symmetric. The positivity follows from%
\[
\left\langle \square_{K}f,f\right\rangle =%
{\displaystyle\int\limits_{\mathbb{Q}_{p}^{N}}}
\widehat{K}(\xi)\left\vert \widehat{f}\right\vert ^{2}d^{N}\xi\geq0,
\]
since $\widehat{K}(\xi)>0$, by hypothesis (H3).
\end{proof}

\begin{example}
Let $k:\mathbb{R}_{+}\rightarrow\mathbb{R}_{+}$, where $\mathbb{R}%
_{+}:=\left\{  x\in\mathbb{R};x\geq0\right\}  $. Any function of type%
\begin{equation}
\widehat{K}(\xi)=k(\left[  \xi\right]  _{p})\text{, with }k(1)+\left(
1-p^{-N}\right)
{\displaystyle\sum\limits_{j=1}^{\infty}}
p^{Nj}k(p^{j})<\infty\text{,} \label{Type-H}%
\end{equation}
satisfies hypotheses (H1)-(H3). The convergence of the series in
(\ref{Type-H}) is equivalent to hypothesis (H1).
\end{example}

\begin{remark}
\label{Note_Inverse_square}By using that $\frac{1}{\widehat{K}(\xi)}$ is a
well-defined continuous function on $\mathbb{Q}_{p}^{N}$, by the
Cauchy-Schwarz inequality, we have $\frac{1}{\widehat{K}(\xi)}\widehat{f}\in
L_{loc}^{1}\left(  \mathbb{Q}_{p}^{N}\right)  $, for any $f\in L_{\mathbb{R}%
}^{2}\left(  \mathbb{Z}_{p}^{N}\right)  $, and thus the mapping%
\[%
\begin{array}
[c]{ccc}%
L_{\mathbb{R}}^{2}\left(  \mathbb{Z}_{p}^{N}\right)  & \rightarrow &
\mathcal{D}_{\mathbb{R}}^{\prime}\left(  \mathbb{Q}_{p}^{N}\right) \\
f & \rightarrow & \square_{K}^{-1}f,
\end{array}
\]
with%
\[
\left(  \square_{K}^{-1}f\right)  \left(  x\right)  :=\mathcal{F}%
_{\xi\rightarrow x}^{-1}\left(  \frac{1}{\widehat{K}(\xi)}\mathcal{F}%
_{x\rightarrow\xi}f\right)
\]
is a well-defined operator. Notice that $\square_{K}^{-1}\square_{K}%
f=\square_{K}\square_{K}^{-1}f=f$ for any $f\in L_{\mathbb{R}}^{2}\left(
\mathbb{Z}_{p}^{N}\right)  $. In the verification of this assertion we use
that $\square_{K}$ has an extension of the subspace of distributions $T\in$
$\mathcal{D}_{\mathbb{R}}^{\prime}\left(  \mathbb{Q}_{p}^{N}\right)  $ such
that $\widehat{K}(\xi)\widehat{T}\in\mathcal{D}^{\prime}\left(  \mathbb{Q}%
_{p}^{N}\right)  $.

On the other hand, $\mathcal{F}_{\xi\rightarrow x}^{-1}\left(  \frac
{1}{\widehat{K}(\xi)}\right)  \in\mathcal{D}_{\mathbb{R}}^{\prime}\left(
\mathbb{Z}_{p}^{N}\right)  $, and
\[
\left(  \square_{K}^{-1}f\right)  \left(  x\right)  =\mathcal{F}%
_{\xi\rightarrow x}^{-1}\left(  \frac{1}{\widehat{K}(\xi)}\right)  \ast f(x)
\]
is a function supported in $\mathbb{Z}_{p}^{N}$. In particular, if $f$ is a
test function $\square_{K}^{-1}f$ is a test function.
\end{remark}

\subsection{$p$-adic wavelets supported in balls}

The set of functions $\left\{  \Psi_{rnj}\right\}  $ defined as%
\begin{equation}
\Psi_{rnj}\left(  x\right)  =p^{\frac{-r}{2}}\chi_{p}\left(  p^{-1}j\left(
p^{r}x-n\right)  \right)  \Omega\left(  \left\vert p^{r}x-n\right\vert
_{p}\right)  , \label{eq4}%
\end{equation}
where $r\in\mathbb{Z}$, $j\in\left\{  1,\cdots,p-1\right\}  $, and $n$ runs
through a fixed set of representatives of $\mathbb{Q}_{p}/\mathbb{Z}_{p}$, is
an orthonormal basis of $L^{2}(\mathbb{Q}_{p})$, see, e.g., \cite[Theorem
3.29]{KKZuniga}, \cite[Theorem 9.4.2]{A-K-S}.\ By using this basis, it is
possible to construct an orthonormal basis for $L^{2}(\mathbb{Z}_{p})$:

\begin{lemma}
[{\cite[Proposition 1]{Zuniga-Eigen}}]The set of functions%
\[
\left\{  \Omega\left(  \left\vert x\right\vert _{p}\right)  \right\}
\bigcup\bigcup\limits_{j\in\left\{  1,\ldots,p-1\right\}  }\text{ }%
\bigcup\limits_{r\leq0}\text{ }\bigcup\limits_{\substack{np^{-r}\in
\mathbb{Z}_{p}\\n\in\mathbb{Q}_{p}/\mathbb{Z}_{p}}}\left\{  \Psi_{rnj}\left(
x\right)  \right\}
\]
is an orthonormal basis of $L^{2}\left(  \mathbb{Z}_{p}\right)  $.
\end{lemma}

We set $\boldsymbol{r}=\left(  r_{1},\ldots,r_{N}\right)  \in\left(
\mathbb{Z}_{\leq0}\right)  ^{N}$, with $\mathbb{Z}_{\leq0}:=\left\{
t\in\mathbb{Z};t\leq0\right\}  $, $\boldsymbol{j}=\left(  j_{1},\ldots
,j_{N}\right)  \in\left\{  1,\ldots,p-1\right\}  ^{N}$,
\[
\boldsymbol{n}=\left(  n_{1},\ldots,n_{N}\right)  \in\left(  p^{r_{1}%
}\mathbb{Z}_{p}\cap\mathbb{Q}_{p}/\mathbb{Z}_{p}\right)  \times\cdots
\times\left(  p^{r_{N}}\mathbb{Z}_{p}\cap\mathbb{Q}_{p}/\mathbb{Z}_{p}\right)
,
\]
$I\subseteq\left\{  1,\ldots,N\right\}  $, and
\begin{equation}
\Psi_{\boldsymbol{rnj}}^{I}\left(  x\right)  :=%
{\textstyle\prod\limits_{i\in I}}
\Psi_{r_{i}n_{i}j_{i}}\left(  x_{i}\right)
{\textstyle\prod\limits_{i\notin I}}
\Omega\left(  \left\vert x_{i}\right\vert _{p}\right)  \text{, for }x=\left(
x_{1},\ldots,x_{N}\right)  \in\mathbb{Z}_{p}^{N}. \label{Eigen-fun-1}%
\end{equation}
By convention $%
{\textstyle\prod\nolimits_{i\in\varnothing}}
\cdot=1$, then
\begin{equation}
\Psi_{\boldsymbol{rnj}}^{\left\{  1,\ldots,N\right\}  }\left(  x\right)  =%
{\textstyle\prod\limits_{i=1}^{N}}
\Psi_{r_{i}n_{i}j_{i}}\left(  x_{i}\right)  \text{, }\Psi^{\varnothing}\left(
x\right)  =%
{\textstyle\prod\limits_{i=1}^{N}}
\Omega\left(  \left\vert x_{i}\right\vert _{p}\right)  =\Omega\left(
\left\Vert x\right\Vert _{p}\right)  . \label{Eigen-fun-2}%
\end{equation}

By a well-known result, see, e.g., \cite[Chap. II, Proposition 2, Theorem
II.10-(a)]{Reed-Simon I}, the set $\left\{  \Psi_{\boldsymbol{rnj}}%
^{I}\right\}  $ is an orthonormal basis for $L^{2}(\mathbb{Z}_{p}^{N})$.

\begin{lemma}
\label{Lemma4}With the above notation, the following assertions hold true:

\begin{enumerate}
\item[(i)] if $I\neq\varnothing$, then $\square_{K}\Psi_{\boldsymbol{rnj}}%
^{I}\left(  x\right)  =\widehat{K}(p^{\lambda})\Psi_{\boldsymbol{rnj}}\left(
x\right)  $, where $\lambda=\max_{i\in I}\left\{  -r_{i}+1\right\}  $. The
multiplicity $mult(\lambda)$ of this eigenvalue satisfies $mult(\lambda
)\leq\left(  2\lambda p^{\lambda+1}\right)  ^{N}$.

\item[(ii)] if $I=\varnothing$, then $\square_{K}\Omega\left(  \left\Vert
x\right\Vert _{p}\right)  =\widehat{K}(\left[  0\right]  _{p})\Omega\left(
\left\Vert x\right\Vert _{p}\right)  $. The multiplicity of this eigenvalue is
$1$.
\end{enumerate}
\end{lemma}

\begin{proof}
(i) Take $I\neq\varnothing$. By using that%
\[
\widehat{\Psi}_{rnj}\left(  \xi\right)  =p^{\frac{r}{2}}\chi_{p}\left(
p^{-r}n\xi\right)  \Omega\left(  \left\vert p^{-r}\xi+p^{-1}j\right\vert
_{p}\right)  ,
\]
it follows that
\[
\widehat{\Psi}_{\boldsymbol{rnj}}^{I}\left(  \xi\right)  =p^{\frac{1}{2}%
\sum_{i\in I}r_{i}}%
{\textstyle\prod\limits_{i\in I}}
\chi_{p}\left(  p^{-r_{i}}n_{i}\xi_{i}\right)  \Omega\left(  \left\vert
p^{-r_{i}}\xi_{i}+p^{-1}j_{i}\right\vert _{p}\right)
{\textstyle\prod\limits_{i\notin I}}
\Omega\left(  \left\vert \xi_{i}\right\vert _{p}\right)  ,
\]
where $\xi=\left(  \xi_{1},\ldots,\xi_{N}\right)  \in\mathbb{Q}_{p}^{N}$ and
\[
\xi_{i}\in-p^{r_{i}-1}j_{i}+p^{r_{i}}\mathbb{Z}_{p}\text{, for }i\in I\text{
and }\xi_{i}\in\mathbb{Z}_{p}\text{, for }i\notin I.
\]
Then
\begin{align*}
\left[  \xi\right]  _{p}  &  =\max\left\{  1,\left\Vert \xi\right\Vert
_{p}\right\}  =\max\left\{  1,\max_{1\leq i\leq N}\left\vert \xi
_{i}\right\vert _{p}\right\}  =\max\left\{  1,\max_{i\in I}\left\vert \xi
_{i}\right\vert _{p}\right\} \\
&  =\max\left\{  1,\max_{i\in I}\left\vert -p^{r_{i}-1}j\right\vert
_{p}\right\}  =\max_{i\in I}p^{-r_{i}+1},
\end{align*}
and since $\square_{K}\Psi_{\boldsymbol{rnj}}^{I}\left(  x\right)
=\mathcal{F}_{\xi\rightarrow x}^{-1}\left(  \widehat{K}(\left[  \xi\right]
_{p})\widehat{\Psi}_{\boldsymbol{rnj}}\left(  \xi\right)  \right)  $, we
conclude that
\[
\square_{K}\Psi_{\boldsymbol{rnj}}^{I}\left(  x\right)  =\widehat{K}%
(\max_{i\in I}p^{-r_{i}+1})\Psi_{\boldsymbol{rnj}}^{I}\left(  x\right)  .
\]
We now estimate the multiplicity of the eigenvalue $\widehat{K}(\max_{i\in
I}p^{-r_{i}+1})$. We fix $p^{\lambda}$, $\lambda\in\mathbb{N}$, and\ count the
amount of \ the $I\neq\varnothing$, $\boldsymbol{r}$, $\boldsymbol{n}$,
$\boldsymbol{j}$ \ such that
\begin{equation}
\left\langle \square_{K}\Psi_{\boldsymbol{rnj}}^{I},\Psi_{\boldsymbol{rnj}%
}^{I}\right\rangle =\widehat{K}(p^{\lambda})\Psi_{\boldsymbol{rnj}}%
^{I}\text{.} \label{condition}%
\end{equation}
By using that $\lambda=$ $\max_{i\in I}\left\{  -r_{i}+1\right\}  $, we have
$0\leq-r_{i}\leq\lambda-1$, and thus the amount of the $r_{i}$, $i\in I$ is
bounded by $\lambda^{\#I}$. The amount of the $j_{i}$, $i\in I$ is $p^{\#I}$;
the amount of the $n_{i}$, $i\in I$ is bounded by
\[
p^{\sum_{i\in I}-r_{i}}\leq p^{\left(  \#I\right)  \lambda}\text{.}%
\]
Therefore, the amount of functions $\Psi_{\boldsymbol{rnj}}^{I}$ satisfying
(\ref{condition}) is bounded by
\[%
{\displaystyle\sum\limits_{I\neq\varnothing}}
\lambda^{\#I}p^{\#I}p^{\left(  \#I\right)  \lambda}\leq2^{N}\left(  \lambda
p^{\lambda+1}\right)  ^{N}=\left(  2\lambda p^{\lambda+1}\right)  ^{N}.
\]
(ii) The second part follows from the fact that $\square_{K}\Omega\left(
\left\Vert x\right\Vert _{p}\right)  =K(x)\ast\Omega\left(  \left\Vert
x\right\Vert _{p}\right)  =\widehat{K}(\left[  0\right]  _{p})$ for
$\left\Vert x\right\Vert _{p}\leq1$.
\end{proof}

\begin{notation}
From now on, we consider $\left(  L_{\mathbb{R}}^{2}\left(  \mathbb{Z}_{p}%
^{N}\right)  ,\mathcal{B},d^{N}x\right)  $ as a measurable space, where
$\mathcal{B}$\ is the Borel $\sigma$-algebra of $L_{\mathbb{R}}^{2}\left(
\mathbb{Z}_{p}^{N}\right)  $.
\end{notation}

\begin{theorem}
\label{Theorem1}Assume that
\begin{equation}
\widehat{K}(1)+%
{\displaystyle\sum\limits_{\lambda=1}^{\infty}}
\left(  2\lambda p^{\lambda+1}\right)  ^{N}\widehat{K}(p^{\lambda})<\infty.
\label{hypo1}%
\end{equation}
Then there exits a unique Gaussian probability measure $\mathbb{P}_{K}$ on
$\left(  L_{\mathbb{R}}^{2}\left(  \mathbb{Z}_{p}^{N}\right)  ,\mathcal{B}%
\right)  $ with mean zero, covariance $\square_{K}$, and Fourier transform
\[%
{\displaystyle\int\limits_{L_{\mathbb{R}}^{2}\left(  \mathbb{Z}_{p}%
^{N}\right)  }}
e^{\sqrt{-1}\left\langle f,w\right\rangle }d\mathbb{P}_{K}\left(  w\right)
=e^{\frac{-1}{2}\left\langle \square_{K}f,f\right\rangle }.
\]

\end{theorem}

\begin{proof}
By using Lemmas \ref{Lemma1}-\ref{Lemma4}, and under the condition
(\ref{hypo1}), $\square_{K}$ is a trace class operator. Indeed,%
\begin{align*}
Tr\left(  \square_{K}\right)   &  =\left\langle \square_{K}\Omega\left(
\left\Vert x\right\Vert _{p}\right)  ,\Omega\left(  \left\Vert x\right\Vert
_{p}\right)  \right\rangle +%
{\displaystyle\sum\limits_{\boldsymbol{r},\boldsymbol{n},\boldsymbol{j}%
,\boldsymbol{I}\neq\varnothing}}
\left\langle \square_{K}\Psi_{\boldsymbol{rnj}}^{I},\Psi_{\boldsymbol{rnj}%
}^{I}\right\rangle \\
&  \leq\widehat{K}(1)+%
{\displaystyle\sum\limits_{\lambda=1}^{\infty}}
\left(  2\lambda p^{\lambda+1}\right)  ^{N}\widehat{K}(p^{\lambda})<\infty
\end{align*}
Now the announced result follows from a well-know result about Gaussian
measures on Hilbert spaces, see \cite[Theorem 1.12]{Da prato}.
\end{proof}

\section{\label{Section_4}Continuous statistical field theories and deep
Boltzmann machines}

\subsection{Fluctuating fields and energy functionals}

We work with a $2$-component\ fluctuating field $\left\{  \boldsymbol{v}%
,\boldsymbol{h}\right\}  $ in the $N$-dimensional $p$-adic unit ball. The
realizations of the field are real-valued functions defined in $\mathbb{Z}%
_{p}^{N}$. The function $\boldsymbol{v}:\mathbb{Z}_{p}^{N}\rightarrow
\mathbb{R}$ is called the \textit{visible field} and the function
$\boldsymbol{h}:\mathbb{Z}_{p}^{N}\rightarrow\mathbb{R}$ is called the
\textit{hidden field}. These fields are used to model signals, or more
generally data. The discrete data take only a finite number of values. A basic
example is a black-and-white image. The continuous data take values in a
finite interval, for instance, the electrical voltages produced by a living organism.

This is the motivation to assume that the fields $\boldsymbol{v}%
,\boldsymbol{h}$ take values in a bounded subset of the real numbers. The
data/signals take values in a bounded set, which means $\left\Vert
\boldsymbol{v}\right\Vert _{\infty}\leq M$, $\left\Vert \boldsymbol{h}%
\right\Vert _{\infty}\leq M$, where $M$ is a fixed positive constant. Which in
turns implies that $\left\Vert \boldsymbol{v}\right\Vert _{2}\leq M$,
$\left\Vert \boldsymbol{h}\right\Vert _{2}\leq M$. For this reason, there are
two different forms of choosing the fields:%
\[
\boldsymbol{v},\boldsymbol{h}\in\boldsymbol{B}_{M}^{\left(  \infty\right)
}:=\left\{  f\in L_{\mathbb{R}}^{\infty}(\mathbb{Z}_{p}^{N});\left\Vert
f\right\Vert _{\infty}\leq M\right\}  ,
\]
or%
\[
\boldsymbol{v},\boldsymbol{h}\in\boldsymbol{B}_{M}^{\left(  2\right)
}:=\left\{  f\in L_{\mathbb{R}}^{2}(\mathbb{Z}_{p}^{N});\left\Vert
f\right\Vert _{2}\leq M\right\}  .
\]
notice that $\boldsymbol{B}_{M}^{\left(  \infty\right)  }\subset
\boldsymbol{B}_{M}^{\left(  2\right)  }$, and that $\frac{1}{\left\Vert
x\right\Vert _{p}^{\alpha}}$, for $\alpha\in\left(  0,\frac{N}{2}\right)  $,
is an element from $\boldsymbol{B}_{M}^{\left(  2\right)  }$, but $\sup
_{x\in\mathbb{Z}_{p}^{N}}\frac{1}{\left\Vert x\right\Vert _{p}^{\alpha}%
}=\infty$.

If the fields are from $\boldsymbol{B}_{M}^{\left(  2\right)  }$, the
parameters of our SFTs satisfy
\begin{equation}
a\left(  x\right)  ,b(x),c(x),d(x)\in L_{\mathbb{R}}^{\infty}(\mathbb{Z}%
_{p}^{N}),w\left(  x,y\right)  \in L_{\mathbb{R}}^{\infty}(\mathbb{Z}_{p}%
^{N}\times\mathbb{Z}_{p}^{N}), \label{condition parameters}%
\end{equation}
and $M>0$. If the fields are from \ $\boldsymbol{B}_{M}^{\left(  2\right)  }$,
we impose an additional condition:
\begin{equation}
c(x),d(x)\geq0. \label{additional condition}%
\end{equation}
The field $\left\{  \boldsymbol{v},\boldsymbol{h}\right\}  $ performs
fluctuations, and the size of these fluctuations is controlled by an energy
functional consisting of two terms:%
\begin{equation}
E(\boldsymbol{v},\boldsymbol{h};\boldsymbol{\theta}):=E(\boldsymbol{v}%
,\boldsymbol{h})=E^{\text{free}}\left(  \boldsymbol{v},\boldsymbol{h}\right)
+E^{\text{int}}\left(  \boldsymbol{v},\boldsymbol{h}\right)  ,
\label{energy_functional}%
\end{equation}
where $\boldsymbol{\theta}=\left(  w,a,b,c,d\right)  $. The first term%
\[
E^{\text{free}}\left(  \boldsymbol{v},\boldsymbol{h}\right)  =-\left\langle
a,\boldsymbol{v}\right\rangle -\left\langle b,\boldsymbol{h}\right\rangle
\]
is an analogue of the \textit{free-field energy}. The second term%
\begin{gather}
E^{\text{int}}\left(  \boldsymbol{v},\boldsymbol{h}\right)  =-%
{\displaystyle\iint\limits_{\mathbb{Z}_{p}^{N}\times\mathbb{Z}_{p}^{N}}}
\boldsymbol{h}\left(  x\right)  w\left(  x,y\right)  \boldsymbol{v}\left(
y\right)  d^{N}yd^{N}x+\left\langle c,\boldsymbol{v}^{4}\right\rangle
+\left\langle d,\boldsymbol{h}^{4}\right\rangle \nonumber\\
=-\left\langle \boldsymbol{h},\boldsymbol{Wv}\right\rangle +\left\langle
c,\boldsymbol{v}^{4}\right\rangle +\left\langle d,\boldsymbol{h}%
^{4}\right\rangle , \label{Interaction-energy}%
\end{gather}
is an analogue of the \textit{interaction energy}. Here $\boldsymbol{Wv}%
\left(  x\right)  :=\int_{\mathbb{Z}_{p}^{N}}w\left(  x,y\right)
\boldsymbol{v}\left(  y\right)  d^{N}y$.

\begin{lemma}
\label{Lemma5}Let $K_{1}$, $K_{2}$ be two kernels satisfying (\ref{hypo1}). We
denote by $\mathbb{P}_{K_{1}}\otimes\mathbb{P}_{K_{2}}$ the product
probability measure defined on the product $\sigma$-algebra $\mathcal{B}%
\times\mathcal{B}$. If the parameters of the theory satisfy
(\ref{condition parameters}), then $\left\vert -E(\boldsymbol{v}%
,\boldsymbol{h})\right\vert \leq C_{0}\left(  a,b,c,d,M\right)  $ and
\begin{equation}
\mathcal{Z}^{\left(  \infty\right)  }:=\mathcal{Z}^{\left(  \infty\right)
}\left(  \boldsymbol{\theta}\right)  =%
{\displaystyle\iint\limits_{\boldsymbol{B}_{M}^{\left(  \infty\right)  }%
\times\boldsymbol{B}_{M}^{\left(  \infty\right)  }}}
\exp\left(  -E(\boldsymbol{v},\boldsymbol{h})\right)  d\mathbb{P}_{K_{1}%
}\left(  \boldsymbol{v}\right)  \otimes d\mathbb{P}_{K_{2}}\left(
\boldsymbol{h}\right)  <\infty. \label{case1}%
\end{equation}
If the parameters of the theory satisfy (\ref{condition parameters}) and
(\ref{additional condition}), then $\left\vert -E(\boldsymbol{v}%
,\boldsymbol{h})\right\vert \leq C_{1}\left(  a,b,M\right)  $ and%
\begin{equation}
\mathcal{Z}^{\left(  2\right)  }:=\mathcal{Z}^{\left(  2\right)  }\left(
\boldsymbol{\theta}\right)  =%
{\displaystyle\iint\limits_{\boldsymbol{B}_{M}^{\left(  2\right)  }%
\times\boldsymbol{B}_{M}^{\left(  2\right)  }}}
\exp\left(  -E(\boldsymbol{v},\boldsymbol{h})\right)  d\mathbb{P}_{K_{1}%
}\left(  \boldsymbol{v}\right)  \otimes d\mathbb{P}_{K_{2}}\left(
\boldsymbol{h}\right)  <\infty. \label{case2}%
\end{equation}

\end{lemma}

\begin{proof}
If the fields are from $\boldsymbol{B}_{M}^{\left(  \infty\right)  }$, the
result follows from the following estimate:%
\[
\exp\left\vert -E(\boldsymbol{v},\boldsymbol{h})\right\vert \leq\exp\left\{
\left\Vert a\right\Vert _{\infty}M+\left\Vert b\right\Vert _{\infty
}M+\left\Vert w\right\Vert _{\infty}M^{2}+\left\Vert c\right\Vert _{\infty
}M^{4}+\left\Vert d\right\Vert _{\infty}M^{4}\right\}  .
\]

If the fields are from $\boldsymbol{B}_{M}^{\left(  2\right)  }$, the result
follows from the following estimate:%
\begin{equation}
\exp\left\vert -E(\boldsymbol{v},\boldsymbol{h})\right\vert \leq\exp\left\{
\left\Vert a\right\Vert _{\infty}M+\left\Vert b\right\Vert _{\infty
}M+\left\Vert w\right\Vert _{\infty}M^{2}\right\}  . \label{Bound}%
\end{equation}
Indeed, the case (\ref{case2}) follows from (\ref{Bound}), by using that
\[
\exp\left\{  -%
{\displaystyle\int\limits_{\mathbb{Z}_{p}^{N}}}
c(x)\boldsymbol{v}^{4}\left(  x\right)  d^{N}x-%
{\displaystyle\int\limits_{\mathbb{Z}_{p}^{N}}}
d(x)\boldsymbol{h}^{4}\left(  x\right)  d^{N}x\right\}  \leq1.
\]
To establish (\ref{Bound}), we proceed as follows. By using the
Cauchy--Schwarz inequality, and the fact that$\int_{\mathbb{Z}_{p}^{N}}%
d^{N}x=1$,
\[
\left\vert \text{ }%
{\displaystyle\int\limits_{\mathbb{Z}_{p}^{N}}}
a(x)\boldsymbol{v}\left(  x\right)  d^{N}x\right\vert \leq\left\Vert
a\right\Vert _{\infty}%
{\displaystyle\int\limits_{\mathbb{Z}_{p}^{N}}}
\left\vert \boldsymbol{v}\left(  x\right)  \right\vert d^{N}x\leq\left\Vert
a\right\Vert _{\infty}\left\Vert \boldsymbol{v}\right\Vert _{2}\leq\left\Vert
a\right\Vert _{\infty}M.
\]
In a similar way, one shows that
\[
\left\vert \text{ }%
{\displaystyle\int\limits_{\mathbb{Z}_{p}^{N}}}
b(x)\boldsymbol{h}\left(  x\right)  d^{N}x\right\vert \leq\left\Vert
b\right\Vert _{\infty}M.
\]
Now, by using that $\left\Vert f\right\Vert _{1}\leq\left\Vert f\right\Vert
_{2}$,%
\begin{gather*}
\left\vert \text{ \ }%
{\displaystyle\iint\limits_{\mathbb{Z}_{p}^{N}\times\mathbb{Z}_{p}^{N}}}
\boldsymbol{h}\left(  x\right)  w\left(  x,y\right)  \boldsymbol{v}\left(
y\right)  d^{N}yd^{N}x\right\vert \leq\left\Vert w\right\Vert _{\infty}%
{\displaystyle\int\limits_{\mathbb{Z}_{p}^{N}}}
\left\vert \boldsymbol{v}\left(  y\right)  \right\vert d^{N}y%
{\displaystyle\int\limits_{\mathbb{Z}_{p}^{N}}}
\left\vert \boldsymbol{h}\left(  x\right)  \right\vert d^{N}x\\
\leq\left\Vert w\right\Vert _{\infty}\left\Vert \boldsymbol{v}\right\Vert
_{1}\left\Vert \boldsymbol{h}\right\Vert _{1}\leq\left\Vert w\right\Vert
_{\infty}\left\Vert \boldsymbol{v}\right\Vert _{2}\left\Vert \boldsymbol{h}%
\right\Vert _{2}\leq\left\Vert w\right\Vert _{\infty}M^{2}.
\end{gather*}

\end{proof}

\subsection{Partition functions and continuous BMs}

Assume that the fields are from $\boldsymbol{B}_{M}^{\left(  2\right)  }$. All
he thermodynamic properties of the system are described\ by the partition
function of the fluctuating fields $\mathcal{Z}^{\left(  2\right)
}=\mathcal{Z}^{\left(  2\right)  }\left(  \boldsymbol{\theta}\right)  $. We
identify the statistical field theory (SFT) corresponding to the energy
functional (\ref{energy_functional}) with the probability measure%
\[
\mathbb{P}(\boldsymbol{v},\boldsymbol{h};\boldsymbol{\theta})=\frac
{\exp\left(  -E(\boldsymbol{v},\boldsymbol{h})\right)  }{\mathcal{Z}^{\left(
2\right)  }\left(  \boldsymbol{\theta}\right)  }\mathbb{P}_{K_{1}}\left(
\boldsymbol{v}\right)  \otimes\mathbb{P}_{K_{2}}\left(  \boldsymbol{h}\right)
\]
on the $\sigma$-algebra $\mathcal{B}\times\mathcal{B}$, where
$\boldsymbol{\theta}=\left(  w,a,b,c,d\right)  $. We do not include the
parameter $M$, a fixed positive number, in the vector $\boldsymbol{\theta}$.
The parameters in the list $\boldsymbol{\theta}$ are tuned during the learning
process, but parameter $M$ remains fixed.

We attach to $\mathbb{P}(\boldsymbol{v},\boldsymbol{h};\boldsymbol{\theta})$
\textit{a }$p$\textit{-adic continuous deep Boltzmann machine (DBM)}. In this
way we have, by definition, a one-to-one correspondence between SFTs and DBMs.

\begin{remark}
(i) $\mathbb{P}(\boldsymbol{v},\boldsymbol{h};\boldsymbol{\theta})$ is a
generative continuous model to explain the connections between of types of
data ( $\boldsymbol{v}$ and $\boldsymbol{h}$).

(ii) If the fields are from $\boldsymbol{B}_{M}^{\left(  \infty\right)  }$,
then the partition function is$\mathcal{Z}^{\left(  \infty\right)
}=\mathcal{Z}^{\left(  \infty\right)  }\left(  \boldsymbol{\theta}\right)  $,
and the corresponding SFT is identified with the probability measure%
\[
\mathbb{P}(\boldsymbol{v},\boldsymbol{h};\boldsymbol{\theta})=\frac
{\exp\left(  -E(\boldsymbol{v},\boldsymbol{h})\right)  }{\mathcal{Z}^{\left(
\infty\right)  }\left(  \boldsymbol{\theta}\right)  }\mathbb{P}_{K_{1}}\left(
\boldsymbol{v}\right)  \otimes\mathbb{P}_{K_{2}}\left(  \boldsymbol{h}\right)
\]
on the $\sigma$-algebra $\mathcal{B}\times\mathcal{B}$, where
$\boldsymbol{\theta}=\left(  w,a,b,c,d\right)  $.
\end{remark}

\section{\label{Section_5}Discretization of the energy functional}

For $l\geq1$, we \ set $G_{l}^{N}:=\left(  \mathbb{Z}_{p}/p^{l}\mathbb{Z}%
_{p}\right)  ^{N}$. The set $G_{l}^{1}=G_{l}$ is a finite rooted tree with $l$
levels, thus in the case $N\geq2$, $G_{l}^{N}$ is a finite forest consisting
of $N$ trees.

We denote \ the elements of $G_{l}^{N}$ as $\boldsymbol{i}=\left(
i_{1},\ldots,i_{N}\right)  $, where
\[
i_{k}=i_{0}^{k}+i_{1}^{k}p+\ldots+i_{l-1}^{k}p^{l-1},\text{ for }%
k=1,\ldots,N\text{,}%
\]
here the $i_{j}^{k}$\ are $p$-adic digits. For $\boldsymbol{i}=\left(
i_{1},\ldots,i_{N}\right)  \in G_{l}^{N}$, we set
\[
\Omega\left(  p^{l}\left\Vert x-\boldsymbol{i}\right\Vert _{p}\right)  =%
{\textstyle\prod\limits_{k=1}^{N}}
\Omega\left(  p^{l}\left\vert x-i_{k}\right\vert _{p}\right)  .
\]
The function $\Omega\left(  p^{l}\left\vert x-i_{k}\right\vert _{p}\right)  $
is the characteristic function of the ball $i_{k}+p^{l}\mathbb{Z}_{p}$, and
$\Omega\left(  p^{l}\left\Vert x-\boldsymbol{i}\right\Vert _{p}\right)  $ is
the characteristic function of the ball
\[
\left(  i_{1}+p^{l}\mathbb{Z}_{p}\right)  \times\cdots\times\left(
i_{N}+p^{l}\mathbb{Z}_{p}\right)  .
\]
We denote by $\mathcal{D}^{l}(\mathbb{Z}_{p}^{N})$ the $\mathbb{R}$-vector
space of all test functions of the form%
\begin{equation}
\varphi\left(  x\right)  =%
{\textstyle\sum\limits_{\boldsymbol{i}\in G_{l}^{N}}}
\varphi\left(  \boldsymbol{i}\right)  \Omega\left(  p^{l}\left\Vert
x-\boldsymbol{i}\right\Vert _{p}\right)  \text{, \ }\varphi\left(
\boldsymbol{i}\right)  \in\mathbb{R}\text{.} \label{Eq_repre}%
\end{equation}
The function $\varphi$ is supported on $\mathbb{Z}_{p}^{N}$ and $\mathcal{D}%
^{l}(\mathbb{Z}_{p}^{N})$ is a finite dimensional vector space spanned by the
basis
\begin{equation}
\left\{  \Omega\left(  p^{l}\left\Vert x-\boldsymbol{i}\right\Vert
_{p}\right)  \right\}  _{\boldsymbol{i}\in G_{l}^{N}}. \label{Basis}%
\end{equation}
The identification of $\varphi\in\mathcal{D}^{l}(\mathbb{Z}_{p}^{N})$ with the
column vector $\left[  \varphi\left(  \boldsymbol{i}\right)  \right]
_{\boldsymbol{i}\in G_{l}^{N}}\in\mathbb{R}^{\left(  \#G_{l}\right)  N}$ gives
rise to an isomorphism between $\mathcal{D}^{l}(\mathbb{Z}_{p}^{N})$ and
$\mathbb{R}^{\left(  \#G_{l}\right)  N}$ endowed with the norm $\left\Vert
\left[  \varphi\left(  \boldsymbol{i}\right)  \right]  _{\boldsymbol{i}\in
G_{l}^{N}}\right\Vert =\max_{\boldsymbol{i}\in G_{l}^{N}}\left\vert
\varphi\left(  \boldsymbol{i}\right)  \right\vert $. Furthermore,
\[
\mathcal{D}^{l}(\mathbb{Z}_{p}^{N})\hookrightarrow\mathcal{D}^{l+1}%
(\mathbb{Z}_{p}^{N})\hookrightarrow\mathcal{D}(\mathbb{Z}_{p}^{N}),
\]
where $\hookrightarrow$ denotes a continuous embedding, and $\mathcal{D}%
(\mathbb{Z}_{p}^{N})=\cup_{l}\mathcal{D}^{l}(\mathbb{Z}_{p}^{N})$.

The space $\mathcal{D}(\mathbb{Z}_{p}^{N})$ is dense in $L^{2}(\mathbb{Z}%
_{p}^{N})$, thus given $f\in L^{2}(\mathbb{Z}_{p}^{N})$ and $\epsilon>0$,
there exist a positive integer $l$, and $\phi\in\mathcal{D}^{l}(\mathbb{Z}%
_{p}^{N})$, such that $\left\Vert f-\phi\right\Vert _{2}<\epsilon$. A
discretization $E_{l}$ of the energy functional $E$ is obtained by restricting
$\boldsymbol{v},\boldsymbol{h}$ to $\mathcal{D}^{l}(\mathbb{Z}_{p}^{N})$,
i.e., by taking
\[
\boldsymbol{v}\left(  y\right)  =%
{\textstyle\sum\limits_{\boldsymbol{j}\in G_{l}^{N}}}
\boldsymbol{v}\left(  \boldsymbol{j}\right)  \Omega\left(  p^{l}\left\Vert
y-\boldsymbol{j}\right\Vert _{p}\right)  \text{, \ }\boldsymbol{h}\left(
x\right)  =%
{\textstyle\sum\limits_{\boldsymbol{i}\in G_{l}^{N}}}
\boldsymbol{h}\left(  \boldsymbol{i}\right)  \Omega\left(  p^{l}\left\Vert
x-\boldsymbol{i}\right\Vert _{p}\right)  .
\]
We now set%
\begin{align*}
w(\boldsymbol{i},\boldsymbol{j})  &  :=%
{\displaystyle\iint\limits_{\mathbb{Z}_{p}^{N}\times\mathbb{Z}_{p}^{N}}}
w(x,y)\Omega\left(  p^{l}\left\Vert x-\boldsymbol{i}\right\Vert _{p}\right)
\Omega\left(  p^{l}\left\Vert y-\boldsymbol{j}\right\Vert _{p}\right)
d^{N}xd^{N}y,\\
a(\boldsymbol{i})  &  :=%
{\displaystyle\int\limits_{\mathbb{Z}_{p}^{N}}}
a(x)\Omega\left(  p^{l}\left\Vert x-\boldsymbol{i}\right\Vert _{p}\right)
d^{n}x\text{, }b(\boldsymbol{i}):=%
{\displaystyle\int\limits_{\mathbb{Z}_{p}^{N}}}
b(x)\Omega\left(  p^{l}\left\Vert x-\boldsymbol{i}\right\Vert _{p}\right)
d^{n}x,\\
c(\boldsymbol{i})  &  :=%
{\displaystyle\int\limits_{\mathbb{Z}_{p}^{N}}}
c(x)\Omega\left(  p^{l}\left\Vert x-\boldsymbol{i}\right\Vert _{p}\right)
d^{n}x\text{, }d(\boldsymbol{i}):=%
{\displaystyle\int\limits_{\mathbb{Z}_{p}^{N}}}
d(x)\Omega\left(  p^{l}\left\Vert x-\boldsymbol{i}\right\Vert _{p}\right)
d^{n}x.
\end{align*}
Then $E_{l}\left(  \boldsymbol{v},\boldsymbol{h};\boldsymbol{\theta}\right)
=E_{l}^{\text{free}}\left(  \boldsymbol{v},\boldsymbol{h};\boldsymbol{\theta
}\right)  +E_{l}^{\text{int}}\left(  \boldsymbol{v},\boldsymbol{h}%
;\boldsymbol{\theta}\right)  $, for\ $\boldsymbol{v},\boldsymbol{h}%
\in\mathcal{D}^{l}(\mathbb{Z}_{p}^{N})$, where%
\[
E_{l}^{\text{free}}\left(  \boldsymbol{v},\boldsymbol{h}\right)  =-%
{\displaystyle\sum\limits_{\boldsymbol{i}\in G_{l}^{N}}}
a(\boldsymbol{i})\boldsymbol{v}\left(  \boldsymbol{i}\right)  -%
{\displaystyle\sum\limits_{\boldsymbol{i}\in G_{l}^{N}}}
b(\boldsymbol{i})\boldsymbol{h}\left(  \boldsymbol{i}\right)  ,
\]%
\begin{align*}
E_{l}^{\text{int}}\left(  \boldsymbol{v},\boldsymbol{h}\right)   &  =-%
{\displaystyle\sum\limits_{\boldsymbol{i},\boldsymbol{j}\in G_{l}^{N}}}
\boldsymbol{h}\left(  \boldsymbol{i}\right)  w\left(  \boldsymbol{i}%
,\boldsymbol{j}\right)  \boldsymbol{v}\left(  \boldsymbol{j}\right)  +\\
&
{\displaystyle\sum\limits_{\boldsymbol{i}\in G_{l}^{N}}}
c(\boldsymbol{i})\boldsymbol{v}^{4}\left(  \boldsymbol{i}\right)  +%
{\displaystyle\sum\limits_{\boldsymbol{i}\in G_{l}^{N}}}
d(\boldsymbol{i})\boldsymbol{h}^{4}\left(  \boldsymbol{i}\right)  .
\end{align*}
By taking
\begin{align*}
v_{\boldsymbol{i}}  &  :=\boldsymbol{v}\left(  \boldsymbol{i}\right)  \text{,
}h_{\boldsymbol{i}}:=\boldsymbol{h}\left(  \boldsymbol{i}\right)  \text{,
}w_{\boldsymbol{i},\boldsymbol{j}}:=w\left(  \boldsymbol{i},\boldsymbol{j}%
\right)  \text{, \ }a_{\boldsymbol{i}}:=a(\boldsymbol{i})\text{,
}b_{\boldsymbol{i}}:=b(\boldsymbol{i})\text{,}\\
c_{\boldsymbol{i}}  &  :=c(\boldsymbol{i})\text{, }d_{\boldsymbol{i}%
}:=d(\boldsymbol{i})\text{,}%
\end{align*}
and%
\begin{align*}
\boldsymbol{v}_{l}  &  \boldsymbol{=}\left[  v_{\boldsymbol{i}}\right]
_{\boldsymbol{i}\in G_{l}^{N}}\text{, }\boldsymbol{h}_{l}\boldsymbol{=}\left[
h_{\boldsymbol{i}}^{l}\right]  _{\boldsymbol{i}\in G_{l}^{N}}\text{,
}\boldsymbol{w}_{l}=\left[  w_{\boldsymbol{i},\boldsymbol{j}}\right]
_{\boldsymbol{i},\boldsymbol{j}\in G_{l}^{N}}\text{, }\boldsymbol{a}%
_{l}=\left[  a_{\boldsymbol{i}}\right]  _{\boldsymbol{i}\in G_{l}^{N}}\text{,
}\\
\boldsymbol{b}_{l}  &  =\left[  b_{\boldsymbol{i}}\right]  _{\boldsymbol{i}\in
G_{l}^{N}}\text{, }\boldsymbol{c}_{l}=\left[  c_{\boldsymbol{i}}\right]
_{\boldsymbol{i}\in G_{l}^{N}}\text{, }\boldsymbol{d}_{l}=\left[
d_{\boldsymbol{i}}\right]  _{\boldsymbol{i}\in G_{l}^{N}}\text{,}%
\end{align*}
and identifying $\boldsymbol{v}$ with $\boldsymbol{v}_{l}$, and
$\boldsymbol{h}$ with $\boldsymbol{h}_{l}$, we have%
\begin{equation}
E_{l}^{\text{free}}\left(  \boldsymbol{v}_{l},\boldsymbol{h}_{l}%
;\boldsymbol{\theta}_{l}\right)  =-%
{\displaystyle\sum\limits_{\boldsymbol{i}\in G_{l}^{N}}}
a_{\boldsymbol{i}}v_{\boldsymbol{i}}-%
{\displaystyle\sum\limits_{\boldsymbol{i}\in G_{l}^{N}}}
b_{\boldsymbol{i}}h_{\boldsymbol{i}}, \label{Energy_free}%
\end{equation}%
\[
E_{l}^{\text{int}}\left(  \boldsymbol{v}_{l},\boldsymbol{h}_{l}%
;\boldsymbol{\theta}_{l}\right)  =-%
{\displaystyle\sum\limits_{\boldsymbol{i},\boldsymbol{j}\in G_{l}^{N}}}
h_{\boldsymbol{i}}w_{\boldsymbol{i},\boldsymbol{j}}v_{\boldsymbol{j}}+%
{\displaystyle\sum\limits_{\boldsymbol{i}\in G_{l}^{N}}}
c_{\boldsymbol{i}}v_{\boldsymbol{i}}^{4}+%
{\displaystyle\sum\limits_{\boldsymbol{i}\in G_{l}^{N}}}
d_{\boldsymbol{i}}h_{\boldsymbol{i}}^{4},
\]
where $\boldsymbol{\theta}_{l}=\left(  \boldsymbol{w}_{l},\boldsymbol{a}%
_{l},\boldsymbol{b}_{l},\boldsymbol{c}_{l},\boldsymbol{d}_{l},m_{l},\right)
$. We set
\[
E_{l}\left(  \boldsymbol{v}_{l},\boldsymbol{h}_{l};\boldsymbol{\theta}%
_{l}\right)  =E_{l}^{\text{free}}\left(  \boldsymbol{v}_{l},\boldsymbol{h}%
_{l};\boldsymbol{\theta}_{l}\right)  +E_{l}^{\text{int}}\left(  \boldsymbol{v}%
_{l},\boldsymbol{h}_{l};\boldsymbol{\theta}_{l}\right)  .
\]

\section{\label{Section_6}Discrete statistical field theories and deep
Boltzmann machines}

\subsection{Covariance matrices}

Given a positive integer $l$, we define the covariance matrix $C_{l}=\left[
C_{\boldsymbol{i},\boldsymbol{j}}\right]  _{\boldsymbol{i},\boldsymbol{j}\in
G_{l}^{N}}$, where%
\begin{gather*}
C_{\boldsymbol{i},\boldsymbol{j}}=\left\langle \square_{K}\text{ }p^{\frac
{lN}{2}}\Omega\left(  p^{l}\left\Vert x-\boldsymbol{i}\right\Vert _{p}\right)
,p^{\frac{lN}{2}}\Omega\left(  p^{l}\left\Vert x-\boldsymbol{j}\right\Vert
_{p}\right)  \right\rangle \\
=\left\langle \mathcal{F}_{x\rightarrow\xi}\left\{  \square_{K}\text{
}p^{\frac{lN}{2}}\Omega\left(  p^{l}\left\Vert x-\boldsymbol{i}\right\Vert
_{p}\right)  \right\}  ,\mathcal{F}_{x\rightarrow\xi}\left\{  p^{\frac{lN}{2}%
}\Omega\left(  p^{l}\left\Vert x-\boldsymbol{j}\right\Vert _{p}\right)
\right\}  \right\rangle .
\end{gather*}
Now by using%
\[
\mathcal{F}_{x\rightarrow\xi}\left(  p^{\frac{lN}{2}}\Omega\left(
p^{l}\left\Vert y-\boldsymbol{i}\right\Vert _{p}\right)  \right)
=p^{\frac{-lN}{2}}\chi_{p}\left(  \boldsymbol{i}\cdot\xi\right)  \Omega\left(
p^{-l}\left\Vert \xi\right\Vert _{p}\right)  ,
\]
we have%
\begin{align}
C_{\boldsymbol{i},\boldsymbol{j}}  &  =p^{-lN}%
{\displaystyle\int\limits_{\mathbb{Q}_{p}^{N}}}
\widehat{K}\left(  \xi\right)  \chi_{p}\left(  \boldsymbol{i}\cdot\xi\right)
\chi_{p}\left(  -\boldsymbol{j}\cdot\xi\right)  \Omega\left(  p^{-l}\left\Vert
\xi\right\Vert _{p}\right)  d^{N}\xi\nonumber\\
&  =p^{-lN}%
{\displaystyle\int\limits_{p^{-l}\mathbb{Z}_{p}^{N}}}
\widehat{K}\left(  \xi\right)  \chi_{p}\left(  \boldsymbol{i}\cdot\xi\right)
\chi_{p}\left(  -\boldsymbol{j}\cdot\xi\right)  d^{N}\xi. \label{Cov-0}%
\end{align}
By \ using the partition%
\[
p^{-l}\mathbb{Z}_{p}^{N}=%
{\displaystyle\bigsqcup\limits_{\boldsymbol{s}\in T_{-l}^{N}}}
\left(  \boldsymbol{s}+\mathbb{Z}_{p}^{N}\right)  ,
\]
where $T_{-l}^{N}$ is a set of representatives of the quotient group $\left(
p^{-l}\mathbb{Z}_{p}/\mathbb{Z}_{p}\right)  ^{N}$, the entry
$C_{\boldsymbol{i},\boldsymbol{j}}$ in (\ref{Cov-0}) can be rewritten as%
\begin{equation}
C_{\boldsymbol{i},\boldsymbol{j}}=p^{-lN}%
{\displaystyle\sum\limits_{\boldsymbol{s}\in T_{-l}^{N}}}
\text{ \ }%
{\displaystyle\int\limits_{\boldsymbol{s}+\mathbb{Z}_{p}^{N}}}
\widehat{K}\left(  \xi\right)  \chi_{p}\left(  \boldsymbol{i}\cdot\xi\right)
\chi_{p}\left(  -\boldsymbol{j}\cdot\xi\right)  d^{N}\xi. \label{Cov}%
\end{equation}
We recall that $\widehat{K}\left(  \xi\right)  =\widehat{K}\left(
\max\left\{  1,\left\Vert \xi\right\Vert _{p}\right\}  \right)  $, and thus
$\widehat{K}\left(  \xi+\xi_{0}\right)  =\widehat{K}\left(  \xi\right)  $ for
any $\xi\in\mathbb{Q}_{p}^{N}$, $\xi_{0}\in\mathbb{Z}_{p}^{N}$. This implies
that (\ref{Cov}) can be rewritten as%
\begin{align*}
C_{\boldsymbol{i},\boldsymbol{j}}  &  =p^{-lN}%
{\displaystyle\sum\limits_{\boldsymbol{s}\in T_{-l}^{N}}}
\text{ }\widehat{K}\left(  \boldsymbol{s}\right)  \chi_{p}\left(  \left(
\boldsymbol{i}-\boldsymbol{j}\right)  \cdot\boldsymbol{s}\right)
{\displaystyle\int\limits_{\mathbb{Z}_{p}^{N}}}
\chi_{p}\left(  \left(  \boldsymbol{i}-\boldsymbol{j}\right)  \cdot\xi\right)
d^{N}\xi\\
&  =p^{-lN}%
{\displaystyle\sum\limits_{\boldsymbol{s}\in T_{-l}^{N}}}
\text{ }\widehat{K}\left(  \boldsymbol{s}\right)  \chi_{p}\left(  \left(
\boldsymbol{i}-\boldsymbol{j}\right)  \cdot\boldsymbol{s}\right)  ,
\end{align*}
since $\left(  \boldsymbol{i}-\boldsymbol{j}\right)  \cdot\xi\in\mathbb{Z}%
_{p}$ for any $\xi\in\mathbb{Z}_{p}^{N}$. Using that $\widehat{K}\left(
\xi\right)  =\widehat{K}\left(  -\xi\right)  $, we conclude that $C_{l}$ is a
symmetric, real-valued matrix.

Finally, since $\square_{K}$ is a positive operator on $L_{\mathbb{R}}%
^{2}\left(  \mathbb{Z}_{p}^{N}\right)  $, the condition $\left\langle
\square_{K}f,f\right\rangle =0$ implies that $f=0$ almost everywhere. Take
$\varphi\in\mathcal{D}^{l}(\mathbb{Z}_{p}^{N})$, by identifying it with the
column vector $\left[  \varphi\left(  \boldsymbol{i}\right)  \right]
_{\boldsymbol{i}\in G_{l}^{N}}$, we have
\[
\left[  \varphi\left(  \boldsymbol{i}\right)  \right]  _{\boldsymbol{i}\in
G_{l}^{N}}C_{l}\left[  \varphi\left(  \boldsymbol{i}\right)  \right]
_{\boldsymbol{i}\in G_{l}^{N}}=\left\langle \square_{K}\varphi,\varphi
\right\rangle \geq0.
\]
The equality implies that $\varphi=0$ almost everywhere, and due to the local
constancy necessarily $\varphi$ is the constant function zero, i.e.,
$\varphi\left(  \boldsymbol{i}\right)  =0$ for any $\boldsymbol{i}\in
G_{l}^{N}$. Which implies that $C_{l}$\ is positive-definite. In conclusion,
we have the following result.

\begin{lemma}
$C_{l}$ is a real-valued, symmetric and positive-definite matrix.
\end{lemma}

\subsection{Discretization of $\mathbb{P}_{K_{1}}\otimes\mathbb{P}_{K_{2}}$}

On $L_{\mathbb{R}}^{2}\left(  \mathbb{Z}_{p}^{N}\right)  \times L_{\mathbb{R}%
}^{2}\left(  \mathbb{Z}_{p}^{N}\right)  $ we define the bilinear form%
\[
\left\langle \left(  f_{1},g_{1}\right)  ,\left(  f_{2},g_{2}\right)
\right\rangle =\left\langle f_{1},g_{1}\right\rangle +\left\langle f_{2}%
,g_{2}\right\rangle \in\mathbb{R}\text{.}%
\]
$L_{\mathbb{R}}^{2}\left(  \mathbb{Z}_{p}^{N}\right)  \times L_{\mathbb{R}%
}^{2}\left(  \mathbb{Z}_{p}^{N}\right)  $ endowed with this bilinear form is a
real Hilbert space. We also define $\square_{K_{1},K_{2}}\left(  f,g\right)
=\left(  \square_{K_{1}}f,\square_{K_{2}}g\right)  $, for $\left(  f,g\right)
\in$ $L_{\mathbb{R}}^{2}\left(  \mathbb{Z}_{p}^{N}\right)  \times
L_{\mathbb{R}}^{2}\left(  \mathbb{Z}_{p}^{N}\right)  $. By Lemmas
\ref{Lemma1}-\ref{Lemma4}, $\square_{K_{1},K_{2}}$ is a trace class operator.
By Theorem \ref{Theorem1}, the Fourier transform of $\mathbb{P}_{K_{1}}%
\otimes\mathbb{P}_{K_{2}}$ is
\[%
{\displaystyle\iint\limits_{L_{\mathbb{R}}^{2}\left(  \mathbb{Z}_{p}%
^{N}\right)  \times L_{\mathbb{R}}^{2}\left(  \mathbb{Z}_{p}^{N}\right)  }}
e^{\sqrt{-1}\left\{  \left\langle f,\boldsymbol{v}\right\rangle +\left\langle
g,\boldsymbol{h}\right\rangle \right\}  }d\mathbb{P}_{K_{1}}\left(
\boldsymbol{v}\right)  \otimes d\mathbb{P}_{K_{2}}\left(  \boldsymbol{v}%
\right)  =e^{\frac{-1}{2}\left\langle \square_{K_{1},K_{2}}\left(  f,g\right)
,\left(  f,g\right)  \right\rangle }.
\]
Which means that $\mathbb{P}_{K_{1}}\otimes\mathbb{P}_{K_{2}}$ is a Gaussian
measure with mean zero and covariance $\square_{K_{1},K_{2}}$.

For a positive integer $l$, we set
\[%
\begin{array}
[c]{ccc}%
L_{\mathbb{R}}^{2}\left(  \mathbb{Z}_{p}^{N}\right)  & \rightarrow &
\mathcal{D}^{l}(\mathbb{Z}_{p}^{N})\\
f & \rightarrow & \Pi_{l}\left(  f\right)  =%
{\textstyle\sum\limits_{\boldsymbol{i}\in G_{l}^{N}}}
\left\langle f,\Omega\left(  p^{l}\left\Vert x-\boldsymbol{i}\right\Vert
_{p}\right)  \right\rangle \Omega\left(  p^{l}\left\Vert x-\boldsymbol{i}%
\right\Vert _{p}\right)  .
\end{array}
\]
By abuse of notation, we also denote by $\Pi_{l}$ the projection%
\[%
\begin{array}
[c]{ccc}%
L_{\mathbb{R}}^{2}\left(  \mathbb{Z}_{p}^{N}\right)  \times L_{\mathbb{R}}%
^{2}\left(  \mathbb{Z}_{p}^{N}\right)  & \rightarrow & \mathcal{D}%
^{l}(\mathbb{Z}_{p}^{N})\times\mathcal{D}^{l}(\mathbb{Z}_{p}^{N})\\
\left(  f,g\right)  & \rightarrow & \Pi_{l}\left(  f,g\right)  :=\left(
\Pi_{l}\left(  f\right)  ,\Pi_{l}\left(  g\right)  \right)  .
\end{array}
\]
Since $\mathcal{D}^{l}(\mathbb{Z}_{p}^{N})\times\mathcal{D}^{l}(\mathbb{Z}%
_{p}^{N})$ is a finite-dimensional subspace of $L_{\mathbb{R}}^{2}\left(
\mathbb{Z}_{p}^{N}\right)  \times L_{\mathbb{R}}^{2}\left(  \mathbb{Z}_{p}%
^{N}\right)  $, the projection $\Pi_{l}$ is continuous in the $L_{\mathbb{R}%
}^{2}\left(  \mathbb{Z}_{p}^{N}\right)  \times L_{\mathbb{R}}^{2}\left(
\mathbb{Z}_{p}^{N}\right)  $-topology.

\begin{remark}
Notice that $E\circ\Pi_{l}$ agrees with the restriction of $E$ to
$\mathcal{D}^{l}(\mathbb{Z}_{p}^{N})\times\mathcal{D}^{l}(\mathbb{Z}_{p}^{N})$.
\end{remark}

\begin{definition}
Given a positive integer $l$, We denote by $\left(  \mathbb{P}_{K_{1}}%
\otimes\mathbb{P}_{K_{2}}\right)  _{l}$ the image of the measure
$\mathbb{P}_{K}\otimes\mathbb{P}_{K}$ under the map $\Pi_{l}$ (also called the
push-forward measure).
\end{definition}

$\left(  \mathbb{P}_{K_{1}}\otimes\mathbb{P}_{K_{2}}\right)  _{l}$ is a
probability measure on the Borel $\sigma$-algebra $\mathcal{B}(\mathbb{R}%
^{\#G_{l}^{N}})\times\mathcal{B}(\mathbb{R}^{\#G_{l}^{N}})$, which is
determined by its Fourier transform:
\begin{gather*}
\widehat{\left(  \mathbb{P}_{K}\otimes\mathbb{P}_{K}\right)  }_{l}\left(
\left[  \widehat{v}\left(  \boldsymbol{i}\right)  \right]  _{\boldsymbol{i}\in
G_{l}^{N}},\left[  \widehat{h}\left(  \boldsymbol{i}\right)  \right]
_{\boldsymbol{i}\in G_{l}^{N}}\right)  =\\
\exp\left\{  \frac{-1}{2}\left[  \widehat{v}\left(  \boldsymbol{i}\right)
\right]  _{\boldsymbol{i}\in G_{l}^{N}}^{T}C_{l,K_{1}}\left[  \widehat
{v}\left(  \boldsymbol{i}\right)  \right]  _{\boldsymbol{i}\in G_{l}^{N}%
}\right\}  \exp\left\{  \frac{-1}{2}\left[  \widehat{h}\left(  \boldsymbol{i}%
\right)  \right]  _{\boldsymbol{i}\in G_{l}^{N}}^{T}C_{l,K_{2}}\left[
\widehat{h}\left(  \boldsymbol{i}\right)  \right]  _{\boldsymbol{i}\in
G_{l}^{N}}\right\}  ,
\end{gather*}
where $C_{l,K_{1}}$, $C_{l,K_{2}}$ are the covariance matrices attached to
$\square_{K_{1}}$, $\square_{K_{2}}$. By a well-known calculation,
\begin{gather}
\left(  \mathbb{P}_{K}\otimes\mathbb{P}_{K}\right)  _{l}\left(  \left[
v\left(  \boldsymbol{i}\right)  \right]  _{\boldsymbol{i}\in G_{l}^{N}%
},\left[  h\left(  \boldsymbol{i}\right)  \right]  _{\boldsymbol{i}\in
G_{l}^{N}}\right)  =\nonumber\\
\frac{1}{\det\left(  2\pi C_{l,K_{1}}\right)  }\exp\left\{  \frac{-1}%
{2}\left[  v\left(  \boldsymbol{i}\right)  \right]  _{\boldsymbol{i}\in
G_{l}^{N}}^{T}C_{l,K_{1}}^{-1}\left[  v\left(  \boldsymbol{i}\right)  \right]
_{\boldsymbol{i}\in G_{l}^{N}}\right\}  \times\nonumber\\
\frac{1}{\det\left(  2\pi C_{l,K_{2}}\right)  }\exp\left\{  \frac{-1}%
{2}\left[  h\left(  \boldsymbol{i}\right)  \right]  _{\boldsymbol{i}\in
G_{l}^{N}}^{T}C_{l,K_{2}}^{-1}\left[  h\left(  \boldsymbol{i}\right)  \right]
_{\boldsymbol{i}\in G_{l}^{N}}\right\}
{\displaystyle\prod\limits_{\boldsymbol{i}\in G_{l}^{N}}}
dv\left(  \boldsymbol{i}\right)
{\displaystyle\prod\limits_{\boldsymbol{i}\in G_{l}^{N}}}
dh\left(  \boldsymbol{i}\right)  , \label{calculation}%
\end{gather}
where $%
{\textstyle\prod\nolimits_{\boldsymbol{i}\in G_{l}^{N}}}
dv\left(  \boldsymbol{i}\right)
{\textstyle\prod\nolimits_{\boldsymbol{i}\in G_{l}^{N}}}
dh\left(  \boldsymbol{i}\right)  $ denotes the Lebesgue measure of
$\mathbb{R}^{\#G_{l}^{N}}\times\mathbb{R}^{\#G_{l}^{N}}$. Notice that $\left(
\mathbb{P}_{K_{1}}\otimes\mathbb{P}_{K_{2}}\right)  _{l}\left(  \left[
v\left(  \boldsymbol{i}\right)  \right]  _{\boldsymbol{i}\in G_{l}^{N}%
},\left[  h\left(  \boldsymbol{i}\right)  \right]  _{\boldsymbol{i}\in
G_{l}^{N}}\right)  =\left(  \mathbb{P}_{K_{1}}\otimes\mathbb{P}_{K_{2}%
}\right)  _{l}\left(  \boldsymbol{v}_{l},\boldsymbol{h}_{l}\right)  $.

We\ now set%
\begin{equation}
\mathbb{P}_{l}(\boldsymbol{v}_{l},\boldsymbol{h}_{l};\boldsymbol{\theta}%
_{l}):=\frac{\exp(-E_{l}\left(  \boldsymbol{v}_{l},\boldsymbol{h}%
_{l};\boldsymbol{\theta}_{l}\right)  )\left(  \mathbb{P}_{K_{1}}%
\otimes\mathbb{P}_{K_{2}}\right)  _{l}}{\mathcal{Z}\left(  \boldsymbol{\theta
}_{l}\right)  }, \label{EQ_P-l}%
\end{equation}
where%
\begin{equation}
\mathcal{Z}\left(  \boldsymbol{\theta}_{l}\right)  =%
{\displaystyle\iint\limits_{\mathbb{R}^{\#G_{l}^{N}}\times\mathbb{R}%
^{\#G_{l}^{N}}}}
\exp(-E_{l}\left(  \boldsymbol{v}_{l},\boldsymbol{h}_{l};\boldsymbol{\theta
}_{l}\right)  )d\left(  \mathbb{P}_{K_{1}}\otimes\mathbb{P}_{K_{2}}\right)
_{l}. \label{EQ_P-l-1}%
\end{equation}

\begin{theorem}
Let $A$ be a Borel subset from $\mathcal{B}(\mathbb{R}^{\#G_{l}^{N}}%
)\times\mathcal{B}(\mathbb{R}^{\#G_{l}^{N}})$. Then%
\begin{align*}%
{\displaystyle\int\limits_{\Pi_{l}^{-1}\left(  A\right)  }}
d\mathbb{P}(\boldsymbol{v},\boldsymbol{h};\boldsymbol{\theta})  &  =\frac
{1}{\mathcal{Z}^{\left(  2\right)  }\left(  \boldsymbol{\theta}\right)  }%
{\displaystyle\int\limits_{\Pi_{l}^{-1}\left(  A\right)  }}
\exp\left(  -E\circ\Pi_{l}\left(  \boldsymbol{v},\boldsymbol{h}\right)
\right)  d\mathbb{P}_{K_{1}}\left(  \boldsymbol{v}\right)  \otimes
\mathbb{P}_{K_{2}}\left(  \boldsymbol{h}\right) \\
&  =%
{\displaystyle\int\limits_{A}}
d\mathbb{P}_{l}(\boldsymbol{v}_{l},\boldsymbol{h}_{l};\boldsymbol{\theta}%
_{l}).
\end{align*}

\end{theorem}

\begin{remark}
Let $X:\mathbb{R}^{\#G_{l}^{N}}\times\mathbb{R}^{\#G_{l}^{N}}\rightarrow
\mathbb{R}$ be a Borel function. Then%
\[%
{\displaystyle\int\limits_{\Pi_{l}^{-1}\left(  A\right)  }}
X(\Pi_{l}\left(  \boldsymbol{v}\right)  ,\Pi_{l}\left(  \boldsymbol{h}\right)
)d\mathbb{P}(\boldsymbol{v},\boldsymbol{h};\boldsymbol{\theta})=%
{\displaystyle\int\limits_{A}}
X(\boldsymbol{v},\boldsymbol{h})d\mathbb{P}_{l}(\boldsymbol{v},\boldsymbol{h}%
;\boldsymbol{\theta})
\]
in the sense that if one of the integrals exists, so does the other. If $X$ is
bounded then the integrals exist.
\end{remark}

\begin{proof}
We show that $\mathbb{P}_{l}(\boldsymbol{v}_{l},\boldsymbol{h}_{l}%
;\boldsymbol{\theta}_{l})$ is the push-forward measure $\left(  \Pi
_{l}\right)  _{\ast}\left(  \mathbb{P}(\boldsymbol{v},\boldsymbol{h}%
;\boldsymbol{\theta})\right)  $ of $\mathbb{P}(\boldsymbol{v},\boldsymbol{h}%
;\boldsymbol{\theta})$ by $\Pi_{l}$. By using
\[
\Pi_{l}^{-1}\left(  \mathcal{D}^{l}(\mathbb{Z}_{p}^{N})\times\mathcal{D}%
^{l}(\mathbb{Z}_{p}^{N})\right)  =L_{\mathbb{R}}^{2}\left(  \mathbb{Z}_{p}%
^{N}\right)  \times L_{\mathbb{R}}^{2}\left(  \mathbb{Z}_{p}^{N}\right)  ,
\]
and the change of variables formula, see, e.g., \cite[Theorem 1.6.12]{Ash},
\cite[Proposition 1.1]{Da prato}, it follows that $\left(  \Pi_{l}\right)
_{\ast}\left(  \mathbb{P}(\boldsymbol{v},\boldsymbol{h};\boldsymbol{\theta
})\right)  $ is a probability measure. Now, since $\mathbb{P}(\boldsymbol{v}%
,\boldsymbol{h};\boldsymbol{\theta})$ has density ($\frac{\exp\left(
-E(\boldsymbol{v},\boldsymbol{h})\right)  }{\mathcal{Z}^{\left(
\infty\right)  }\left(  \boldsymbol{\theta}\right)  }$) with respect to
$\mathbb{P}_{K_{1}}\left(  \boldsymbol{v}\right)  \otimes\mathbb{P}_{K_{2}%
}\left(  \boldsymbol{h}\right)  $ and the push-forward of this measure is
$\left(  \mathbb{P}_{K_{1}}\otimes\mathbb{P}_{K_{2}}\right)  _{l}$, the change
of variables formula implies that%
\[
\left(  \Pi_{l}\right)  _{\ast}\left(  \mathbb{P}(\boldsymbol{v}%
,\boldsymbol{h};\boldsymbol{\theta})\right)  \sim E^{\cdot}(\boldsymbol{v}%
,\boldsymbol{h};\boldsymbol{\theta})\left(  \mathbb{P}_{K_{1}}\otimes
\mathbb{P}_{K_{2}}\right)  _{l},
\]
and
\begin{align*}
E^{\cdot}(\boldsymbol{v},\boldsymbol{h};\boldsymbol{\theta})  &  =A\exp\left(
-\left(  E_{l}\circ\Pi_{l}\right)  \left(  \boldsymbol{v},\boldsymbol{h}%
\right)  \right)  =A\exp\left(  -E_{l}\left(  \boldsymbol{v}_{l}%
,\boldsymbol{h}_{l}\right)  \right) \\
&  =A\exp\left(  -E\circ\Pi_{l}\left(  \boldsymbol{v},\boldsymbol{h}\right)
\right)  ,
\end{align*}
and thus $A=\mathcal{Z}\left(  \boldsymbol{\theta}_{l}\right)  $.
\end{proof}

\begin{remark}
We set%
\[
E^{\text{kin}}\left(  \boldsymbol{v},\boldsymbol{h}\right)  =\frac{1}{2}\text{
\ }%
{\displaystyle\iint\limits_{\mathbb{Z}_{p}^{N}\times\mathbb{Z}_{p}^{N}}}
\boldsymbol{v}\left(  x\right)  \square_{K}^{-1}\boldsymbol{v}\left(
x\right)  d^{N}x+\frac{1}{2}\text{ \ }%
{\displaystyle\iint\limits_{\mathbb{Z}_{p}^{N}\times\mathbb{Z}_{p}^{N}}}
\boldsymbol{h}\left(  x\right)  \square_{K}^{-1}\boldsymbol{h}\left(
x\right)  d^{N}x.
\]
This is an analog of a classical kinetic energy functional. $E^{\text{kin}%
}\left(  \boldsymbol{v},\boldsymbol{h}\right)  $ is well-defined if
$\boldsymbol{v},\boldsymbol{h}\in\mathcal{D}_{\mathbb{R}}(\mathbb{Z}_{p}^{N}%
)$, see Remark \ref{Note_Inverse_square}. The energy functional $E^{\text{kin}%
}\left(  \boldsymbol{v},\boldsymbol{h}\right)  $ is well-defined in the
subspace%
\[
\mathcal{H}_{K}:=\mathcal{H}_{K}(\mathbb{Z}_{p}^{N})=\left\{  f\in
L_{\mathbb{R}}^{2}\left(  \mathbb{Z}_{p}^{N}\right)  ;%
{\displaystyle\int\limits_{\mathbb{Q}_{p}^{N}}}
\frac{\left\vert \widehat{f}\right\vert ^{2}}{\widehat{K}\left(  \xi\right)
}d^{N}\xi<\infty\right\}  .
\]
The discretization $E_{l}^{\text{kin}}\left(  \boldsymbol{v}_{l}%
,\boldsymbol{h}_{l}\right)  $\ of $E^{\text{kin}}\left(  \boldsymbol{v}%
,\boldsymbol{h}\right)  $ is%
\[
E_{l}^{\text{kin}}\left(  \boldsymbol{v}_{l},\boldsymbol{h}_{l}\right)
=\frac{1}{2}\left[  v_{\boldsymbol{i}}^{l}\right]  _{\boldsymbol{i}\in
G_{l}^{N}}^{T}C_{l,K_{1}}^{-1}\left[  v_{\boldsymbol{i}}^{l}\right]
_{\boldsymbol{i}\in G_{l}^{N}}^{T}+\frac{1}{2}\left[  h_{\boldsymbol{i}}%
^{l}\right]  _{\boldsymbol{i}\in G_{l}^{N}}^{T}C_{l,K_{2}}^{-1}\left[
h_{\boldsymbol{i}}^{l}\right]  _{\boldsymbol{i}\in G_{l}^{N}}^{T}.
\]

\end{remark}

\subsection{$p$-adic discrete DBMS}

Notice that (\ref{EQ_P-l})-(\ref{EQ_P-l-1}) can be rewritten as
\[
\mathbb{P}_{l}(\boldsymbol{v}_{l},\boldsymbol{h}_{l};\boldsymbol{\theta}%
_{l})=\frac{\exp(-E_{l}\left(  \boldsymbol{v}_{l},\boldsymbol{h}%
_{l};\boldsymbol{\theta}_{l}\right)  -E_{l}^{\text{kin}}\left(  \boldsymbol{v}%
_{l},\boldsymbol{h}_{l}\right)  )}{\mathcal{Z}\left(  \boldsymbol{\theta}%
_{l}\right)  }%
{\displaystyle\prod\limits_{_{\boldsymbol{i}\in G_{l}^{N}}}}
dv\left(  \boldsymbol{i}\right)
{\displaystyle\prod\limits_{_{\boldsymbol{i}\in G_{l}^{N}}}}
dh\left(  \boldsymbol{i}\right)  ,
\]%
\[
\mathcal{Z}\left(  \boldsymbol{\theta}_{l}\right)  =%
{\displaystyle\iint\limits_{\mathbb{R}^{\#G_{l}^{N}}\times\mathbb{R}%
^{\#G_{l}^{N}}}}
\exp(-E_{l}\left(  \boldsymbol{v}_{l},\boldsymbol{h}_{l};\boldsymbol{\theta
}_{l}\right)  -E_{l}^{\text{kin}}\left(  \boldsymbol{v}_{l},\boldsymbol{h}%
_{l}\right)  )%
{\displaystyle\prod\limits_{_{\boldsymbol{i}\in G_{l}^{N}}}}
dv\left(  \boldsymbol{i}\right)
{\displaystyle\prod\limits_{_{\boldsymbol{i}\in G_{l}^{N}}}}
dh\left(  \boldsymbol{i}\right)  .
\]

We attach to the energy functional $E_{l}\left(  \boldsymbol{v}_{l}%
,\boldsymbol{h}_{l}\right)  $ a $p$\textit{-adic discrete deep Boltzmann
machine} (DBM), which we identify with a discrete SFT with energy functional
$E_{l}\left(  \boldsymbol{v}_{l},\boldsymbol{h}_{l}\right)  $ and Boltzmann
distribution $\mathbb{P}_{l}(\boldsymbol{v}_{l},\boldsymbol{h}_{l}%
;\boldsymbol{\theta}_{l})$. The random variables $\boldsymbol{v}%
_{l},\boldsymbol{h}_{l}$ are continuous.

In a $p$-adic discrete DBM the visible units, and the hidden units,\ have a
tree-like architecture ($G_{l}^{N}$) as in the classical DBMs. However, only
the vertices (units) located at level $l$ of each tree in $G_{l}^{N}$ carry
states. The other vertices describe the topology of the network. Then $p$-adic
discrete DBMs have less parameters that their classical counterparts. By
changing the functions $a\left(  x\right)  $, $b(x)$, $c(x)$, $d(x)$, $w(x,y)$
several different types of networks are obtained.

A relevant case occurs taking $a(x)$, $b(x)$, $c(x)$, $d(x)\in L_{\mathbb{R}%
}^{\infty}(\mathbb{Z}_{p}^{N})$, and , $w\left(  x,y\right)  =w\left(
x-y\right)  $, with $w$ in $L_{\mathbb{R}}^{\infty}(\mathbb{Z}_{p}^{N})$. In
this case, the discrete energy functional has the form $E_{l}\left(
\boldsymbol{v}_{l},\boldsymbol{h}_{l}\right)  =E_{l}^{\text{free}}\left(
\boldsymbol{v}_{l},\boldsymbol{h}_{l}\right)  +E_{l}^{\text{int}}\left(
\boldsymbol{v}_{l},\boldsymbol{h}_{l}\right)  $, where $E_{l}^{\text{free}%
}\left(  \boldsymbol{v}_{l},\boldsymbol{h}_{l}\right)  $ is given in
(\ref{Energy_free}) and $E_{l}^{\text{int}}\left(  \boldsymbol{v}%
_{l},\boldsymbol{h}_{l}\right)  $ has the form%
\[
E_{l}^{\text{int}}\left(  \boldsymbol{v}_{l},\boldsymbol{h}_{l}\right)  =-%
{\displaystyle\sum\limits_{\boldsymbol{j}\in G_{l}^{N}}}
\text{ }%
{\displaystyle\sum\limits_{\boldsymbol{k}\in G_{l}^{N}}}
w_{k}v_{j+k}h_{j}+%
{\displaystyle\sum\limits_{\boldsymbol{i}\in G_{l}^{N}}}
c_{\boldsymbol{i}}v_{\boldsymbol{i}}^{4}+%
{\displaystyle\sum\limits_{\boldsymbol{i}\in G_{l}^{N}}}
d_{\boldsymbol{i}}h_{\boldsymbol{i}}^{4}.
\]
Since $G_{l}^{N}$ is an additive product group, the energy functional $E_{l}$
has translational symmetry, i.e., $E_{l}$ invariant under the transformations
$\boldsymbol{j}\rightarrow\boldsymbol{j}+\boldsymbol{j}_{0}$, $\boldsymbol{k}%
\rightarrow\boldsymbol{k}+\boldsymbol{k}_{0}$, for any $\boldsymbol{j}_{0}$,
$\boldsymbol{k}_{0}\in G_{l}^{N}$.

\section{\label{Section_7}Correlation functions and continuous spin glasses}

\subsection{Correlation functions}

Along this section, we assume that the fields are from $\boldsymbol{B}%
_{M}^{\left(  2\right)  }$. The information about the local properties of the
system is contained in the \textit{correlation functions }$\boldsymbol{G}%
^{\left(  m+n\right)  }\left(  x,y\right)  $, $x=\left(  x_{1},\ldots
,x_{m}\right)  \in\mathbb{Z}_{p}^{Nm}$, $y=\left(  y_{1},\ldots,y_{n}\right)
\in\mathbb{Z}_{p}^{Nn}$ of the field $\left\{  \boldsymbol{v},\boldsymbol{h}%
\right\}  $:
\begin{gather}
\boldsymbol{G}^{\left(  m+n\right)  }\left(  x,y\right)  :=\left\langle
{\displaystyle\prod\limits_{i=1}^{m}}
\boldsymbol{v}\left(  x_{i}\right)  \text{ }%
{\displaystyle\prod\limits_{j=1}^{n}}
\boldsymbol{h}\left(  y_{j}\right)  \right\rangle \label{correlation_Function}%
\\
=\frac{1}{\mathcal{Z}^{\left(  2\right)  }\left(  \boldsymbol{\theta}\right)
}%
{\displaystyle\iint\limits_{\boldsymbol{B}_{M}^{\left(  2\right)  }%
\times\boldsymbol{B}_{M}^{\left(  2\right)  }}}
\text{ }\left\{
{\displaystyle\prod\limits_{i=1}^{m}}
\boldsymbol{v}\left(  x_{i}\right)  \text{ }%
{\displaystyle\prod\limits_{j=1}^{n}}
\boldsymbol{h}\left(  y_{j}\right)  \right\}  \text{ }e^{-E(\boldsymbol{v}%
,\boldsymbol{h})}d\mathbb{P}_{K_{1}}\left(  \boldsymbol{v}\right)  \otimes
d\mathbb{P}_{K_{2}}\left(  \boldsymbol{h}\right)  .\nonumber
\end{gather}
These functions are also called \textit{the }$\left(  m+n\right)
$\textit{-point Green functions}.

\begin{remark}
If the fields are from $\boldsymbol{B}_{M}^{\left(  \infty\right)  }$, the
correlation functions are defined as in (\ref{correlation_Function}) replacing
$\mathcal{Z}^{\left(  2\right)  }\left(  \boldsymbol{\theta}\right)  $ by
$\mathcal{Z}^{\left(  \infty\right)  }\left(  \boldsymbol{\theta}\right)  $,
and $\boldsymbol{B}_{M}^{\left(  2\right)  }\times\boldsymbol{B}_{M}^{\left(
2\right)  }$ by $\boldsymbol{B}_{M}^{\left(  \infty\right)  }\times
\boldsymbol{B}_{M}^{\left(  \infty\right)  }$.
\end{remark}

\begin{notation}
From now on, we adopt the convention $%
{\textstyle\prod\nolimits_{i=0}^{0}}
\boldsymbol{\cdot}=1$, or $%
{\textstyle\prod\nolimits_{i\in\emptyset}}
\boldsymbol{\cdot}=1$.
\end{notation}

To study of these functions, one introduces two auxiliary external fields
$J_{1},$ $J_{2}$ called \textit{currents, }and adds to the energy functional
$E$ as a linear interaction energy of these currents with the field $\left\{
\boldsymbol{v},\boldsymbol{h}\right\}  $,%
\[
E_{\text{sour}}(\boldsymbol{v},\boldsymbol{h},J_{1},J_{2})=-\left\langle
J_{1},\boldsymbol{v}\right\rangle -\left\langle J_{2},\boldsymbol{h}%
\right\rangle ,
\]
now the energy functional is $E(\boldsymbol{v},\boldsymbol{h},J_{1}%
,J_{2})=E\left(  \boldsymbol{v},\boldsymbol{h}\right)  +E_{\text{sour}%
}(\boldsymbol{v},\boldsymbol{h},J_{1},J_{2})$. The partition function formed
with this energy is%
\[
Z(J_{1},J_{2})=\frac{1}{\mathcal{Z}^{\left(  2\right)  }}%
{\displaystyle\iint\limits_{\boldsymbol{B}_{M}^{\left(  2\right)  }%
\times\boldsymbol{B}_{M}^{\left(  2\right)  }}}
\text{ }e^{-E(\boldsymbol{v},\boldsymbol{h},J_{1},J_{2})}d\mathbb{P}_{K_{1}%
}\left(  \boldsymbol{v}\right)  \otimes d\mathbb{P}_{K_{2}}\left(
\boldsymbol{h}\right)  .
\]

\begin{remark}
\label{Nota_Zeta}If the fields are from $\boldsymbol{B}_{M}^{\left(
\infty\right)  }$, the generating functions has the form%
\[
Z(J_{1},J_{2})=\frac{1}{\mathcal{Z}^{\left(  \infty\right)  }}%
{\displaystyle\iint\limits_{\boldsymbol{B}_{M}^{\left(  \infty\right)  }%
\times\boldsymbol{B}_{M}^{\left(  \infty\right)  }}}
\text{ }e^{-E(\boldsymbol{v},\boldsymbol{h},J_{1},J_{2})}d\mathbb{P}_{K_{1}%
}\left(  \boldsymbol{v}\right)  \otimes d\mathbb{P}_{K_{2}}\left(
\boldsymbol{h}\right)  .
\]

\end{remark}

\begin{lemma}
$\boldsymbol{G}^{\left(  m+n\right)  }\left(  x,y\right)  \in\mathcal{D}%
_{\mathbb{R}}^{\prime}\left(  \mathbb{Z}_{p}^{Nm}\times\mathbb{Z}_{p}%
^{Nn}\right)  $.
\end{lemma}

\begin{proof}
Given $\phi_{i}$, $1\leq i\leq m$ and $\theta_{j}$, $1\leq j\leq n$, the
integral
\[
\mathcal{I}=\frac{1}{\mathcal{Z}^{\left(  2\right)  }}\text{\ }%
{\displaystyle\iint\limits_{\boldsymbol{B}_{M}^{\left(  2\right)  }%
\times\boldsymbol{B}_{M}^{\left(  2\right)  }}}
e^{-E\left(  \boldsymbol{v},\boldsymbol{h}\right)  }%
{\textstyle\prod\limits_{i=1}^{m}}
\left\langle \phi_{i},\boldsymbol{v}\right\rangle
{\textstyle\prod\limits_{j=1}^{n}}
\left\langle \theta_{i},\boldsymbol{h}\right\rangle d\mathbb{P}_{K_{1}}\left(
\boldsymbol{v}\right)  \otimes d\mathbb{P}_{K_{2}}\left(  \boldsymbol{h}%
\right)
\]
converges. Indeed, $e^{-E\left(  \boldsymbol{v},\boldsymbol{h}\right)  }$ is
bounded by a constant depending on $M$, see Lemma \ref{Lemma5}, \ and by the
Cauchy-Schwarz inequality,%
\[%
{\textstyle\prod\nolimits_{i=1}^{m}}
\left\vert \left\langle \phi_{i},\boldsymbol{v}\right\rangle \right\vert \leq
M^{m}%
{\textstyle\prod\nolimits_{i=1}^{m}}
\left\Vert \phi_{i}\right\Vert _{2}\text{, and }%
{\textstyle\prod\nolimits_{i=1}^{m}}
\left\vert \left\langle \theta_{i},\boldsymbol{h}\right\rangle \right\vert
\leq M^{n}%
{\textstyle\prod\nolimits_{i=1}^{m}}
\left\Vert \phi_{i}\right\Vert _{2}\text{.}%
\]
Consequently $\left\vert \mathcal{I}\right\vert <\infty$. Now, by Fubini's
theorem,%
\begin{multline*}
\mathcal{I}=%
{\displaystyle\idotsint\limits_{\mathbb{Z}_{p}^{Nm}\times\mathbb{Z}_{p}^{Nn}}}
\text{ }%
{\textstyle\prod\limits_{i=1}^{m}}
\phi_{i}\left(  x_{i}\right)
{\textstyle\prod\limits_{j=1}^{n}}
\theta_{j}\left(  y_{j}\right)  \times\\
\left\{  \frac{1}{\mathcal{Z}^{\left(  2\right)  }}\text{\ }%
{\displaystyle\iint\limits_{\boldsymbol{B}_{M}^{\left(  2\right)  }%
\times\boldsymbol{B}_{M}^{\left(  2\right)  }}}
e^{-E\left(  \boldsymbol{v},\boldsymbol{h}\right)  }%
{\textstyle\prod\limits_{i=1}^{m}}
\boldsymbol{v}\left(  x_{i}\right)
{\textstyle\prod\limits_{j=1}^{n}}
\boldsymbol{h}\left(  y_{j}\right)  d\mathbb{P}_{K_{1}}\left(  \boldsymbol{v}%
\right)  \otimes d\mathbb{P}_{K_{2}}\left(  \boldsymbol{h}\right)  \right\}
\times\\%
{\textstyle\prod\limits_{i=1}^{m}}
d^{N}\left(  x_{i}\right)
{\textstyle\prod\limits_{j=1}^{n}}
d^{N}\left(  y_{j}\right)  .
\end{multline*}
Therefore, the functional%
\[
\left(  \left\{  \phi_{i};1\leq i\leq m\right\}  ,\left\{  \theta_{j};1\leq
j\leq n\right\}  \right)  \rightarrow\mathcal{I}%
\]
is a well-defined element from $\mathcal{D}_{\mathbb{R}}^{\prime}\left(
\mathbb{Z}_{p}^{Nm}\times\mathbb{Z}_{p}^{Nn}\right)  $.
\end{proof}

\begin{definition}
For $\theta\in\mathcal{D}_{\mathbb{R}}\left(  \mathbb{Z}_{p}^{N}\right)  $,
the functional derivative ${\Huge \partial}_{\theta}^{\left(  1\right)
}Z(J_{1},J_{2})$ of $Z(J_{1},J_{2})$ is defined as
\[
{\Huge \partial}_{\theta}^{\left(  1\right)  }Z(J_{1},J_{2})=\lim
_{\epsilon\rightarrow0}\frac{Z(J_{1}+\epsilon\theta,J_{2})-Z(J_{1},J_{2}%
)}{\epsilon}=\left[  \frac{d}{d\epsilon}Z(J_{1}+\epsilon\theta,J_{2})\right]
_{\epsilon=0}.
\]
The functional derivative ${\Huge \partial}_{\theta}^{\left(  2\right)
}Z(J_{1},J_{2})$ of $Z(J_{1},J_{2})$ is defined as
\[
{\Huge \partial}_{\theta}^{\left(  2\right)  }Z(J_{1},J_{2})=\lim
_{\epsilon\rightarrow0}\frac{Z(J_{1},J_{2}+\epsilon\theta)-Z(J_{1},J_{2}%
)}{\epsilon}=\left[  \frac{d}{d\epsilon}Z(J_{1},J_{2}+\epsilon\theta)\right]
_{\epsilon=0}.
\]

\end{definition}

\begin{lemma}
[{\cite[Lemma 6.4]{Zuniga-RMP-2022}}]\label{Lemma18}Let $\phi_{1}$%
,\ldots,$\phi_{m}$, $\theta_{1}$,\ldots,$\theta_{n}$, in $\mathcal{D}%
_{\mathbb{R}}\left(  \mathbb{Z}_{p}^{N}\right)  $. The functional derivative $%
{\textstyle\prod\nolimits_{i=1}^{n}}
{\Huge \partial}_{\theta_{i}}^{\left(  2\right)  }$ $%
{\textstyle\prod\nolimits_{i=1}^{m}}
{\Huge \partial}_{\phi_{i}}^{\left(  1\right)  }Z(J_{1},J_{2})$ exists, and
the following formula holds true:
\begin{gather*}%
{\textstyle\prod\nolimits_{i=1}^{n}}
{\Huge \partial}_{\theta_{i}}^{\left(  2\right)  }%
{\textstyle\prod\nolimits_{i=1}^{m}}
{\Huge \partial}_{\phi_{i}}^{\left(  1\right)  }Z(J_{1},J_{2})=\\
\frac{1}{\mathcal{Z}^{\left(  2\right)  }}\text{ \ }%
{\displaystyle\iint\limits_{\boldsymbol{B}_{M}^{\left(  2\right)  }%
\times\boldsymbol{B}_{M}^{\left(  2\right)  }}}
e^{-E\left(  \boldsymbol{v},\boldsymbol{h}\right)  +\left\langle
J_{1},\boldsymbol{v}\right\rangle +\left\langle J_{2},\boldsymbol{h}%
\right\rangle }%
{\textstyle\prod\nolimits_{i=1}^{n}}
\left\langle \theta_{i},\boldsymbol{h}\right\rangle \text{ }%
{\textstyle\prod\nolimits_{i=1}^{m}}
\left\langle \phi_{i},\boldsymbol{v}\right\rangle \text{ }d\mathbb{P}_{K_{1}%
}\left(  \boldsymbol{v}\right)  \otimes d\mathbb{P}_{K_{2}}\left(
\boldsymbol{h}\right)  .
\end{gather*}
Furthermore, the functional derivative $%
{\textstyle\prod\nolimits_{i=1}^{n}}
{\Huge \partial}_{\theta_{i}}^{\left(  2\right)  }%
{\textstyle\prod\nolimits_{i=1}^{m}}
{\Huge \partial}_{\phi_{i}}^{\left(  1\right)  }Z(J_{1},J_{2})$ can be
uniquely identified with the distribution%
\begin{gather}
\left(  \left\{  \phi_{i};1\leq i\leq m\right\}  ,\left\{  \theta_{j};1\leq
j\leq n\right\}  \right)  \rightarrow\frac{1}{\mathcal{Z}^{\left(  2\right)
}}\text{ \ }%
{\displaystyle\idotsint\limits_{\mathbb{Z}_{p}^{Nm}\times\mathbb{Z}_{p}^{Nn}}}
\text{ }%
{\textstyle\prod\limits_{i=1}^{n}}
\theta_{i}\left(  y_{i}\right)
{\textstyle\prod\limits_{i=1}^{m}}
\phi_{i}\left(  x_{i}\right)  \times\nonumber\\
\left\{  \text{ }%
{\displaystyle\iint\limits_{\boldsymbol{B}_{M}^{\left(  2\right)  }%
\times\boldsymbol{B}_{M}^{\left(  2\right)  }}}
e^{-E\left(  \boldsymbol{v},\boldsymbol{h}\right)  +E_{\text{sour}%
}(\boldsymbol{v},\boldsymbol{h},J_{1},J_{2})}%
{\textstyle\prod\limits_{i=1}^{n}}
\boldsymbol{h}\left(  y_{i}\right)
{\textstyle\prod\limits_{i=1}^{m}}
\boldsymbol{v}\left(  x_{i}\right)  d\mathbb{P}_{K_{1}}\left(  \boldsymbol{v}%
\right)  \otimes d\mathbb{P}_{K_{2}}\left(  \boldsymbol{h}\right)  \right\}
\times\nonumber\\%
{\textstyle\prod\limits_{i=1}^{m}}
d^{N}x_{i}%
{\textstyle\prod\limits_{i=1}^{n}}
d^{N}y_{i} \label{Eq_31}%
\end{gather}
from $\mathcal{D}_{\mathbb{R}}^{\prime}\left(  \mathbb{Z}_{p}^{Nm}%
\times\mathbb{Z}_{p}^{Nn}\right)  $.
\end{lemma}

In an alternative way, the functional derivative $\frac{\delta}{\delta
_{1}J_{1}\left(  y\right)  }Z(J_{1},J_{2})$ \ can be defined as the
distribution from $\mathcal{D}_{\mathbb{R}}^{\prime}\left(  \mathbb{Z}_{p}%
^{N}\right)  $ satisfying%
\[%
{\textstyle\int\limits_{\mathbb{Z}_{p}^{N}}}
\theta\left(  y\right)  \left(  \frac{\delta}{\delta J\left(  y\right)
}Z(J_{1},J_{2})\right)  \left(  y\right)  d^{N}y=\left[  \frac{d}{d\epsilon
}Z(J_{1}+\epsilon\theta,J_{2})\right]  _{\epsilon=0}.
\]
The distribution attached to $\frac{\delta}{\delta_{2}J_{2}\left(  y\right)
}Z(J_{1},J_{2})$ is defined in a similar way. Using this notation and formula
(\ref{Eq_31}), one gets%
\begin{gather*}%
{\displaystyle\prod\limits_{i=1}^{n}}
\frac{\delta}{\delta_{2}J_{2}\left(  y_{m}\right)  }%
{\displaystyle\prod\limits_{i=1}^{m}}
\frac{\delta}{\delta_{1}J_{1}\left(  x_{i}\right)  }Z(J_{1},J_{2})=\\
\frac{1}{\mathcal{Z}^{\left(  2\right)  }}\text{ \ }%
{\displaystyle\iint\limits_{\boldsymbol{B}_{M}^{\left(  2\right)  }%
\times\boldsymbol{B}_{M}^{\left(  2\right)  }}}
e^{-E\left(  \boldsymbol{v},\boldsymbol{h}\right)  +\left\langle
J_{1},\boldsymbol{v}\right\rangle +\left\langle J_{2},\boldsymbol{h}%
\right\rangle }%
{\textstyle\prod\limits_{i=1}^{n}}
\boldsymbol{h}\left(  y_{i}\right)
{\textstyle\prod\limits_{i=1}^{m}}
\boldsymbol{v}\left(  x_{i}\right)  d\mathbb{P}_{K_{1}}\left(  \boldsymbol{v}%
\right)  \otimes d\mathbb{P}_{K_{2}}\left(  \boldsymbol{h}\right)
\end{gather*}
as a distribution from $\mathcal{D}_{\mathbb{R}}^{\prime}\left(
\mathbb{Z}_{p}^{Nm}\times\mathbb{Z}_{p}^{Nn}\right)  $. In conclusion, we have
the following formula:

\begin{lemma}
\label{Lemma6A}The functional derivatives of $Z(J_{1},J_{2})$ with respect to
$J_{1}$, $J_{2}$ evaluated at $J_{0}=0$, $J_{1}=0$\ give the correlation
functions of the system:%
\[
\boldsymbol{G}^{\left(  m+n\right)  }\left(  x,y\right)  =\left[
{\textstyle\prod\limits_{i=1}^{n}}
\frac{\delta}{\delta_{2}J_{2}\left(  y_{i}\right)  }\text{ }%
{\textstyle\prod\limits_{i=1}^{m}}
\frac{\delta}{\delta_{1}J_{1}\left(  x_{i}\right)  }Z(J_{1},J_{2})\right]
_{\substack{J_{0}=0\\J_{1}=0}},
\]
for $x=\left(  x_{1},\ldots,x_{m}\right)  \in\mathbb{Z}_{p}^{Nm}$, $y=\left(
y_{1},\ldots,y_{n}\right)  \in\mathbb{Z}_{p}^{Nn}$.
\end{lemma}

\begin{remark}
This formula is also valid if the fields are from $\boldsymbol{B}_{M}^{\left(
\infty\right)  }$, see Remark \ref{Nota_Zeta}.
\end{remark}

\subsection{\label{Section_Spin_Glassses}Continuous spin glasses}

We set the operators
\begin{align*}
\boldsymbol{W}\phi\left(  x\right)   &  :=%
{\displaystyle\int\limits_{\mathbb{Z}_{p}^{N}}}
w\left(  x,z\right)  \phi\left(  z\right)  d^{N}z,\\
\widetilde{\boldsymbol{W}}\phi\left(  y\right)   &  :=%
{\displaystyle\int\limits_{\mathbb{Z}_{p}^{N}}}
w\left(  z,y\right)  \phi\left(  z\right)  d^{N}z,
\end{align*}
for $\phi\in L_{\mathbb{R}}^{2}\left(  \mathbb{Z}_{p}^{N}\right)  $. Since
$w\left(  x,y\right)  \in L_{\mathbb{R}}^{\infty}(\mathbb{Z}_{p}^{N}%
\times\mathbb{Z}_{p}^{N})$, the Cauchy-Schwarz inequality implies that
$\boldsymbol{W}$, $\widetilde{\boldsymbol{W}}:L_{\mathbb{R}}^{2}\left(
\mathbb{Z}_{p}^{N}\right)  \rightarrow L_{\mathbb{R}}^{2}\left(
\mathbb{Z}_{p}^{N}\right)  $ are well-defined linear bounded operators.

We aim to develop a general perturbative theory for general kernels $w\left(
x,y\right)  \in L_{\mathbb{R}}^{\infty}(\mathbb{Z}_{p}^{N}\times\mathbb{Z}%
_{p}^{N})$. This requires some hypotheses \ about the equations
$\boldsymbol{W}\phi\left(  x\right)  =J(x)$, $\widetilde{\boldsymbol{W}%
}\widetilde{\phi}\left(  x\right)  =J(x)$. Motivated by \cite[Section
7.2]{Zinn-Justin}, we introduce the following hypotheses.

The equation $\boldsymbol{W}\phi\left(  x\right)  =J(x)$, with $J\in
L_{\mathbb{R}}^{2}\left(  \mathbb{Z}_{p}^{N}\right)  $, has a solution of the
form%
\begin{equation}
\phi\left(  x\right)  =%
{\displaystyle\int\limits_{\mathbb{Z}_{p}^{N}}}
G_{w}\left(  x,z\right)  J\left(  z\right)  d^{N}z\in L_{\mathbb{R}}%
^{2}\left(  \mathbb{Z}_{p}^{N}\right)  \text{, for }x\in\mathbb{Z}_{p}%
^{N}\text{;} \tag{H4}\label{Hypothesis 1}%
\end{equation}
the equation $\widetilde{\boldsymbol{W}}\widetilde{\phi}\left(  x\right)
=J(x)$, with $J\in L_{\mathbb{R}}^{2}\left(  \mathbb{Z}_{p}^{N}\right)  $, has
a solution of the form%
\begin{equation}
\widetilde{\phi}\left(  x\right)  =%
{\displaystyle\int\limits_{\mathbb{Z}_{p}^{N}}}
\widetilde{G}_{w}\left(  z,x\right)  J\left(  z\right)  d^{N}z\in
L_{\mathbb{R}}^{2}\left(  \mathbb{Q}_{p}^{N}\right)  \text{, for }%
x\in\mathbb{Z}_{p}^{N}\text{.} \tag{H5}\label{Hypothesis 2}%
\end{equation}

\begin{remark}
\label{Note_1}(i) Hypotheses (\ref{Hypothesis 1})-(\ref{Hypothesis 2}) are
mathematical reformulations about the existence of an inverse, in the
distributional sense, for the kernels of the operators $\boldsymbol{W}$,
$\widetilde{\boldsymbol{W}}$:%
\[
\int d^{N}z\text{ }w\left(  x,z\right)  G_{w}\left(  z,y\right)
=\delta\left(  x-y\right)  ,
\]%
\[
\int d^{N}z\text{ }w\left(  z,x\right)  \widetilde{G}_{w}\left(  y,z\right)
=\delta\left(  y-x\right)  ,
\]
see \cite[Section 7.2]{Zinn-Justin}.

(ii) Notice that $\boldsymbol{W}$, $\widetilde{\boldsymbol{W}}$ are
well-defined linear bounded operators on $L_{\mathbb{R}}^{2}\left(
\mathbb{Q}_{p}^{N}\right)  $, and that
\[
\left\langle f,\boldsymbol{W}g\right\rangle =%
{\displaystyle\int\limits_{\mathbb{Z}_{p}^{N}}}
f(x)\boldsymbol{W}g\left(  x\right)  d^{N}x
\]
is well-defined for $f,g\in L_{\mathbb{R}}^{2}\left(  \mathbb{Q}_{p}%
^{N}\right)  $.
\end{remark}

We now consider the energy functional:%
\begin{align}
E^{\text{spin}}(\boldsymbol{v},\boldsymbol{h},J_{1},J_{2})  &  :=-\left\langle
\boldsymbol{h},\boldsymbol{Wv}\right\rangle -\left\langle \boldsymbol{v}%
,J_{1}\right\rangle -\left\langle \boldsymbol{h},J_{2}\right\rangle
\nonumber\\
&  =-\left\langle \boldsymbol{v},\widetilde{\boldsymbol{W}}\boldsymbol{h}%
\right\rangle -\left\langle \boldsymbol{v},J_{1}\right\rangle -\left\langle
\boldsymbol{h},J_{2}\right\rangle . \label{Energy_Spin_Glass}%
\end{align}
This functional resembles the energy functional of a continuous spin glass
with two types of continuous spins $\boldsymbol{v},\boldsymbol{h}$ and two
different magnetic fields \ $J_{1},J_{2}$. The partition function
corresponding to $E^{\text{spin}}(\boldsymbol{v},\boldsymbol{h},J_{1},J_{2})$
is given by
\[
Z^{\text{spin}}(J_{1},J_{2})=\frac{1}{\mathcal{Z}_{\text{spin}}^{\left(
2\right)  }}\text{ }%
{\displaystyle\iint\limits_{\boldsymbol{B}_{M}^{\left(  2\right)  }%
\times\boldsymbol{B}_{M}^{\left(  2\right)  }}}
\text{ }e^{-E^{\text{spin}}(\boldsymbol{v},\boldsymbol{h},J_{1},J_{2}%
)}d\mathbb{P}_{K_{1}}\left(  \boldsymbol{v}\right)  \otimes d\mathbb{P}%
_{K_{2}}\left(  \boldsymbol{h}\right)  ,
\]
where%
\[
\mathcal{Z}_{\text{spin}}^{\left(  2\right)  }=%
{\displaystyle\iint\limits_{\boldsymbol{B}_{M}^{\left(  2\right)  }%
\times\boldsymbol{B}_{M}^{\left(  2\right)  }}}
\text{ }e^{\left\langle \boldsymbol{h},\boldsymbol{Wv}\right\rangle
}d\mathbb{P}_{K_{1}}\left(  \boldsymbol{v}\right)  \otimes d\mathbb{P}_{K_{2}%
}\left(  \boldsymbol{h}\right)  .
\]
The goal of this section is to establish the following formula:

\begin{theorem}
\label{Lemma6}Under hypotheses (\ref{Hypothesis 1})-(\ref{Hypothesis 2}), the
following formula holds true:
\begin{equation}
Z^{\text{spin}}(J_{1},J_{2})=C(\boldsymbol{v}_{0},\boldsymbol{h}%
_{0},\mathcal{Z}_{\text{spin}}^{\left(  2\right)  })\exp\left(  \left\langle
\text{ }%
{\displaystyle\int\limits_{\mathbb{Z}_{p}^{N}}}
G_{w}\left(  x,z\right)  J_{2}\left(  z\right)  d^{N}z,J_{1}\left(  x\right)
\right\rangle \right)  , \label{Z_spin}%
\end{equation}
where $C(\boldsymbol{v}_{0},\boldsymbol{h}_{0},\mathcal{Z}_{\text{spin}%
}^{\left(  2\right)  })$ is a positive constant.
\end{theorem}

\begin{proof}
We first notice that%
\begin{equation}
E^{\text{spin}}(\boldsymbol{v},\boldsymbol{h},J_{1},J_{2})=-\left\langle
\boldsymbol{h},\boldsymbol{Wv+}J_{2}\right\rangle -\left\langle \boldsymbol{v}%
,J_{1}\right\rangle , \label{E_spin_1}%
\end{equation}
and%
\begin{equation}
E^{\text{spin}}(\boldsymbol{v},\boldsymbol{h},J_{1},J_{2})=-\left\langle
\boldsymbol{v},\widetilde{\boldsymbol{W}}\boldsymbol{h+}J_{1}\right\rangle
-\left\langle \boldsymbol{h},J_{2}\right\rangle . \label{E_spin_2}%
\end{equation}
By using hypotheses (\ref{Hypothesis 1})-(\ref{Hypothesis 2}), we pick two
functions $\boldsymbol{v}_{0},\boldsymbol{h}_{0}\in L_{\mathbb{R}}^{2}\left(
\mathbb{Z}_{p}^{N}\right)  $ satisfying%
\begin{equation}
\boldsymbol{Wv}_{0}\boldsymbol{+}J_{2}=0\text{, \ and }\widetilde
{\boldsymbol{W}}\boldsymbol{h}_{0}\boldsymbol{+}J_{1}=0.
\label{special_solutions}%
\end{equation}
We now change variables in $Z^{\text{spin}}(J_{1},J_{2})$ as%
\[%
\begin{array}
[c]{ccc}%
\boldsymbol{B}_{M}^{\left(  2\right)  }\times\boldsymbol{B}_{M}^{\left(
2\right)  } & \rightarrow & \boldsymbol{B}_{M}^{\left(  2\right)  }\left(
-\boldsymbol{v}_{0}\right)  \times\boldsymbol{B}_{M}^{\left(  2\right)  }\\
\left(  \boldsymbol{v},\boldsymbol{h}\right)  & \rightarrow & \left(
\boldsymbol{v}^{\prime},\boldsymbol{h}^{\prime}\right)  ,\\
&  &
\end{array}
\]
where $\boldsymbol{v}^{\prime}=\boldsymbol{v}-\boldsymbol{v}_{0}$,
$\boldsymbol{h}^{\prime}=\boldsymbol{h}$. We denote by $dQ\left(
\boldsymbol{v}^{\prime},\boldsymbol{h}^{\prime}\right)  $ the push-forward of
the measure $d\mathbb{P}_{K_{1}}\left(  \boldsymbol{v}\right)  \otimes
d\mathbb{P}_{K_{2}}\left(  \boldsymbol{h}\right)  $\ to $\boldsymbol{B}%
_{M}^{\left(  2\right)  }\left(  -\boldsymbol{v}_{0}\right)  \times
\boldsymbol{B}_{M}^{\left(  2\right)  }$, here $\boldsymbol{B}_{M}^{\left(
2\right)  }\left(  -\boldsymbol{v}_{0}\right)  \times\boldsymbol{B}%
_{M}^{\left(  2\right)  }$\ is endowed with the Borel $\sigma$-algebra. For
further details, the reader may consult \cite[Theorem 1.6.12]{Ash}.

By using (\ref{E_spin_1}) and (\ref{special_solutions}), we have%
\begin{align*}
Z^{\text{spin}}(J_{1},J_{2})  &  =\frac{e^{\left\langle \boldsymbol{v}%
_{0},J_{1}\right\rangle }}{\mathcal{Z}_{\text{spin}}^{\left(  2\right)  }%
}\text{ }%
{\displaystyle\iint\limits_{\boldsymbol{B}_{M}^{\left(  2\right)  }\left(
-\boldsymbol{v}_{0}\right)  \times\boldsymbol{B}_{M}^{\left(  2\right)  }}}
\text{ }e^{\left\langle \boldsymbol{h}^{\prime},\boldsymbol{Wv}^{\prime
}\right\rangle +\left\langle \boldsymbol{v}^{\prime},J_{1}\right\rangle
}dQ\left(  \boldsymbol{v}^{\prime},\boldsymbol{h}^{\prime}\right) \\
&  =\frac{e^{\left\langle \boldsymbol{v}_{0},J_{1}\right\rangle }}%
{\mathcal{Z}_{\text{spin}}^{\left(  2\right)  }}\text{ }%
{\displaystyle\iint\limits_{\boldsymbol{B}_{M}^{\left(  2\right)  }\left(
-\boldsymbol{v}_{0}\right)  \times\boldsymbol{B}_{M}^{\left(  2\right)  }}}
\text{ }e^{\left\langle \boldsymbol{v}^{\prime},\widetilde{\boldsymbol{W}%
}\boldsymbol{h}^{\prime}+J_{1}\right\rangle }dQ\left(  \boldsymbol{v}^{\prime
},\boldsymbol{h}^{\prime}\right)  .
\end{align*}
Notice that $\left\langle \boldsymbol{v}^{\prime},\widetilde{\boldsymbol{W}%
}\boldsymbol{h}^{\prime}+J_{1}\right\rangle $ is well-defined, see Remark
\ref{Note_1}-(ii). We now change variables \ as $\boldsymbol{v}^{\prime\prime
}=\boldsymbol{v}^{\prime}$, $\boldsymbol{h}^{\prime\prime}=\boldsymbol{h}%
^{\prime}-\boldsymbol{h}_{0}$. We denote\ by $dS\left(  \boldsymbol{v}%
^{\prime\prime},\boldsymbol{h}^{\prime\prime}\right)  $ the push-forward of
the measure $dQ\left(  \boldsymbol{v}^{\prime},\boldsymbol{h}^{\prime}\right)
$ to $\boldsymbol{B}_{M}^{\left(  2\right)  }\left(  -\boldsymbol{v}%
_{0}\right)  \times\boldsymbol{B}_{M}^{\left(  2\right)  }\left(
-\boldsymbol{h}_{0}\right)  $, then%
\begin{align*}
Z^{\text{spin}}(J_{1},J_{2})  &  =\frac{e^{\left\langle \boldsymbol{v}%
_{0},J_{1}\right\rangle }}{\mathcal{Z}_{\text{spin}}^{\left(  2\right)  }%
}\text{ }%
{\displaystyle\iint\limits_{\boldsymbol{B}_{M}^{\left(  2\right)  }\left(
-\boldsymbol{v}_{0}\right)  \times\boldsymbol{B}_{M}^{\left(  2\right)
}\left(  -\boldsymbol{h}_{0}\right)  }}
\text{ }e^{\left\langle \boldsymbol{v}^{\prime\prime},\widetilde
{\boldsymbol{W}}\boldsymbol{h}^{\prime\prime}\right\rangle }dR\left(
\boldsymbol{v}^{\prime\prime},\boldsymbol{h}^{\prime\prime}\right) \\
&  =C(\boldsymbol{v}_{0},\boldsymbol{h}_{0},\mathcal{Z}_{\text{spin}}^{\left(
2\right)  })e^{\left\langle \boldsymbol{v}_{0},J_{1}\right\rangle }\\
&  =C(\boldsymbol{v}_{0},\boldsymbol{h}_{0},\mathcal{Z}_{\text{spin}}^{\left(
2\right)  })\exp\left(  \left\langle \text{ }%
{\displaystyle\int\limits_{\mathbb{Z}_{p}^{N}}}
G_{w}\left(  x,z\right)  J_{2}\left(  z\right)  d^{N}z,J_{1}\left(  x\right)
\right\rangle \right)  ,
\end{align*}
where $C(\boldsymbol{v}_{0},\boldsymbol{h}_{0},\mathcal{Z}_{\text{spin}%
}^{\left(  2\right)  })$\ is a normalization constant.
\end{proof}

\begin{remark}
The following formula is also valid: \
\[
Z^{\text{spin}}(J_{1},J_{2})=C^{\prime}(\boldsymbol{v}_{0},\boldsymbol{h}%
_{0},\mathcal{Z}_{\text{spin}}^{\left(  2\right)  })\exp\left(  \left\langle
\text{ }%
{\displaystyle\int\limits_{\mathbb{Z}_{p}^{N}}}
\widetilde{G}_{w}\left(  z,x\right)  J_{1}\left(  z\right)  d^{N}%
z,J_{2}\left(  x\right)  \right\rangle \right)  .
\]

\end{remark}

\section{\label{Section_8}Continuous Boltzmann machines}

Along this section, we assume that the fields are from $\boldsymbol{B}%
_{M}^{\left(  2\right)  }$. We consider continuous Boltzmann machines with
parameters of the form $\boldsymbol{\theta}=\left(  w,a,b,0,0\right)  $. The
generating functional $Z(J_{1},J_{2};\boldsymbol{\theta})$ can be computed by
replacing $J_{1}$ by $a+J_{1}$ and $J_{2}$ by $b+J_{2}$ in $Z^{\text{spin}%
}(J_{1},J_{2})$, see (\ref{Z_spin}):%
\begin{gather*}
Z(J_{1},J_{2};\boldsymbol{\theta})=A\exp\left(  \left\langle \text{ }%
{\displaystyle\int\limits_{\mathbb{Z}_{p}^{N}}}
G_{w}\left(  x,z\right)  \left(  b\left(  z\right)  +J_{2}\left(  z\right)
\right)  d^{N}z,J_{1}\left(  x\right)  \right\rangle \right)  \times\\
\exp\left(  \left\langle \text{ }%
{\displaystyle\int\limits_{\mathbb{Z}_{p}^{N}}}
G_{w}\left(  x,z\right)  \left(  b\left(  z\right)  +J_{2}\left(  z\right)
\right)  d^{N}z,a\left(  x\right)  \right\rangle \right)  ,
\end{gather*}
where $A:=A(\boldsymbol{v}_{0},\boldsymbol{h}_{0},\mathcal{Z}_{\text{spin}%
}^{\left(  2\right)  },a,b)$ is defined as
\[
A=C^{\prime}(\boldsymbol{v}_{0},\boldsymbol{h}_{0},\mathcal{Z}_{\text{spin}%
}^{\left(  2\right)  })\exp\left(  \left\langle \text{ }\int_{\mathbb{Z}%
_{p}^{N}}G_{w}\left(  x,z\right)  b\left(  z\right)  d^{N}z,a\left(  x\right)
\right\rangle \right)  .
\]
By applying Lemma \ref{Lemma6A}, we have the following formula:

\begin{theorem}
\label{Theorem2}Set
\begin{align*}
{\huge F}(J_{1},J_{2};a,b,G_{w})  &  :=\left\langle \text{ }%
{\displaystyle\int\limits_{\mathbb{Z}_{p}^{N}}}
G_{w}\left(  x,z\right)  \left(  b\left(  z\right)  +J_{2}\left(  z\right)
\right)  d^{N}z,\text{ }J_{1}\left(  x\right)  \right\rangle \\
&  +\left\langle \text{ }%
{\displaystyle\int\limits_{\mathbb{Z}_{p}^{N}}}
G_{w}\left(  x,z\right)  \left(  b\left(  z\right)  +J_{2}\left(  z\right)
\right)  d^{N}z,\text{ }a\left(  x\right)  \right\rangle .
\end{align*}
Then
\[
\boldsymbol{G}^{\left(  m+n\right)  }\left(  x,y;\boldsymbol{\theta}\right)
=A%
{\textstyle\prod\limits_{j=1}^{n}}
\frac{\delta}{\delta_{2}J_{2}\left(  y_{j}\right)  }\text{ }%
{\textstyle\prod\limits_{i=1}^{m}}
\frac{\delta}{\delta_{1}J_{1}\left(  x_{j}\right)  }\left.  \exp\left(
{\huge F}(J_{1},J_{2};a,b,G_{w})\right)  \right\vert _{_{\substack{J_{0}%
=0\\J_{1}=0}}},
\]
where $\boldsymbol{\theta}=\left(  w,a,b,0,0\right)  $, for $x\in
\mathbb{Z}_{p}^{Nm}$, $y\in\mathbb{Z}_{p}^{Nn}$.
\end{theorem}

\subsection{Some examples of BMs}

In this section, we study the correlation functions $\boldsymbol{G}^{\left(
1+0\right)  }\left(  x\right)  $, $\boldsymbol{G}^{\left(  0+1\right)
}\left(  y\right)  $, $\boldsymbol{G}^{\left(  1+1\right)  }\left(
x,y\right)  $, with $x,y\in\mathbb{Z}_{p}^{N}$, which are, respectively, the
means of the visible and hidden fields, and the correlation between the
visible and hidden fields. In this section, we assume that $\boldsymbol{W}%
\Phi\left(  x\right)  =w\left(  x\right)  \ast\Phi\left(  x\right)  $, then
$\frac{1}{\widehat{w}\left(  \xi\right)  }$is the corresponding propagator .
In $p$-adic QFT several types of propagators have been proposed:%
\[
\text{\ }\frac{1}{\left\vert \sum_{i=1}^{N}\xi_{i}^{2}\right\vert _{p}+m^{2}%
}\text{, \ }\frac{1}{\sum_{i=1}^{N}\left\vert \xi_{i}\right\vert _{p}%
^{2}+m^{2}}\text{, }\frac{1}{\left\Vert \xi\right\Vert _{p}^{2}+m^{2}}\text{,
}\frac{1}{\left[  \max\left\{  \left\Vert \xi\right\Vert _{p},m\right\}
\right]  ^{2}},
\]
where $m\in\mathbb{R}$, and%
\[
\frac{1}{\left\vert \sum_{i=1}^{N}\xi_{i}^{2}+m^{2}\right\vert _{p}},
\]
where $m\in\mathbb{Q}_{p}$, see \cite{Smirnov}. At this moment, it is not
clear how to pick propagators relevant for applications in artificial
intelligence. In our view, this task is out of the scope of the quantum field theory.

In this section, we take
\[
\frac{1}{\widehat{w}\left(  \xi\right)  }=\frac{1}{\left[  \max\left\{
\left\Vert \xi\right\Vert _{p},1\right\}  \right]  ^{\alpha}}=\frac{1}{\left[
\xi\right]  _{p}^{\alpha}}\text{ for }\alpha>0.
\]
Then%
\begin{equation}
G_{w}\left(  x\right)  :=\mathcal{F}_{\xi\rightarrow x}^{-1}\left(  \frac
{1}{\left[  \xi\right]  _{p}^{\alpha}}\right)  =\left\{
\begin{array}
[c]{ccc}%
\frac{1}{\Gamma\left(  \alpha\right)  }\left\{  \left\Vert x\right\Vert
_{p}^{\alpha-N}-p^{\alpha-N}\right\}  \Omega\left(  \left\Vert x\right\Vert
_{p}\right)  & \text{if} & \alpha\neq N\\
&  & \\
\left(  1-p^{-N}\right)  \log_{p}\left(  \frac{p}{\left\Vert x\right\Vert
_{p}}\right)  \Omega\left(  \left\Vert x\right\Vert _{p}\right)  & \text{if} &
\alpha=N,
\end{array}
\right.  \label{Eq_G_w}%
\end{equation}
where $\Gamma\left(  \alpha\right)  =\frac{1-p^{\alpha-N}}{1-p^{-\alpha}}$.
Furthermore, $G_{w}\in L_{\mathbb{R}}^{1}\left(  \mathbb{Z}_{p}^{N}\right)  $,
see \cite[Chap III, Lemma 5.2]{Taibleson}. Thus%
\[
\left(  w\ast\Phi\right)  \left(  x\right)  =\mathcal{F}_{\xi\rightarrow
x}^{-1}\left(  \left[  \xi\right]  _{p}^{\alpha}\mathcal{F}_{x\rightarrow\xi
}\Phi\right)  ,
\]
and the solution of the equation
\[
\boldsymbol{W}\Phi\left(  x\right)  :=\left(  w\ast\Phi\right)  \left(
x\right)  =J(x)\text{, with }J\in L_{\mathbb{R}}^{2}\left(  \mathbb{Z}_{p}%
^{N}\right)  \text{,}%
\]
is%
\[
\Phi\left(  x\right)  =G_{w}\left(  x\right)  \ast J(x)\in L_{\mathbb{R}}%
^{2}\left(  \mathbb{Z}_{p}^{N}\right)  ,
\]
since $\left\Vert G_{w}\ast J\right\Vert _{2}\leq\left\Vert G_{w}\right\Vert
_{1}\left\Vert J\right\Vert _{2}$, see \cite[Chap III, Section 1]{Taibleson}.
We set$\Phi$%
\[
\widetilde{G}_{w}\left(  x\right)  :=G_{w}\left(  -x\right)  \text{ \ \ \ and
\ \ \ \ }\widetilde{\boldsymbol{W}}\Phi\left(  x\right)  =\widetilde{G}%
_{w}\left(  x\right)  \ast\Phi\left(  x\right)  .
\]

\begin{remark}
\label{Note_2}Let $m\in\mathbb{N}$, $w\in L_{\mathbb{R}}^{\infty}\left(
p^{m}\mathbb{Z}_{p}^{N}\right)  \subset L_{\mathbb{R}}^{1}\left(
p^{m}\mathbb{Z}_{p}^{N}\right)  $. Assume that%
\[
\widehat{\Phi}\left(  \xi\right)  =\frac{\widehat{J}\left(  \xi\right)
}{\widehat{w}\left(  \xi\right)  }\in\mathcal{F}(L_{\mathbb{R}}^{2}%
(\mathbb{Q}_{p}^{N})),
\]
for any $J\in L_{\mathbb{R}}^{1}\left(  p^{m}\mathbb{Z}_{p}^{N}\right)  $.
Take
\[
G_{w}\left(  x\right)  :=\mathcal{F}_{\xi\rightarrow x}^{-1}\left(  \frac
{1}{\widehat{w}\left(  \xi\right)  }\right)  \text{, for }x\in p^{m}%
\mathbb{Z}_{p}^{N},
\]
as a distribution from $\mathcal{D}_{\mathbb{R}}^{\prime}(p^{m}\mathbb{Z}%
_{p}^{N})$, then%
\[
\Phi\left(  x\right)  =G_{w}\left(  x\right)  \ast J\left(  x\right)  \in
L_{\mathbb{R}}^{2}\left(  p^{m}\mathbb{Z}_{p}^{N}\right)  .
\]

\end{remark}

\subsubsection{Computation of $\boldsymbol{G}^{\left(  m+0\right)  }\left(
\left[  x_{i}\right]  _{1\leq i\leq m}\right)  $}

We now consider BMs with energy functionals of the form%
\[
E(\boldsymbol{v},\boldsymbol{h})=-\left\langle \boldsymbol{h},\boldsymbol{Wv}%
\right\rangle -\left\langle \boldsymbol{v},a\right\rangle -\left\langle
\boldsymbol{h},b\right\rangle =-\left\langle \boldsymbol{v},\widetilde
{\boldsymbol{W}}\boldsymbol{h}\right\rangle -\left\langle \boldsymbol{v}%
,a\right\rangle -\left\langle \boldsymbol{h},b\right\rangle .
\]
To compute the correlations function of a such theory, we\ use Theorem
\ref{Theorem2}. In this case%
\[
{\huge F}(J_{1},J_{2};a,b,G_{w}):={\huge F}(J_{1},J_{2})=\left\langle
G_{w}\ast\left(  b+J_{2}\right)  ,\text{ }J_{1}\right\rangle +\left\langle
G_{w}\ast J_{2},\text{ }a\right\rangle .
\]
The calculations are valid if $G_{w}\left(  x\right)  :=\mathcal{F}%
_{\xi\rightarrow x}^{-1}\left(  \frac{1}{\widehat{w}\left(  \xi\right)
}\right)  \in\mathcal{D}_{\mathbb{R}}^{\prime}(\mathbb{Q}_{p}^{N})$. In
particular, if $G_{w}\left(  x\right)  $ has the form (\ref{Eq_G_w}).

We first compute
\begin{gather*}
\left[  \frac{d}{d\epsilon}\exp\left(  {\huge F}(J_{1}+\epsilon\theta
,J_{2})\right)  \right]  _{\epsilon=0}=\\
\exp\left(  {\huge F}(J_{1},J_{2})\right)  \left[  \frac{d}{d\epsilon}%
\exp\left(  \left\langle G_{w}\ast\left(  b+J_{2}\right)  ,\text{ }%
\epsilon\theta\right\rangle \right)  \right]  _{\epsilon=0}=\\
\exp\left(  {\huge F}(J_{1},J_{2})\right)  \left\langle G_{w}\ast\left(
b+J_{2}\right)  ,\text{ }\theta\right\rangle =\\
\exp\left(  {\huge F}(J_{1},J_{2})\right)
{\displaystyle\int\limits_{\mathbb{Z}_{p}^{N}}}
\left(  G_{w}\ast\left(  b+J_{2}\right)  \right)  \left(  x\right)  \text{
}\theta\left(  x\right)  d^{N}y,
\end{gather*}
i.e.,%
\[
\frac{\delta}{\delta_{1}J_{1}\left(  x_{i}\right)  }\exp\left(  {\huge F}%
(J_{1},J_{2})\right)  =\exp\left(  {\huge F}(J_{1},J_{2})\right)  \left(
G_{w}\ast\left(  b+J_{2}\right)  \right)  \left(  x_{i}\right)  .
\]
Now by using this the formula recursively,%
\begin{equation}%
{\displaystyle\prod\limits_{i=1}^{m}}
\frac{\delta}{\delta_{1}J_{1}\left(  x_{i}\right)  }\exp\left(  {\huge F}%
(J_{1},J_{2})\right)  =\exp\left(  {\huge F}(J_{1},J_{2})\right)
{\displaystyle\prod\limits_{i=1}^{m}}
\left[  G_{w}\ast\left(  b+J_{2}\right)  \right]  \left(  x_{i}\right)  .
\label{Formula_1}%
\end{equation}
By applying Theorem \ref{Theorem2},%
\[
\boldsymbol{G}^{\left(  m+0\right)  }\left(  \left[  x_{i}\right]  _{1\leq
i\leq m}\right)  =A%
{\displaystyle\prod\limits_{i=1}^{m}}
\left(  G_{w}\ast b\right)  \left(  x_{i}\right)  .
\]

\subsubsection{Computation of $\boldsymbol{G}^{\left(  0+n\right)  }\left(
\left[  y_{i}\right]  _{1\leq i\leq n}\right)  $}

Now, we compute%
\begin{gather*}
\left[  \frac{d}{d\epsilon}\exp\left(  {\huge F}(J_{1},J_{2}+\epsilon
\phi)\right)  \right]  _{\epsilon=0}=\\
\exp\left(  {\huge F}(J_{1},J_{2})\right)  \left[  \frac{d}{d\epsilon}%
\exp\left(  \epsilon\left\langle G_{w}\ast\phi,\text{ }a+J_{1}\right\rangle
\right)  \right]  _{\epsilon=0}=\\
\exp\left(  {\huge F}(J_{1},J_{2})\right)  \left\langle G_{w}\ast\phi,\text{
}a+J_{1}\right\rangle =\\
\exp\left(  {\huge F}(J_{1},J_{2})\right)
{\displaystyle\int\limits_{\mathbb{Z}_{p}^{N}}}
\left(  a+J_{1}\right)  \left(  z\right)  \left(  G_{w}\ast\theta\right)
\left(  z\right)  d^{N}z=\\
\exp\left(  {\huge F}(J_{1},J_{2})\right)
{\displaystyle\int\limits_{\mathbb{Z}_{p}^{N}}}
{\displaystyle\int\limits_{\mathbb{Z}_{p}^{N}}}
G_{w}\left(  z-y\right)  \theta\left(  y\right)  \left(  a+J_{1}\right)
\left(  z\right)  d^{N}yd^{N}z=\\
\exp\left(  {\huge F}(J_{1},J_{2})\right)
{\displaystyle\int\limits_{\mathbb{Z}_{p}^{N}}}
\theta\left(  y\right)  \left\{
{\displaystyle\int\limits_{\mathbb{Z}_{p}^{N}}}
G_{w}\left(  z-y\right)  \left(  a+J_{1}\right)  \left(  z\right)
d^{N}z\right\}  d^{N}y=\\
\exp\left(  {\huge F}(J_{1},J_{2})\right)
{\displaystyle\int\limits_{\mathbb{Z}_{p}^{N}}}
\theta\left(  y\right)  \left[  \widetilde{G}_{w}\ast\left(  a+J_{1}\right)
\right]  \left(  y\right)  d^{N}y.
\end{gather*}
Thus%
\begin{equation}
\frac{\delta}{\delta_{2}J_{2}\left(  y_{i}\right)  }\exp\left(  {\huge F}%
(J_{1},J_{2})\right)  =\exp\left(  {\huge F}(J_{1},J_{2})\right)  \left(
\widetilde{G}_{w}\ast\left(  a+J_{1}\right)  \right)  \left(  y_{i}\right)  .
\label{Eq_Dervative_2}%
\end{equation}
More generally,%
\begin{equation}%
{\displaystyle\prod\limits_{i=1}^{n}}
\frac{\delta}{\delta_{2}J_{2}\left(  y_{i}\right)  }\exp\left(  {\huge F}%
(J_{1},J_{2})\right)  =\exp\left(  {\huge F}(J_{1},J_{2})\right)
{\displaystyle\prod\limits_{i=1}^{n}}
\left(  \widetilde{G}_{w}\ast\left(  a+J_{1}\right)  \right)  \left(
y_{i}\right)  \label{Formula_2}%
\end{equation}
and
\[
\boldsymbol{G}^{\left(  0+n\right)  }\left(  \left[  y_{i}\right]  _{1\leq
i\leq n}\right)  =A%
{\displaystyle\prod\limits_{i=1}^{n}}
\left(  \widetilde{G}_{w}\ast a\right)  \left(  y_{i}\right)  .
\]

\subsubsection{Computation of $\boldsymbol{G}^{\left(  m+1\right)  }\left(
\left[  x_{i}\right]  _{1\leq i\leq m},y\right)  $}

By using (\ref{Formula_1}) and
\begin{equation}
\frac{\delta}{\delta_{1}J_{1}\left(  x_{i}\right)  }H\left(  J_{1}\right)
K\left(  J_{1}\right)  =K\left(  J_{1}\right)  \frac{\delta}{\delta_{1}%
J_{1}\left(  x_{i}\right)  }H\left(  J_{1}\right)  +H\left(  J_{1}\right)
\frac{\delta}{\delta_{1}J_{1}\left(  x_{i}\right)  }K\left(  J_{1}\right)  ,
\label{Eq_Detivative_0}%
\end{equation}
one gets%
\begin{gather*}
\frac{\delta}{\delta_{2}J_{2}\left(  y\right)  }%
{\displaystyle\prod\limits_{i=1}^{m}}
\frac{\delta}{\delta_{1}J_{1}\left(  x_{i}\right)  }\exp\left(  {\huge F}%
(J_{1},J_{2})\right)  =\\
\frac{\delta}{\delta_{2}J_{2}\left(  y\right)  }\left(  \exp\left(
{\huge F}(J_{1},J_{2})\right)
{\displaystyle\prod\limits_{i=1}^{m}}
\left[  G_{w}\ast\left(  b+J_{2}\right)  \right]  \left(  x_{i}\right)
\right)  =\\%
{\displaystyle\prod\limits_{i=1}^{m}}
\left[  G_{w}\ast\left(  b+J_{2}\right)  \right]  \left(  x_{i}\right)
\frac{\delta}{\delta_{2}J_{2}\left(  y\right)  }\left(  \exp\left(
{\huge F}(J_{1},J_{2})\right)  \right)  +\\
\exp\left(  {\huge F}(J_{1},J_{2})\right)  \frac{\delta}{\delta_{2}%
J_{2}\left(  y\right)  }\left(
{\displaystyle\prod\limits_{i=1}^{m}}
\left[  G_{w}\ast\left(  b+J_{2}\right)  \right]  \left(  x_{i}\right)
\right)  .
\end{gather*}
Now, by using (\ref{Eq_Dervative_2}) and (\ref{Eq_Detivative_0}),%
\begin{gather*}
\frac{\delta}{\delta_{2}J_{2}\left(  y\right)  }%
{\displaystyle\prod\limits_{i=1}^{m}}
\frac{\delta}{\delta_{1}J_{1}\left(  x_{i}\right)  }\exp\left(  {\huge F}%
(J_{1},J_{2})\right)  =\\%
{\displaystyle\prod\limits_{i=1}^{m}}
\left[  G_{w}\ast\left(  b+J_{2}\right)  \right]  \left(  x_{i}\right)  \text{
}\left[  \widetilde{G}_{w}\ast\left(  a+J_{1}\right)  \right]  \left(
y\right)  \text{ }\exp\left(  {\huge F}(J_{1},J_{2})\right)  +\\
\exp\left(  {\huge F}(J_{1},J_{2})\right)
{\displaystyle\sum\limits_{k=1}^{m}}
G_{w}\left(  x_{k}-y\right)
{\displaystyle\prod\limits_{\substack{i=1\\i\neq k}}^{m}}
\left[  G_{w}\ast\left(  b+J_{2}\right)  \right]  \left(  x_{i}\right)  .
\end{gather*}
In the last line, we used that
\begin{equation}
\frac{\delta}{\delta_{2}J_{2}\left(  y\right)  }\left[  G\ast\left(
b+J_{2}\right)  \right]  \left(  x_{k}\right)  =G_{w}\left(  x_{k}-y\right)  .
\label{Eq_Propagator}%
\end{equation}
Therefore,%
\begin{align}
\boldsymbol{G}^{\left(  m+1\right)  }\left(  \left[  x_{i}\right]  _{1\leq
i\leq m},y\right)   &  =\left[  \widetilde{G}_{w}\ast a\right]  \left(
y\right)  \text{ }%
{\displaystyle\prod\limits_{i=1}^{m}}
\left[  G_{w}\ast b\right]  \left(  x_{i}\right)  +\label{Eq-Correlation}\\
&
{\displaystyle\sum\limits_{k=1}^{m}}
G_{w}\left(  x_{k}-y\right)
{\displaystyle\prod\limits_{\substack{i=1\\i\neq k}}^{m}}
\left[  G_{w}\ast b\right]  \left(  x_{i}\right)  .\nonumber
\end{align}

\subsubsection{$\boldsymbol{G}^{\left(  1+1\right)  }\left(  x,y\right)  $
versus $G_{w}\left(  \left\Vert x-y\right\Vert _{p}\right)  $}

By using formula (\ref{Eq-Correlation}) and assuming that the Green function
$\mathcal{F}_{\xi\rightarrow x}^{-1}\left(  \frac{1}{\widehat{w}\left(
\xi\right)  }\right)  =G_{w}\left(  \left\Vert x\right\Vert _{p}\right)  $ is
radial, we have%
\begin{equation}
\boldsymbol{G}^{\left(  1+1\right)  }\left(  x,y\right)  =\left(
\widetilde{G}_{w}\ast a\right)  \left(  x\right)  \left(  G_{w}\ast b\right)
\left(  y\right)  =\left(  G_{w}\ast a\right)  \left(  x\right)  \left(
G_{w}\ast b\right)  \left(  y\right)  . \label{Eq-Correlation-2}%
\end{equation}
We now assume that $a(x)=\Omega\left(  p^{L}\left\Vert x-\boldsymbol{x}%
_{0}\right\Vert _{p}\right)  $, where $L\in\mathbb{N\smallsetminus}\left\{
0\right\}  $, and $\boldsymbol{x}_{0}\in G_{L}^{N}=\left(  \mathbb{Z}%
_{p}/p^{l}\mathbb{Z}_{p}\right)  ^{N}$. This choice can be interpreted as
saying that the neurons (storing visible variables) in the ball
$\boldsymbol{x}_{0}+p^{L}\mathbb{Z}_{p}^{N}$ are excited while the neurons
outside of this ball are inhibited. We also take $b(x)=\Omega\left(
p^{L}\left\Vert x-\boldsymbol{y}_{0}\right\Vert _{p}\right)  $, where
$L\in\mathbb{N\smallsetminus}\left\{  0\right\}  $, and $\boldsymbol{y}_{0}\in
G_{L}^{N}$, with $\boldsymbol{x}_{0}\neq\boldsymbol{y}_{0}$. This choice has
an analog interpretation as the one given for $a(x)$.

Given a radial function $f\left(  \left\Vert x\right\Vert _{p}\right)  $
supported in $\mathbb{Z}_{p}^{N}$, and $\Omega\left(  p^{L}\left\Vert
x-\boldsymbol{x}_{0}\right\Vert _{p}\right)  $ as before,%
\begin{align*}
f\left(  \left\Vert x\right\Vert _{p}\right)  \ast\Omega\left(  p^{L}%
\left\Vert x-\boldsymbol{x}_{0}\right\Vert _{p}\right)   &  =%
{\displaystyle\sum\limits_{\boldsymbol{i}\in G_{L}^{N}}}
p^{-LN}f\left(  \left\Vert \boldsymbol{i}-\boldsymbol{x}_{0}\right\Vert
_{p}\right)  \Omega\left(  p^{L}\left\Vert x-\boldsymbol{i}\right\Vert
_{p}\right)  +\\
&  \left(
{\displaystyle\int\nolimits_{p^{L}\mathbb{Z}_{p}^{N}}}
f\left(  \left\Vert z\right\Vert _{p}\right)  d^{N}z\right)  \Omega\left(
p^{L}\left\Vert x-\boldsymbol{x}_{0}\right\Vert _{p}\right)  .
\end{align*}
By using this formula and (\ref{Eq-Correlation-2}), we have%
\[
\boldsymbol{G}^{\left(  1+1\right)  }\left(  x,y\right)  =%
{\displaystyle\sum\limits_{\boldsymbol{i}\in G_{L}^{N}}}
\text{ }%
{\displaystyle\sum\limits_{\boldsymbol{j}\in G_{L}^{N}}}
C_{\boldsymbol{i},\boldsymbol{j}}\Omega\left(  p^{L}\left\Vert
x-\boldsymbol{i}\right\Vert _{p}\right)  \Omega\left(  p^{L}\left\Vert
y-\boldsymbol{j}\right\Vert _{p}\right)  ,
\]
where%
\[
C_{\boldsymbol{i},\boldsymbol{j}}=\left\{
\begin{array}
[c]{ll}%
p^{-2LN}G_{w}(\left\Vert \boldsymbol{i}-\boldsymbol{x}_{0}\right\Vert
_{p})G_{w}(\left\Vert \boldsymbol{j}-\boldsymbol{y}_{0}\right\Vert _{p}) &
\boldsymbol{i}\neq\boldsymbol{x}_{0}\text{ and }\boldsymbol{j}\neq
\boldsymbol{y}_{0}\\
& \\
p^{-LN}\left(
{\displaystyle\int\nolimits_{p^{L}\mathbb{Z}_{p}^{N}}}
G_{w}\left(  \left\Vert z\right\Vert _{p}\right)  d^{N}z\right)
G_{w}(\left\Vert \boldsymbol{i}-\boldsymbol{x}_{0}\right\Vert _{p}) &
\boldsymbol{i}\neq\boldsymbol{x}_{0}\text{ and }\boldsymbol{j}=\boldsymbol{y}%
_{0}\\
& \\
p^{-LN}\left(
{\displaystyle\int\nolimits_{p^{L}\mathbb{Z}_{p}^{N}}}
G_{w}\left(  \left\Vert z\right\Vert _{p}\right)  d^{N}z\right)
G_{w}(\left\Vert \boldsymbol{j}-\boldsymbol{y}_{0}\right\Vert _{p}) &
\boldsymbol{i}=\boldsymbol{x}_{0}\text{ and }\boldsymbol{j}\neq\boldsymbol{y}%
_{0}\\
& \\
\left(
{\displaystyle\int\nolimits_{p^{L}\mathbb{Z}_{p}^{N}}}
G_{w}\left(  \left\Vert z\right\Vert _{p}\right)  d^{N}z\right)  ^{2} &
\boldsymbol{i}=\boldsymbol{x}_{0}\text{ and }\boldsymbol{j}=\boldsymbol{y}%
_{0}.
\end{array}
\right.
\]

\subsection{Computation of the correlation functions}

In this section, we provide a recursive formula to compute $\boldsymbol{G}%
^{\left(  m+n\right)  }\left(  \left[  x_{i}\right]  _{1\leq i\leq m},\left[
y_{i}\right]  _{1\leq i\leq n}\right)  $, for $m\geq1$, $n\geq0$, with $m\geq
n$, in terms of
\[
\left(  G_{w}\ast b\right)  \left(  x_{i}\right)  \text{, }G(x_{i}%
-y_{j})\text{, }\left(  \widetilde{G}_{w}\ast a\right)  \left(  y_{j}\right)
\text{, }1\leq i\leq m\text{, }1\leq j\leq n\text{.}%
\]
This formula is the analogue of the Wick-Isserlis Theorem. We first introduce
some formulae that we use in computation of the correlation functions.

\subsubsection{Formula 1}

Let $m,n\geq1$, and let $\left\{  r_{1},\ldots,r_{n}\right\}  $ be a finite
subset of positive integers. Then the following formula holds true:%
\begin{gather}
\left(
{\displaystyle\prod\limits_{i=1}^{n}}
\frac{\delta}{\delta_{2}J_{2}\left(  y_{r_{i}}\right)  }\right)
{\displaystyle\prod\limits_{i=1}^{m}}
\left[  G_{w}\ast\left(  b+J_{2}\right)  \right]  \left(  x_{i}\right)
=\label{Formula_3}\\
\nonumber\\
\left\{
\begin{array}
[c]{lll}%
{\displaystyle\sum\limits_{\left\{  k_{1},\ldots,k_{n}\right\}  \subseteq
\left\{  1,\ldots,m\right\}  }}
\text{ }%
{\displaystyle\prod\limits_{i=1}^{n}}
G_{w}(x_{k_{i}}-y_{r_{i}})%
{\displaystyle\prod\limits_{\substack{i=1\\i\notin\left\{  k_{1},\ldots
,k_{n}\right\}  }}^{m}}
\left[  G_{w}\ast\left(  b+J_{2}\right)  \right]  \left(  x_{i}\right)  &
\text{if} & m\geq n\\
&  & \\
0 & \text{if} & m<n.
\end{array}
\right. \nonumber
\end{gather}

This formula follows from (\ref{Eq_Propagator}) by induction.

\subsubsection{Formula 2}

For $n\geq1$,%

\begin{equation}%
{\displaystyle\prod\limits_{i=1}^{n}}
\frac{\delta}{\delta_{2}J_{2}\left(  y_{i}\right)  }A(J_{2})B(J_{2})=%
{\displaystyle\sum\limits_{k=0}^{n}}
A^{\left(  k\right)  }(J_{2})B^{\left(  n-k\right)  }(J_{2}),
\label{Formula_4}%
\end{equation}
where $A^{\left(  0\right)  }(J_{2}):=A(J_{2})$ and $B^{\left(  0\right)
}(J_{2}):=B(J_{2})$. If $k\geq1$ and $n-k\geq1$, then there are partitions of
$\left\{  1,\ldots,n\right\}  $ of the form%
\begin{equation}
\left\{  1,\ldots,n\right\}  =\left\{  i_{1},\ldots,i_{k}\right\}
{\textstyle\bigsqcup}
\left\{  j_{1},\ldots,j_{n-k}\right\}  , \label{partitions}%
\end{equation}
with $i_{1}<\ldots<i_{k}$\ \ and $j_{1}<\ldots<j_{n-k}$. In the cases $k=0$,
$n$ there are only trivial partitions: $\left\{  1,\ldots,n\right\}
=\varnothing%
{\displaystyle\bigsqcup}
\left\{  1,\ldots,n\right\}  =\left\{  1,\ldots,n\right\}
{\textstyle\bigsqcup}
\varnothing$. With this notation%
\[
A^{\left(  k\right)  }(J_{2})B^{\left(  n-k\right)  }(J_{2})=%
{\displaystyle\sum\limits_{\substack{\left\{  i_{1},\ldots,i_{k}\right\}
\\\left\{  j_{1},\ldots,j_{n-k}\right\}  }}}
\left(
{\displaystyle\prod\limits_{r=1}^{k}}
\frac{\delta}{\delta_{2}J_{2}\left(  y_{i_{r}}\right)  }A(J_{2})\right)
\left(
{\displaystyle\prod\limits_{r=1}^{n-k}}
\frac{\delta}{\delta_{2}J_{2}\left(  y_{j_{r}}\right)  }B(J_{2})\right)  ,
\]
where the sum runs over all partition of the form (\ref{partitions}),
including the trivial ones.

\subsubsection{Main calculation}

By using (\ref{Formula_1})-(\ref{Formula_4}),%

\begin{gather}%
{\displaystyle\prod\limits_{i=1}^{n}}
\frac{\delta}{\delta_{2}J_{2}\left(  y_{i}\right)  }%
{\displaystyle\prod\limits_{i=1}^{m}}
\frac{\delta}{\delta_{1}J_{1}\left(  x_{i}\right)  }\exp\left(  {\huge F}%
(J_{1},J_{2})\right)  =\label{Formula_5}\\
\left(
{\displaystyle\prod\limits_{i=1}^{n}}
\frac{\delta}{\delta_{2}J_{2}\left(  y_{i}\right)  }\right)  \exp\left(
{\huge F}(J_{1},J_{2})\right)
{\displaystyle\prod\limits_{i=1}^{m}}
\left[  G_{w}\ast\left(  b+J_{2}\right)  \right]  \left(  x_{i}\right)
=\nonumber\\%
{\displaystyle\sum\limits_{k=0}^{n}}
(\exp\left(  {\huge F}(J_{1},J_{2})\right)  ^{\left(  k\right)  })(%
{\displaystyle\prod\limits_{i=1}^{m}}
\left[  G_{w}\ast\left(  b+J_{2}\right)  \right]  \left(  x_{i}\right)
)^{\left(  n-k\right)  },\nonumber
\end{gather}
where%
\begin{gather}
(\exp\left(  {\huge F}(J_{1},J_{2})\right)  ^{\left(  k\right)  })(%
{\displaystyle\prod\limits_{i=1}^{m}}
\left[  G_{w}\ast\left(  b+J_{2}\right)  \right]  \left(  x_{i}\right)
)^{\left(  n-k\right)  }=\label{Formula_6}\\%
{\displaystyle\sum\limits_{\substack{\left\{  i_{1},\ldots,i_{k}\right\}
\\\left\{  j_{1},\ldots,j_{n-k}\right\}  }}}
\left(
{\displaystyle\prod\limits_{r=1}^{k}}
\frac{\delta}{\delta_{2}J_{2}\left(  y_{i_{r}}\right)  }\exp\left(
{\huge F}(J_{1},J_{2})\right)  \right)  \left(
{\displaystyle\prod\limits_{r=1}^{n-k}}
\frac{\delta}{\delta_{2}J_{2}\left(  y_{j_{r}}\right)  }%
{\displaystyle\prod\limits_{i=1}^{m}}
\left[  G_{w}\ast\left(  b+J_{2}\right)  \right]  \left(  x_{i}\right)
\right)  ,\nonumber
\end{gather}
and by (\ref{Formula_2}),%
\begin{gather}%
{\displaystyle\prod\limits_{r=1}^{k}}
\frac{\delta}{\delta_{2}J_{2}\left(  y_{i_{r}}\right)  }\exp\left(
{\huge F}(J_{1},J_{2})\right)  =\label{Formula_7}\\
\exp\left(  {\huge F}(J_{1},J_{2})\right)
{\displaystyle\prod\limits_{i=1}^{k}}
\left(  \widetilde{G}_{w}\ast\left(  a+J_{1}\right)  \right)  \left(
y_{i_{r}}\right)  ,\nonumber
\end{gather}
and by (\ref{Formula_3}),%
\begin{gather}%
{\displaystyle\prod\limits_{r=1}^{n-k}}
\frac{\delta}{\delta_{2}J_{2}\left(  y_{j_{r}}\right)  }%
{\displaystyle\prod\limits_{i=1}^{m}}
\left[  G_{w}\ast\left(  b+J_{2}\right)  \right]  \left(  x_{i}\right)
=\label{Formula_8}\\
\nonumber\\
\left\{
\begin{array}
[c]{ll}%
{\displaystyle\sum\limits_{\left\{  j_{1},\ldots,j_{n-k}\right\}
\subseteq\left\{  1,\ldots,m\right\}  }}
\text{ }%
{\displaystyle\prod\limits_{i=1}^{n-k}}
G_{w}(x_{j_{i}}-y_{j_{i}})%
{\displaystyle\prod\limits_{\substack{i=1\\i\notin\left\{  j_{1}%
,\ldots,j_{n-k}\right\}  }}^{m}}
\left[  G_{w}\ast\left(  b+J_{2}\right)  \right]  \left(  x_{i}\right)  &
\text{if }m\geq n-k\\
& \\
0 & \text{if }m<n-k.
\end{array}
\right. \nonumber
\end{gather}
Therefore,%
\begin{gather}%
{\displaystyle\sum\limits_{\substack{\left\{  i_{1},\ldots,i_{k}\right\}
\\\left\{  j_{1},\ldots,j_{n-k}\right\}  }}}
\left(
{\displaystyle\prod\limits_{r=1}^{k}}
\frac{\delta}{\delta_{2}J_{2}\left(  y_{i_{r}}\right)  }\exp\left(
{\huge F}(J_{1},J_{2})\right)  \right)  \left(
{\displaystyle\prod\limits_{r=1}^{n-k}}
\frac{\delta}{\delta_{2}J_{2}\left(  y_{j_{r}}\right)  }%
{\displaystyle\prod\limits_{i=1}^{m}}
\left[  G_{w}\ast\left(  b+J_{2}\right)  \right]  \left(  x_{i}\right)
\right) \label{Formula_9}\\
=%
{\displaystyle\sum\limits_{\substack{\left\{  i_{1},\ldots,i_{k}\right\}
\\\left\{  j_{1},\ldots,j_{n-k}\right\}  }}}
\left[  \exp\left(  {\huge F}(J_{1},J_{2})\right)
{\displaystyle\prod\limits_{i=1}^{k}}
\left(  \widetilde{G}_{w}\ast\left(  a+J_{1}\right)  \right)  \left(
y_{i_{r}}\right)  \right]  \times\nonumber\\
\left[
{\displaystyle\sum\limits_{\left\{  j_{1},\ldots,j_{n-k}\right\}
\subseteq\left\{  1,\ldots,m\right\}  }}
\text{ }%
{\displaystyle\prod\limits_{i=1}^{n-k}}
G_{w}(x_{j_{i}}-y_{j_{i}})%
{\displaystyle\prod\limits_{\substack{i=1\\i\notin\left\{  j_{1}%
,\ldots,j_{n-k}\right\}  }}^{m}}
\left[  G_{w}\ast\left(  b+J_{2}\right)  \right]  \left(  x_{i}\right)
\right]  ,\nonumber
\end{gather}
where we assume that the last factor is zero if $m<n-k$, and that $\left\{
1,\ldots,n\right\}  =\left\{  i_{1},\ldots,i_{k}\right\}
{\textstyle\bigsqcup}
\left\{  j_{1},\ldots,j_{n-k}\right\}  $.

In order to give an explicit formula for $\boldsymbol{G}^{\left(  m+n\right)
}\left(  \left[  x_{i}\right]  _{1\leq i\leq m},\left[  y_{i}\right]  _{1\leq
i\leq n}\right)  $, we recast formulae (\ref{Formula_5})-(\ref{Formula_9})
into a combinatorial framework.

\begin{definition}
Let $\mathbb{I}$, $\mathbb{J}$ be two finite subsets of positive integers. A
$3$-partition of $\mathbb{I}$ subordinated to $\mathbb{J}$ is determined by
three subsets $(A,B,C)$, where $A\subseteq\mathbb{I}$, and $B$, $C$
$\subseteq\mathbb{J}$. These subsets are determined as follows:

\begin{enumerate}
\item If $\mathbb{I=}\left\{  1,\ldots,n\right\}  $, $\mathbb{J=}\left\{
1,\ldots,m\right\}  $ with $m,n\geq1$. Given any $A\subseteq\mathbb{I}$ with
cardinality $\#A=k$, if $n-k\leq m$, then $B$ is a subset of $\mathbb{J}$ with
cardinality $n-k$ and $C$ is the complement of $B$\ in $\mathbb{J}$, i.e.,
$\mathbb{J}=B%
{\textstyle\bigsqcup}
C$.

\item If $\mathbb{I=\emptyset}$, $\mathbb{J=}\left\{  1,\ldots,m\right\}  $
with $m\geq1$, then $B$ $=\mathbb{J}$, $A$, $C\mathbb{=\emptyset}$.

\item If $\mathbb{I=}1,\ldots,n$, $\mathbb{J}=\emptyset$, then $A=\mathbb{I}$,
$B=C=\varnothing$.
\end{enumerate}
\end{definition}

\begin{definition}
Given $A$, $B$, $C$ a $3$-partition of $\mathbb{I=}\left\{  1,\ldots
,n\right\}  $ subordinated to $\mathbb{J=}\left\{  1,\ldots,m\right\}  $, we
define%
\begin{align*}
\boldsymbol{G}_{(A,B,C)}^{\left(  m+n\right)  }\left(  \left[  x_{i}\right]
_{1\leq i\leq m},\left[  y_{i}\right]  _{1\leq i\leq n}\right)   &  =\\
&
{\displaystyle\prod\limits_{j\in A}}
\left(  \widetilde{G}_{w}\ast a\right)  \left(  y_{j}\right)
{\displaystyle\prod\limits_{j\in B}}
G_{w}\left(  x_{j}-y_{j}\right)
{\displaystyle\prod\limits_{j\in C}}
\left(  G_{w}\ast b\right)  \left(  x_{j}\right)
\end{align*}
as a distribution from $\mathcal{D}^{\prime}(\mathbb{Z}_{p}^{\left(
m+n\right)  })$. Here we use the convention $%
{\textstyle\prod\nolimits_{j\in\emptyset}}
$ $=1$.
\end{definition}

\begin{theorem}
\label{Theorem3}With the above notation,%
\[
\boldsymbol{G}^{\left(  m+n\right)  }\left(  \left[  x_{i}\right]  _{1\leq
i\leq m},\left[  y_{i}\right]  _{1\leq i\leq n}\right)  =%
{\displaystyle\sum\limits_{\left(  A,B,C\right)  }}
\boldsymbol{G}_{(A,B,C)}^{\left(  m+n\right)  }\left(  \left[  x_{i}\right]
_{1\leq i\leq m},\left[  y_{i}\right]  _{1\leq i\leq n}\right)  ,
\]
where $\left(  A,B,C\right)  $ runs over all the $3$-partitions of
$\mathbb{I}$ subordinated to $\mathbb{J}$.
\end{theorem}

\section{\label{Section_9}$\left\{  \boldsymbol{v},\boldsymbol{h}\right\}
^{4}$-deep Boltzmann machines}

In this section, we compute the correlation functions for DBMs having energy
functionals of the form
\begin{gather}
E(\boldsymbol{v},\boldsymbol{h};\boldsymbol{\theta})=-\left\langle
\boldsymbol{v},a\right\rangle -\left\langle \boldsymbol{h},b\right\rangle -%
{\displaystyle\iint\limits_{\mathbb{Z}_{p}^{N}\times\mathbb{Z}_{p}^{N}}}
\boldsymbol{h}\left(  x\right)  w\left(  x-y\right)  \boldsymbol{v}\left(
y\right)  d^{N}yd^{N}x+\label{energy_functional_ultimo}\\
c%
{\displaystyle\int\limits_{\mathbb{Z}_{p}^{N}}}
\boldsymbol{v}^{4}\left(  x\right)  d^{N}x+d%
{\displaystyle\int\limits_{\mathbb{Z}_{p}^{N}}}
\boldsymbol{h}^{4}\left(  x\right)  d^{N}x=E^{\text{spin}}(\boldsymbol{v}%
,\boldsymbol{h};\boldsymbol{\theta})-c\left\langle 1,\boldsymbol{v}%
^{4}\right\rangle -d\left\langle 1,\boldsymbol{h}^{4}\right\rangle ,\nonumber
\end{gather}
where $a,b,w\in L_{\mathbb{R}}^{\infty}\left(  \mathbb{Z}_{p}^{N}\right)  $,
$c,d\in\mathbb{R}$, and $\boldsymbol{\theta}=\left(  w,a,b,c,d\right)  $. In
practical applications the parameters $\boldsymbol{\theta}=\left(
w,a,b,c,d\right)  $ of a discretization of $E(\boldsymbol{v},\boldsymbol{h}%
;\boldsymbol{\theta})$ are tuned by using a stochastic gradient descent
algorithm, and thus the sign of the parameters $c,d$ cannot be fixed a priory.
For this reason, in this section, we assume that $\boldsymbol{v}%
,\boldsymbol{h}\in\boldsymbol{B}_{M}^{\left(  \infty\right)  }=\left\{  f\in
L^{\infty}(\mathbb{Z}_{p}^{N});\left\Vert f\right\Vert _{\infty}\leq
M\right\}  $. We set%
\[
\mathcal{Z}^{\left(  \infty\right)  }:=\mathcal{Z}^{\left(  \infty\right)
}\left(  \boldsymbol{\theta}\right)  =%
{\displaystyle\iint\limits_{\boldsymbol{B}_{M}^{\left(  \infty\right)  }%
\times\boldsymbol{B}_{M}^{\left(  \infty\right)  }}}
e^{-E(\boldsymbol{v},\boldsymbol{h};\boldsymbol{\theta})}d\mathbb{P}_{K_{1}%
}\left(  \boldsymbol{v}\right)  \otimes d\mathbb{P}_{K_{2}}\left(
\boldsymbol{h}\right)
\]
The generating functional attached to $E(\boldsymbol{v},\boldsymbol{h}%
;\boldsymbol{\theta})$ is
\begin{multline*}
Z(J_{1},J_{2};c,d):=Z(J_{1},J_{2};\boldsymbol{\theta})=\frac{1}{\mathcal{Z}%
^{\left(  \infty\right)  }}%
{\displaystyle\iint\limits_{\boldsymbol{B}_{M}^{\left(  \infty\right)  }%
\times\boldsymbol{B}_{M}^{\left(  \infty\right)  }}}
e^{-E^{\text{spin}}(\boldsymbol{v},\boldsymbol{h};\boldsymbol{\theta
})-c\left\langle 1,\boldsymbol{v}^{4}\right\rangle -d\left\langle
1,\boldsymbol{h}^{4}\right\rangle }\times\\
e^{\left\langle \boldsymbol{v},J_{1}\right\rangle +\left\langle \boldsymbol{h}%
,J_{2}\right\rangle }d\mathbb{P}_{K_{1}}\left(  \boldsymbol{v}\right)  \otimes
d\mathbb{P}_{K_{2}}\left(  \boldsymbol{h}\right)  .
\end{multline*}
Since
\[
\exp\left(  E^{\text{spin}}(\boldsymbol{v},\boldsymbol{h};\boldsymbol{\theta
})-c\left\langle 1,\boldsymbol{v}^{4}\right\rangle -d\left\langle
1,\boldsymbol{h}^{4}\right\rangle +\left\langle \boldsymbol{v},J_{1}%
\right\rangle +\left\langle \boldsymbol{h},J_{2}\right\rangle \right)
\]
is integrable, cf. Lemma \ref{Lemma5}, by using the dominated convergence
theorem and the Taylor expansion of the functions $e^{-c\left\langle
1,\boldsymbol{v}^{4}\right\rangle }$, $e^{-d\left\langle 1,\boldsymbol{h}%
^{4}\right\rangle }$, one gets that%
\begin{gather*}
Z(J_{1},J_{2};c,d)=\frac{1}{\mathcal{Z}^{\left(  \infty\right)  }}%
{\displaystyle\iint\limits_{\boldsymbol{B}_{M}^{\left(  \infty\right)  }%
\times\boldsymbol{B}_{M}^{\left(  \infty\right)  }}}
e^{-E^{\text{spin}}(\boldsymbol{v},\boldsymbol{h};\boldsymbol{\theta
})+\left\langle \boldsymbol{v},J_{1}\right\rangle +\left\langle \boldsymbol{h}%
,J_{2}\right\rangle }d\mathbb{P}_{K_{1}}\left(  \boldsymbol{v}\right)  \otimes
d\mathbb{P}_{K_{2}}\left(  \boldsymbol{h}\right)  -\\
\frac{c}{\mathcal{Z}^{\left(  \infty\right)  }}%
{\displaystyle\iint\limits_{\boldsymbol{B}_{M}^{\left(  \infty\right)  }%
\times\boldsymbol{B}_{M}^{\left(  \infty\right)  }}}
e^{-E^{\text{spin}}(\boldsymbol{v},\boldsymbol{h};\boldsymbol{\theta
})+\left\langle \boldsymbol{v},J_{1}\right\rangle +\left\langle \boldsymbol{h}%
,J_{2}\right\rangle }\left\{
{\displaystyle\int\limits_{\mathbb{Z}_{p}^{N}}}
\boldsymbol{v}^{4}\left(  r_{1}\right)  d^{N}r_{1}\right\}  \text{
}d\mathbb{P}_{K_{1}}\left(  \boldsymbol{v}\right)  \otimes d\mathbb{P}_{K_{2}%
}\left(  \boldsymbol{h}\right)  -\\
\frac{d}{\mathcal{Z}^{\left(  \infty\right)  }}%
{\displaystyle\iint\limits_{\boldsymbol{B}_{M}^{\left(  \infty\right)  }%
\times\boldsymbol{B}_{M}^{\left(  \infty\right)  }}}
e^{-E^{\text{spin}}(\boldsymbol{v},\boldsymbol{h};\boldsymbol{\theta
})+\left\langle \boldsymbol{v},J_{1}\right\rangle +\left\langle \boldsymbol{h}%
,J_{2}\right\rangle }\left\{
{\displaystyle\int\limits_{\mathbb{Z}_{p}^{N}}}
\boldsymbol{h}^{4}\left(  s_{1}\right)  d^{N}s_{1}\right\}  \text{
}d\mathbb{P}_{K_{1}}\left(  \boldsymbol{v}\right)  \otimes d\mathbb{P}_{K_{2}%
}\left(  \boldsymbol{h}\right)  +\\
\frac{1}{\mathcal{Z}^{\left(  \infty\right)  }}%
{\displaystyle\sum\limits_{i=1}^{\infty}}
{\displaystyle\sum\limits_{j=1}^{\infty}}
\frac{\left(  -1\right)  ^{i+j}c^{i}d^{j}}{i!j!}%
{\displaystyle\iint\limits_{\boldsymbol{B}_{M}^{\left(  \infty\right)  }%
\times\boldsymbol{B}_{M}^{\left(  \infty\right)  }}}
\left\{  e^{-E^{\text{spin}}(\boldsymbol{v},\boldsymbol{h};\boldsymbol{\theta
})+\left\langle \boldsymbol{v},J_{1}\right\rangle +\left\langle \boldsymbol{h}%
,J_{2}\right\rangle }\times\right. \\
\left.  \left(
{\displaystyle\prod_{k=1}^{i}}
\text{ }%
{\displaystyle\int\limits_{\mathbb{Z}_{p}^{N}}}
\boldsymbol{v}^{4}\left(  r_{k}\right)  d^{N}x_{k}\right)  \left(
{\displaystyle\prod_{k=1}^{j}}
\text{ }%
{\displaystyle\int\limits_{\mathbb{Z}_{p}^{N}}}
\boldsymbol{h}^{4}\left(  s_{k}\right)  d^{N}y_{k}\right)  \right\}
d\mathbb{P}_{K_{1}}\left(  \boldsymbol{v}\right)  \otimes d\mathbb{P}_{K_{2}%
}\left(  \boldsymbol{h}\right)
\end{gather*}%
\[
=:\frac{\mathcal{Z}_{\text{spin}}^{\left(  \infty\right)  }}{\mathcal{Z}%
^{\left(  \infty\right)  }}%
{\displaystyle\sum\limits_{i=0}^{\infty}}
{\displaystyle\sum\limits_{j=0}^{\infty}}
\frac{\left(  -1\right)  ^{i+j}c^{i}d^{j}}{i!j!}Z_{i,j}^{\text{spin}}%
(J_{1},J_{2};\boldsymbol{\theta}),
\]
where in $Z_{i,j}^{\text{spin}}(J_{1},J_{2};\boldsymbol{\theta})$,
$\boldsymbol{\theta}=(w,a,b,0,0)$, and
\[
\mathcal{Z}_{\text{spin}}^{\left(  \infty\right)  }=%
{\displaystyle\iint\limits_{\boldsymbol{B}_{M}^{\left(  \infty\right)  }%
\times\boldsymbol{B}_{M}^{\left(  \infty\right)  }}}
e^{\left\langle \boldsymbol{h},\boldsymbol{Wv}\right\rangle }d\mathbb{P}%
_{K_{1}}\left(  \boldsymbol{v}\right)  \otimes d\mathbb{P}_{K_{2}}\left(
\boldsymbol{h}\right)  .
\]

All the results given in Section \ref{Section_Spin_Glassses}, in particular
Theorem \ref{Theorem3}, are valid for%
\[
Z_{0,0}^{\text{spin}}(J_{1},J_{2};\boldsymbol{\theta})=\frac{1}{\mathcal{Z}%
^{\left(  \infty\right)  }}%
{\displaystyle\iint\limits_{\boldsymbol{B}_{M}^{\left(  \infty\right)  }%
\times\boldsymbol{B}_{M}^{\left(  \infty\right)  }}}
e^{-E^{\text{spin}}(\boldsymbol{v},\boldsymbol{h};\boldsymbol{\theta
})+\left\langle \boldsymbol{v},J_{1}\right\rangle +\left\langle \boldsymbol{h}%
,J_{2}\right\rangle }d\mathbb{P}_{K_{1}}\left(  \boldsymbol{v}\right)  \otimes
d\mathbb{P}_{K_{2}}\left(  \boldsymbol{h}\right)  .
\]
It is only necessary to change $\mathcal{Z}^{\left(  2\right)  }$\ by
$\mathcal{Z}^{\left(  \infty\right)  }$, and $\boldsymbol{B}_{M}^{\left(
2\right)  }\times\boldsymbol{B}_{M}^{\left(  2\right)  }$ by $\boldsymbol{B}%
_{M}^{\left(  \infty\right)  }\times\boldsymbol{B}_{M}^{\left(  \infty\right)
}$.\ We denote the correlation functions attached to $Z_{0,0}^{\text{spin}%
}(J_{1},J_{2};\boldsymbol{\theta})$ as $\boldsymbol{G}^{\left(  m+n\right)
}\left(  x,y;0,0\right)  $, $x\in\mathbb{Z}_{p}^{Nm}$, $y\in\mathbb{Z}%
_{p}^{Nn}$.

By using Fubini's theorem , one gets that%
\begin{multline*}
Z_{1,0}^{\text{spin}}(J_{1},J_{2};\boldsymbol{\theta})=\frac{1}{\mathcal{Z}%
_{\text{spin}}^{\left(  \infty\right)  }}%
{\displaystyle\int\limits_{\mathbb{Z}_{p}^{N}}}
\text{ \ }%
{\displaystyle\iint\limits_{\boldsymbol{B}_{M}^{\left(  \infty\right)  }%
\times\boldsymbol{B}_{M}^{\left(  \infty\right)  }}}
e^{-E^{\text{spin}}(\boldsymbol{v},\boldsymbol{h};\boldsymbol{\theta
})+\left\langle \boldsymbol{v},J_{1}\right\rangle +\left\langle \boldsymbol{h}%
,J_{2}\right\rangle }\boldsymbol{v}^{4}\left(  r_{1}\right)  \times\\
d\mathbb{P}_{K_{1}}\left(  \boldsymbol{v}\right)  \otimes d\mathbb{P}_{K_{2}%
}\left(  \boldsymbol{h}\right)  d^{N}r_{1}=%
{\displaystyle\int\limits_{\mathbb{Z}_{p}^{N}}}
\text{ }\boldsymbol{G}^{\left(  4+0\right)  }\left(  r_{1},r_{1},r_{1}%
,r_{1};0,0\right)  d^{N}r_{1}\text{,}%
\end{multline*}%
\begin{multline*}
Z_{0,1}^{\text{spin}}(J_{1},J_{2};\boldsymbol{\theta})=\frac{1}{\mathcal{Z}%
_{\text{spin}}^{\left(  \infty\right)  }}%
{\displaystyle\int\limits_{\mathbb{Z}_{p}^{N}}}
\text{ \ }%
{\displaystyle\iint\limits_{\boldsymbol{B}_{M}^{\left(  \infty\right)  }%
\times\boldsymbol{B}_{M}^{\left(  \infty\right)  }}}
e^{-E^{\text{spin}}(\boldsymbol{v},\boldsymbol{h};\boldsymbol{\theta
})+\left\langle \boldsymbol{v},J_{1}\right\rangle +\left\langle \boldsymbol{h}%
,J_{2}\right\rangle }\boldsymbol{h}^{4}\left(  s_{1}\right)  \text{ }\times\\
d\mathbb{P}_{K_{1}}\left(  \boldsymbol{v}\right)  \otimes d\mathbb{P}_{K_{2}%
}\left(  \boldsymbol{h}\right)  \text{ \ }d^{N}s_{1}=%
{\displaystyle\int\limits_{\mathbb{Z}_{p}^{N}}}
\text{ }\boldsymbol{G}^{\left(  0+4\right)  }\left(  s_{1},s_{1},s_{1}%
,s_{1};0,0\right)  d^{N}s_{1},
\end{multline*}%
\begin{gather*}
Z_{\left(  i,j\right)  }^{\text{spin}}(J_{1},J_{2};\boldsymbol{\theta}%
)=\frac{1}{\mathcal{Z}^{\left(  \infty\right)  }}%
{\displaystyle\int\limits_{\mathbb{Z}_{p}^{iN}}}
\text{ \ }%
{\displaystyle\int\limits_{\mathbb{Z}_{p}^{jN}}}
\text{ \ }%
{\displaystyle\iint\limits_{\boldsymbol{B}_{M}^{\left(  \infty\right)  }%
\times\boldsymbol{B}_{M}^{\left(  \infty\right)  }}}
e^{-E^{\text{spin}}(\boldsymbol{v},\boldsymbol{h};\boldsymbol{\theta
})+\left\langle \boldsymbol{v},J_{1}\right\rangle +\left\langle \boldsymbol{h}%
,J_{2}\right\rangle }\times\\
\left(
{\displaystyle\prod_{k=1}^{i}}
\text{ }\boldsymbol{v}^{4}\left(  r_{k}\right)  \right)  \left(
{\displaystyle\prod_{k=1}^{j}}
\text{ }\boldsymbol{h}^{4}\left(  s_{k}\right)  \right)  d\mathbb{P}_{K_{1}%
}\left(  \boldsymbol{v}\right)  \otimes d\mathbb{P}_{K_{2}}\left(
\boldsymbol{h}\right)  \text{ \ }%
{\displaystyle\prod_{k=1}^{i}}
d^{N}r_{k}%
{\displaystyle\prod_{k=1}^{j}}
d^{N}s_{k}.
\end{gather*}
By using the notation
\[
H\left(  4\left[  x_{i}\right]  _{1\leq i\leq m},4\left[  y_{i}\right]
_{1\leq i\leq n}\right)  :=H\left(  \ldots,x_{i},x_{i},x_{i},x_{i}%
,\ldots,y_{i},y_{i},y_{i},y_{i},\ldots\right)  ,
\]
we have%
\[
Z_{m,n}^{\text{spin}}(0,0;\boldsymbol{\theta})=%
{\displaystyle\int\limits_{\mathbb{Z}_{p}^{mN}}}
\text{ \ }%
{\displaystyle\int\limits_{\mathbb{Z}_{p}^{nN}}}
\text{ \ }\boldsymbol{G}^{\left(  4m+4n\right)  }\left(  4\left[
r_{i}\right]  _{1\leq i\leq m},4\left[  s_{i}\right]  _{1\leq i\leq
n};0,0\right)
{\displaystyle\prod_{k=1}^{m}}
d^{N}x_{k}%
{\displaystyle\prod_{k=1}^{n}}
d^{N}y_{i}.
\]
We denote by $\boldsymbol{G}^{\left(  m+n\right)  }\left(  x,y;c,d\right)  $,
$x\in\mathbb{Z}_{p}^{Nm}$, $y\in\mathbb{Z}_{p}^{Nn}$ the correlation functions
attached to $Z(J_{1},J_{2};c,d)$. By using Lemma \ref{Lemma6A},%
\[
\boldsymbol{G}^{\left(  m+n\right)  }\left(  x,y;c,d\right)  =\left[
{\textstyle\prod\limits_{j=1}^{n}}
\frac{\delta}{\delta_{2}J_{2}\left(  y_{j}\right)  }\text{ }%
{\textstyle\prod\limits_{i=1}^{m}}
\frac{\delta}{\delta_{1}J_{1}\left(  x_{i}\right)  }Z(J_{1},J_{2};c,d)\right]
_{\substack{J_{0}=0\\J_{1}=0}}.
\]
In order to compute $\boldsymbol{G}^{\left(  m+n\right)  }\left(
x,y;c,d\right)  $, we need some preliminary results.

\begin{lemma}
\label{Lemma8}Set
\[
\mathcal{G}(J_{1},J_{2}):=\left(
{\displaystyle\sum\limits_{i=1}^{\infty}}
\text{ }%
{\displaystyle\sum\limits_{j=1}^{\infty}}
\frac{\left(  -1\right)  ^{i+j}c^{i}d^{j}}{i!j!}Z_{i,j}^{\text{spin}}%
(J_{1},J_{2};\boldsymbol{\theta})\right)  .
\]
Then%
\begin{gather*}
\left[  \frac{d}{d\epsilon}\mathcal{G}(J_{1}+\epsilon\phi,J_{2})\right]
_{\epsilon=0}=\\
=%
{\displaystyle\sum\limits_{i=1}^{\infty}}
\text{ }%
{\displaystyle\sum\limits_{j=1}^{\infty}}
\frac{\left(  -1\right)  ^{i+j}c^{i}d^{j}}{i!j!}\frac{\mathcal{Z}%
_{\text{spin}}^{\left(  \infty\right)  }}{\mathcal{Z}^{\left(  \infty\right)
}}%
{\displaystyle\int\limits_{\mathbb{Z}_{p}^{iN}}}
\text{ \ }%
{\displaystyle\int\limits_{\mathbb{Z}_{p}^{jN}}}
\text{ \ }%
{\displaystyle\iint\limits_{\boldsymbol{B}_{M}^{\left(  \infty\right)  }%
\times\boldsymbol{B}_{M}^{\left(  \infty\right)  }}}
e^{-E^{\text{spin}}(\boldsymbol{v},\boldsymbol{h};\boldsymbol{\theta
})+\left\langle \boldsymbol{v},J_{1}\right\rangle +\left\langle \boldsymbol{h}%
,J_{2}\right\rangle }\left\langle \boldsymbol{v},\phi\right\rangle \times\\
\left(
{\displaystyle\prod_{k=1}^{i}}
\text{ }\boldsymbol{v}^{4}\left(  r_{k}\right)  \right)  \left(
{\displaystyle\prod_{k=1}^{j}}
\text{ }\boldsymbol{h}^{4}\left(  s_{k}\right)  \right)  d\mathbb{P}_{K_{1}%
}\left(  \boldsymbol{v}\right)  \otimes d\mathbb{P}_{K_{2}}\left(
\boldsymbol{h}\right)  \text{ \ }%
{\displaystyle\prod_{k=1}^{i}}
d^{N}r_{k}%
{\displaystyle\prod_{k=1}^{j}}
d^{N}s_{k},
\end{gather*}
and%
\begin{multline*}
\frac{\delta}{\delta_{1}J_{1}\left(  x_{i}\right)  }\left(
{\displaystyle\sum\limits_{i=1}^{\infty}}
\text{ }%
{\displaystyle\sum\limits_{j=1}^{\infty}}
\frac{\left(  -1\right)  ^{i+j}c^{i}d^{j}}{i!j!}Z_{i,j}^{\text{spin}}%
(J_{1},J_{2};\boldsymbol{\theta})\right)  =\\
=%
{\displaystyle\sum\limits_{i=1}^{\infty}}
\text{ }%
{\displaystyle\sum\limits_{j=1}^{\infty}}
\frac{\left(  -1\right)  ^{i+j}c^{i}d^{j}}{i!j!}\text{ \ }\frac{\mathcal{Z}%
_{\text{spin}}^{\left(  \infty\right)  }}{\mathcal{Z}^{\left(  \infty\right)
}}%
{\displaystyle\int\limits_{\mathbb{Z}_{p}^{iN}}}
\text{ \ }%
{\displaystyle\int\limits_{\mathbb{Z}_{p}^{jN}}}
\text{ \ }%
{\displaystyle\iint\limits_{\boldsymbol{B}_{M}^{\left(  \infty\right)  }%
\times\boldsymbol{B}_{M}^{\left(  \infty\right)  }}}
e^{-E^{\text{spin}}(\boldsymbol{v},\boldsymbol{h};\boldsymbol{\theta
})+\left\langle \boldsymbol{v},J_{1}\right\rangle +\left\langle \boldsymbol{h}%
,J_{2}\right\rangle }\times\\
\boldsymbol{v}\left(  x_{i}\right)  \left(
{\displaystyle\prod_{k=1}^{i}}
\text{ }\boldsymbol{v}^{4}\left(  r_{k}\right)  \right)  \left(
{\displaystyle\prod_{k=1}^{j}}
\text{ }\boldsymbol{h}^{4}\left(  s_{k}\right)  \right)  d\mathbb{P}_{K_{1}%
}\left(  \boldsymbol{v}\right)  \otimes d\mathbb{P}_{K_{2}}\left(
\boldsymbol{h}\right)  \text{ \ }%
{\displaystyle\prod_{k=1}^{i}}
d^{N}r_{k}%
{\displaystyle\prod_{k=1}^{j}}
d^{N}s_{k}.
\end{multline*}

\end{lemma}

\begin{proof}
We first compute%
\begin{gather}
\left[  \frac{d}{d\epsilon}\mathcal{G}(J_{1}+\epsilon\phi,J_{2})\right]
_{\epsilon=0}=\lim_{\epsilon\rightarrow0}\frac{\mathcal{G}(J_{1}+\epsilon
\phi,J_{2})-\mathcal{G}(J_{1},J_{2})}{\epsilon}=\nonumber\\
\nonumber\\
\lim_{\epsilon\rightarrow0}%
{\displaystyle\sum\limits_{i=1}^{\infty}}
\text{ }%
{\displaystyle\sum\limits_{j=1}^{\infty}}
\frac{\left(  -1\right)  ^{i+j}c^{i}d^{j}}{i!j!}\left[  \frac{Z_{i,j}%
^{\text{spin}}(J_{1}+\epsilon\phi,J_{2};\boldsymbol{\theta})-Z_{i,j}%
^{\text{spin}}(J_{1},J_{2};\boldsymbol{\theta})}{\epsilon}\right]  .
\label{limit}%
\end{gather}
On the other hand,%
\begin{gather*}
\frac{Z_{i,j}^{\text{spin}}(J_{1}+\epsilon\phi,J_{2};\boldsymbol{\theta
})-Z_{i,j}^{\text{spin}}(J_{1},J_{2};\boldsymbol{\theta})}{\epsilon}=\\
\frac{\mathcal{Z}_{\text{spin}}^{\left(  \infty\right)  }}{\mathcal{Z}%
^{\left(  \infty\right)  }}%
{\displaystyle\int\limits_{\mathbb{Z}_{p}^{iN}}}
\text{ \ }%
{\displaystyle\int\limits_{\mathbb{Z}_{p}^{jN}}}
\text{ \ }%
{\displaystyle\iint\limits_{\boldsymbol{B}_{M}^{\left(  \infty\right)  }%
\times\boldsymbol{B}_{M}^{\left(  \infty\right)  }}}
e^{-E^{\text{spin}}(\boldsymbol{v},\boldsymbol{h};\boldsymbol{\theta
})+\left\langle \boldsymbol{v},J_{1}\right\rangle +\left\langle \boldsymbol{h}%
,J_{2}\right\rangle }\left(  \frac{e^{\epsilon\left\langle \boldsymbol{v}%
,\phi\right\rangle }-1}{\epsilon}\right)  \times\\
\left(
{\displaystyle\prod_{k=1}^{i}}
\text{ }\boldsymbol{v}^{4}\left(  r_{k}\right)  \right)  \left(
{\displaystyle\prod_{k=1}^{j}}
\text{ }\boldsymbol{h}^{4}\left(  s_{k}\right)  \right)  d\mathbb{P}_{K_{1}%
}\left(  \boldsymbol{v}\right)  \otimes d\mathbb{P}_{K_{2}}\left(
\boldsymbol{h}\right)  \text{ \ }%
{\displaystyle\prod_{k=1}^{i}}
d^{N}r_{k}%
{\displaystyle\prod_{k=1}^{j}}
d^{N}s_{k}.
\end{gather*}
We \ consider that $\epsilon>0$ sufficiently small, the case $\epsilon<0$ is
treated in similar way. By using the mean value theorem,%
\begin{equation}
\frac{e^{\epsilon\left\langle \boldsymbol{v},\phi\right\rangle }-1}{\epsilon
}=\left\langle \boldsymbol{v},\phi\right\rangle e^{\epsilon_{0}\left\langle
\boldsymbol{v},\phi\right\rangle }\text{, for some }\epsilon_{0}\in\left(
0,\epsilon\right)  . \label{estimation}%
\end{equation}
Then%
\begin{multline*}
\left\vert \frac{Z_{i,j}^{\text{spin}}(J_{1}+\epsilon\phi,J_{2}%
;\boldsymbol{\theta})-Z_{i,j}^{\text{spin}}(J_{1},J_{2};\boldsymbol{\theta}%
)}{\epsilon}\right\vert \leq\\
\frac{\mathcal{Z}_{\text{spin}}^{\left(  \infty\right)  }}{\mathcal{Z}%
^{\left(  \infty\right)  }}%
{\displaystyle\iint\limits_{\boldsymbol{B}_{M}^{\left(  \infty\right)  }%
\times\boldsymbol{B}_{M}^{\left(  \infty\right)  }}}
\left\vert \left\langle \boldsymbol{v},\phi\right\rangle \right\vert
e^{-E^{\text{spin}}(\boldsymbol{v},\boldsymbol{h};\boldsymbol{\theta
})+c\left\langle 1,\boldsymbol{v}^{4}\right\rangle +d\left\langle
1,\boldsymbol{h}^{4}\right\rangle }\times\\
e^{\left\langle \boldsymbol{v},J_{1}+\epsilon_{0}\phi\right\rangle
+\left\langle \boldsymbol{h},J_{2}\right\rangle }d\mathbb{P}_{K_{1}}\left(
\boldsymbol{v}\right)  \otimes d\mathbb{P}_{K_{2}}\left(  \boldsymbol{h}%
\right)  <\infty,
\end{multline*}
cf. Lemma \ref{Lemma5}. And, by applying the dominated convergence theorem, we
can interchange $\lim_{\epsilon\rightarrow0}$ and $%
{\displaystyle\sum\limits_{i=1}^{\infty}}
$ $%
{\displaystyle\sum\limits_{j=1}^{\infty}}
$ in (\ref{limit}):%
\begin{align*}
&  \lim_{\epsilon\rightarrow0}%
{\displaystyle\sum\limits_{i=1}^{\infty}}
\text{ }%
{\displaystyle\sum\limits_{j=1}^{\infty}}
\frac{\left(  -1\right)  ^{i+j}c^{i}d^{j}}{i!j!}\left[  \frac{Z_{i,j}%
^{\text{spin}}(J_{1}+\epsilon\phi,J_{2};\boldsymbol{\theta})-Z_{i,j}%
^{\text{spin}}(J_{1},J_{2};\boldsymbol{\theta})}{\epsilon}\right] \\
&  =%
{\displaystyle\sum\limits_{i=1}^{\infty}}
\text{ }%
{\displaystyle\sum\limits_{j=1}^{\infty}}
\frac{\left(  -1\right)  ^{i+j}c^{i}d^{j}}{i!j!}\lim_{\epsilon\rightarrow
0}\left[  \frac{Z_{i,j}^{\text{spin}}(J_{1}+\epsilon\phi,J_{2}%
;\boldsymbol{\theta})-Z_{i,nj}^{\text{spin}}(J_{1},J_{2};\boldsymbol{\theta}%
)}{\epsilon}\right]  .
\end{align*}
We now\ compute the limit $\epsilon\rightarrow0^{+}$, the other limit is
treated in a similar way. Now, since $J_{1},J_{2},\boldsymbol{v}%
,\boldsymbol{h}\in\boldsymbol{B}_{M}^{\left(  \infty\right)  }$, by using
Lemma \ref{Lemma5}, (\ref{estimation}), and the Cauchy-Schwarz inequality,%
\begin{gather*}
\frac{\mathcal{Z}_{\text{spin}}^{\left(  \infty\right)  }}{\mathcal{Z}%
^{\left(  \infty\right)  }}\left\vert e^{-E^{\text{spin}}(\boldsymbol{v}%
,\boldsymbol{h};\boldsymbol{\theta})+\left\langle \boldsymbol{v}%
,J_{1}\right\rangle +\left\langle \boldsymbol{h},J_{2}\right\rangle
}\left\langle \boldsymbol{v},\phi\right\rangle e^{\epsilon_{0}\left\langle
\boldsymbol{v},\phi\right\rangle }\left(
{\displaystyle\prod_{k=1}^{i}}
\text{ }\boldsymbol{v}^{4}\left(  r_{k}\right)  \right)  \left(
{\displaystyle\prod_{k=1}^{j}}
\text{ }\boldsymbol{h}^{4}\left(  s_{k}\right)  \right)  \right\vert \leq\\
C(M,\left\Vert \phi\right\Vert _{2})1_{\mathbb{Z}_{p}^{iN}\times\mathbb{Z}%
_{p}^{jN}}\left(  r,s\right)  1_{\boldsymbol{B}_{M}^{\left(  \infty\right)
}\times\boldsymbol{B}_{M}^{\left(  \infty\right)  }}\left(  \boldsymbol{v}%
,\boldsymbol{h}\right)  ,
\end{gather*}
where the last function is integrable with respect to the product measure%
\[
d\mathbb{P}_{K_{1}}\left(  \boldsymbol{v}\right)  \otimes d\mathbb{P}_{K_{2}%
}\left(  \boldsymbol{h}\right)  \text{ \ }%
{\displaystyle\prod_{k=1}^{i}}
d^{N}r_{k}%
{\displaystyle\prod_{k=1}^{j}}
d^{N}s_{k}.
\]
Then, by applying again, the dominated convergence theorem,%
\begin{gather*}
\lim_{\epsilon\rightarrow0}\frac{Z_{i,j}^{\text{spin}}(J_{1}+\epsilon
\phi,J_{2};\boldsymbol{\theta})-Z_{i,j}^{\text{spin}}(J_{1},J_{2}%
;\boldsymbol{\theta})}{\epsilon}=\\
\frac{\mathcal{Z}_{\text{spin}}^{\left(  \infty\right)  }}{\mathcal{Z}%
^{\left(  \infty\right)  }}%
{\displaystyle\int\limits_{\mathbb{Z}_{p}^{iN}}}
\text{ \ }%
{\displaystyle\int\limits_{\mathbb{Z}_{p}^{jN}}}
\text{ \ }%
{\displaystyle\iint\limits_{\boldsymbol{B}_{M}^{\left(  \infty\right)  }%
\times\boldsymbol{B}_{M}^{\left(  \infty\right)  }}}
e^{-E^{\text{spin}}(\boldsymbol{v},\boldsymbol{h};\boldsymbol{\theta
})+\left\langle \boldsymbol{v},J_{1}\right\rangle +\left\langle \boldsymbol{h}%
,J_{2}\right\rangle }\left\langle \boldsymbol{v},\phi\right\rangle \times\\
\left(
{\displaystyle\prod_{i=1}^{m}}
\text{ }\boldsymbol{v}^{4}\left(  r_{i}\right)  \right)  \left(
{\displaystyle\prod_{i=1}^{n}}
\text{ }\boldsymbol{h}^{4}\left(  s_{i}\right)  \right)  d\mathbb{P}_{K_{1}%
}\left(  \boldsymbol{v}\right)  \otimes d\mathbb{P}_{K_{2}}\left(
\boldsymbol{h}\right)  \text{ \ }%
{\displaystyle\prod_{i=1}^{m}}
d^{N}r_{i}%
{\displaystyle\prod_{i=1}^{n}}
d^{N}s_{i}.
\end{gather*}

\end{proof}

By applying Lemma \ref{Lemma8} recursively, one gets the following formula:

\begin{lemma}
\label{Lemma9}%
\begin{gather*}%
{\textstyle\prod\limits_{i=1}^{n}}
\frac{\delta}{\delta_{2}J_{2}\left(  y_{i}\right)  }%
{\textstyle\prod\limits_{i=1}^{m}}
\frac{\delta}{\delta_{1}J_{1}\left(  x_{i}\right)  }\left(  \frac
{\mathcal{Z}_{\text{spin}}^{\left(  \infty\right)  }}{\mathcal{Z}^{\left(
\infty\right)  }}%
{\displaystyle\sum\limits_{i=0}^{\infty}}
\text{ }%
{\displaystyle\sum\limits_{j=10}^{\infty}}
\frac{\left(  -1\right)  ^{i+j}c^{i}d^{j}}{i!j!}Z_{i,j}^{\text{spin}}%
(J_{1},J_{2};\boldsymbol{\theta})\right)  =\\
\frac{\mathcal{Z}_{\text{spin}}^{\left(  \infty\right)  }}{\mathcal{Z}%
^{\left(  \infty\right)  }}%
{\displaystyle\sum\limits_{i=0}^{\infty}}
\text{ }%
{\displaystyle\sum\limits_{j=0}^{\infty}}
\frac{\left(  -1\right)  ^{i+j}c^{i}d^{j}}{i!j!}%
{\textstyle\prod\limits_{j=1}^{n}}
\frac{\delta}{\delta_{2}J_{2}\left(  y_{j}\right)  }\text{ \ }%
{\textstyle\prod\limits_{i=1}^{m}}
\frac{\delta}{\delta_{1}J_{1}\left(  x_{i}\right)  }Z_{i,j}^{\text{spin}%
}(J_{1},J_{2};\boldsymbol{\theta})=\\
\frac{\mathcal{Z}_{\text{spin}}^{\left(  \infty\right)  }}{\mathcal{Z}%
^{\left(  \infty\right)  }}%
{\displaystyle\sum\limits_{i=0}^{\infty}}
\text{ }%
{\displaystyle\sum\limits_{j=0}^{\infty}}
\frac{\left(  -1\right)  ^{i+j}c^{i}d^{j}}{i!j!}\text{\ }\frac{1}%
{\mathcal{Z}_{\text{spin}}^{\left(  \infty\right)  }}\text{ \ }%
{\displaystyle\int\limits_{\mathbb{Z}_{p}^{iN}}}
\text{ \ }%
{\displaystyle\int\limits_{\mathbb{Z}_{p}^{jN}}}
\text{ \ }%
{\displaystyle\iint\limits_{\boldsymbol{B}_{M}^{\left(  \infty\right)  }%
\times\boldsymbol{B}_{M}^{\left(  \infty\right)  }}}
e^{-E^{\text{spin}}(\boldsymbol{v},\boldsymbol{h};\boldsymbol{\theta
})+\left\langle \boldsymbol{v},J_{1}\right\rangle +\left\langle \boldsymbol{h}%
,J_{2}\right\rangle }\times\\%
{\textstyle\prod\limits_{k=1}^{n}}
\boldsymbol{h}\left(  y_{k}\right)
{\displaystyle\prod_{k=1}^{j}}
\text{ }\boldsymbol{h}^{4}\left(  s_{k}\right)
{\textstyle\prod\limits_{k=1}^{m}}
\boldsymbol{v}\left(  x_{k}\right)
{\displaystyle\prod_{k=1}^{i}}
\text{ }\boldsymbol{v}^{4}\left(  r_{i}\right)  d\mathbb{P}_{K_{1}}\left(
\boldsymbol{v}\right)  \otimes d\mathbb{P}_{K_{2}}\left(  \boldsymbol{h}%
\right)  \text{ \ }%
{\displaystyle\prod_{k=1}^{i}}
d^{N}r_{k}%
{\displaystyle\prod_{k=1}^{j}}
d^{N}s_{k},
\end{gather*}
with the convention that for $i=0$, respectively $j=0$, the integral
$\int_{\mathbb{Z}_{p}^{iN}}%
{\displaystyle\prod_{k=1}^{i}}
d^{N}r_{k}$ is omitted, respectively the integral $\int_{\mathbb{Z}_{p}^{jN}}%
{\displaystyle\prod_{k=1}^{j}}
d^{N}s_{k}$.
\end{lemma}

\begin{theorem}
\label{Theorem4}The $(m+n)$-correlation function $\boldsymbol{G}^{\left(
m+n\right)  }\left(  x,y;c,d\right)  $ of the $\left\{  \boldsymbol{v}%
,\boldsymbol{h}\right\}  ^{4}$-SFT, with energy functional
(\ref{energy_functional_ultimo}), admits the following convergent power series
expansion in the coupling parameters $c$, $d\in\mathbb{R}$ :%
\[
\boldsymbol{G}^{\left(  m+n\right)  }\left(  x,y;c,d\right)  =\frac
{\mathcal{Z}_{\text{spin}}^{\left(  \infty\right)  }}{\mathcal{Z}^{\left(
\infty\right)  }}%
{\displaystyle\sum\limits_{i=0}^{\infty}}
\text{ }%
{\displaystyle\sum\limits_{j=0}^{\infty}}
\frac{\left(  -1\right)  ^{i+j}c^{i}d^{j}}{i!j!}\boldsymbol{G}_{i,j}^{\left(
m+n\right)  }\left(  x,y;c,d\right)
\]
in $\mathcal{D}^{\prime}(\mathbb{Z}_{p}^{Nm}\times\mathbb{Z}_{p}^{Nn})$,
where
\begin{multline*}
\boldsymbol{G}_{i,j}^{\left(  m+n\right)  }\left(  x,y;c,d\right)  =%
{\displaystyle\int\limits_{\mathbb{Z}_{p}^{iN}}}
\
{\displaystyle\int\limits_{\mathbb{Z}_{p}^{jN}}}
\boldsymbol{G}^{\left(  m+4i+n+4j\right)  }\left(  x,4\left[  r_{k}\right]
_{1\leq k\leq i},y,\left[  s_{k}\right]  _{1\leq k\leq j};0,0\right)  \times\\%
{\displaystyle\prod_{k=1}^{i}}
d^{N}r_{k}%
{\displaystyle\prod_{k=1}^{j}}
d^{N}s_{k}\in\mathcal{D}^{\prime}(\mathbb{Z}_{p}^{Nm}\times\mathbb{Z}_{p}%
^{Nn}).
\end{multline*}

\end{theorem}

The announced perturbative expansion in the coupling parameters $c$, $d$,
follows from Lemmas (\ref{Lemma8})-(\ref{Lemma9}) by the following calculation:%

\begin{gather*}
\boldsymbol{G}^{\left(  m+n\right)  }\left(  x,y;c,d\right)  =\left[
{\textstyle\prod\limits_{j=1}^{n}}
\frac{\delta}{\delta_{2}J_{2}\left(  y_{j}\right)  }\text{ }%
{\textstyle\prod\limits_{i=1}^{m}}
\frac{\delta}{\delta_{1}J_{1}\left(  x_{i}\right)  }Z(J_{1},J_{2};c,d)\right]
_{\substack{J_{0}=0\\J_{1}=0}}=\\
\left[
{\textstyle\prod\limits_{j=1}^{n}}
\frac{\delta}{\delta_{2}J_{2}\left(  y_{j}\right)  }\text{ }%
{\textstyle\prod\limits_{i=1}^{m}}
\frac{\delta}{\delta_{1}J_{1}\left(  x_{i}\right)  }\frac{\mathcal{Z}%
_{\text{spin}}^{\left(  \infty\right)  }}{\mathcal{Z}^{\left(  \infty\right)
}}%
{\displaystyle\sum\limits_{i=0}^{\infty}}
{\displaystyle\sum\limits_{j=0}^{\infty}}
\frac{\left(  -1\right)  ^{i+j}c^{i}d^{j}}{i!j!}Z_{i,j}^{\text{spin}}%
(J_{1},J_{2};\boldsymbol{\theta})\right]  _{\substack{J_{0}=0\\J_{1}=0}}=
\end{gather*}%
\begin{gather*}
\frac{\mathcal{Z}_{\text{spin}}^{\left(  \infty\right)  }}{\mathcal{Z}%
^{\left(  \infty\right)  }}%
{\displaystyle\sum\limits_{i=0}^{\infty}}
{\displaystyle\sum\limits_{j=0}^{\infty}}
\frac{\left(  -1\right)  ^{i+j}c^{i}d^{j}}{i!j!}\left[
{\textstyle\prod\limits_{j=1}^{n}}
\frac{\delta}{\delta_{2}J_{2}\left(  y_{j}\right)  }\text{ }%
{\textstyle\prod\limits_{i=1}^{m}}
\frac{\delta}{\delta_{1}J_{1}\left(  x_{i}\right)  }Z_{i,j}^{\text{spin}%
}(J_{1},J_{2};\boldsymbol{\theta})\right]  _{\substack{J_{0}=0\\J_{1}=0}}=\\
\frac{\mathcal{Z}_{\text{spin}}^{\left(  \infty\right)  }}{\mathcal{Z}%
^{\left(  \infty\right)  }}%
{\displaystyle\sum\limits_{i=0}^{\infty}}
{\displaystyle\sum\limits_{j=0}^{\infty}}
\frac{\left(  -1\right)  ^{i+j}c^{i}d^{j}}{i!j!}\frac{1}{\mathcal{Z}%
_{\text{spin}}^{\left(  \infty\right)  }}\text{ \ }%
{\displaystyle\int\limits_{\mathbb{Z}_{p}^{iN}}}
\text{ \ }%
{\displaystyle\int\limits_{\mathbb{Z}_{p}^{jN}}}
\text{ \ }%
{\displaystyle\iint\limits_{\boldsymbol{B}_{M}^{\left(  \infty\right)  }%
\times\boldsymbol{B}_{M}^{\left(  \infty\right)  }}}
e^{-E^{\text{spin}}(\boldsymbol{v},\boldsymbol{h};\boldsymbol{\theta})}%
{\textstyle\prod\limits_{k=1}^{m}}
\boldsymbol{v}\left(  x_{k}\right) \\%
{\displaystyle\prod_{k=1}^{i}}
\text{ }\boldsymbol{v}^{4}\left(  r_{k}\right)  \times%
{\textstyle\prod\limits_{k=1}^{n}}
\boldsymbol{h}\left(  y_{k}\right)
{\displaystyle\prod_{k=1}^{j}}
\text{ }\boldsymbol{h}^{4}\left(  s_{k}\right)  d\mathbb{P}_{K_{1}}\left(
\boldsymbol{v}\right)  \otimes d\mathbb{P}_{K_{2}}\left(  \boldsymbol{h}%
\right)
{\displaystyle\prod_{k=1}^{i}}
d^{N}r_{k}%
{\displaystyle\prod_{k=1}^{j}}
d^{N}s_{k}=
\end{gather*}%
\begin{multline*}
\frac{\mathcal{Z}_{\text{spin}}^{\left(  \infty\right)  }}{\mathcal{Z}%
^{\left(  \infty\right)  }}%
{\displaystyle\sum\limits_{i=0}^{\infty}}
{\displaystyle\sum\limits_{j=0}^{\infty}}
\frac{\left(  -1\right)  ^{i+j}c^{i}d^{j}}{i!j!}%
{\displaystyle\int\limits_{\mathbb{Z}_{p}^{iN}}}
\
{\displaystyle\int\limits_{\mathbb{Z}_{p}^{jN}}}
\boldsymbol{G}^{\left(  m+4i+n+4j\right)  }\left(  x,4\left[  r_{k}\right]
_{1\leq k\leq i},y,\left[  s_{k}\right]  _{1\leq k\leq j};0,0\right)  \times\\%
{\displaystyle\prod_{k=1}^{i}}
d^{N}r_{k}%
{\displaystyle\prod_{k=1}^{j}}
d^{N}s_{k}.
\end{multline*}

\end{document}